\documentclass[11pt,letterpaper]{article}
\usepackage[utf8]{inputenc}
\usepackage[T1]{fontenc}
\usepackage{lmodern}
\usepackage[DIV=11]{typearea} 

\usepackage{todonotes}
\usepackage{microtype}

\usepackage{amssymb}
\usepackage{amsmath}
\usepackage{amsthm}
\usepackage[basic]{complexity}
\usepackage{thmtools}
\usepackage[ruled,vlined,linesnumbered,nokwfunc]{algorithm2e_}

\usepackage{tikz}
\usetikzlibrary{arrows,shapes}
\usetikzlibrary{decorations,shadows}
\tikzstyle{player}=[draw, thick, circle, fill=gray!15,inner sep=2pt, minimum width=12pt]
\tikzstyle{vplayer}=[draw, thick, circle, fill=gray!15,inner sep=0.5pt, minimum width=12pt]
\tikzstyle{random}=[draw, thick, diamond, rounded corners, fill=gray!15,inner sep=2pt, minimum width=12pt]
\tikzstyle{vrandom}=[draw, thick, diamond, rounded corners, fill=gray!15,inner sep=0.5pt, minimum width=12pt]

\usepackage{xcolor}
\usepackage{xspace}
\usepackage{paralist}
\usepackage{booktabs}
\usepackage{multirow}

\makeatletter
\newsavebox{\@brx}
\newcommand{\llangle}[1][]{\savebox{\@brx}{\(\m@th{#1\langle}\)}%
  \mathopen{\copy\@brx\mkern2mu\kern-0.9\wd\@brx\usebox{\@brx}}}
\newcommand{\rrangle}[1][]{\savebox{\@brx}{\(\m@th{#1\rangle}\)}%
  \mathclose{\copy\@brx\mkern2mu\kern-0.9\wd\@brx\usebox{\@brx}}}
\makeatother
\newcommand{\as}[1]{\llangle 1 \rrangle_\textit{as}\left(#1\right)}

\newcommand{\set}[1]{\{#1\}}
\newcommand{\lu}{\textup{(}}
\newcommand{\ru}{\textup{)}\xspace}
\newcommand{\upbr}[1]{\lu #1\ru}

\newcommand{\at}{\mathit{Attr}}
\newcommand{\ate}{\mathit{Attr}^+}
\newcommand{\Inf}{\mathrm{Inf}}
\newcommand{\pr}[3]{\mathrm{Pr}^{#1}_{#2}\left(#3\right)}
\newcommand{\reacht}[1]{\textrm{Reach}\left(#1\right)}
\newcommand{\streett}[1]{\textrm{Streett}\left(#1\right)}
\newcommand{\objsty}[2]{\textrm{#1}\left(#2\right)}
\newcommand{\SP}{\mathrm{SP}}
\renewcommand{\RP}{\mathrm{RP}}

\newcommand{\pat}{\omega\xspace}
\newcommand{\Path}{\Omega\xspace}
\newcommand{\str}{\sigma\xspace}
\newcommand{\Str}{\Sigma\xspace}
\newcommand{\obj}{\psi\xspace}
\newcommand{\sseq}{\langle v_0,v_1,v_2,\ldots\rangle}

\newcommand{\mdp}{P\xspace}
\newcommand{\vo}{V_1\xspace}
\newcommand{\vr}{V_R\xspace}
\newcommand{\trans}{\delta\xspace}
\newcommand{\target}{T\xspace}
\newcommand{\intarget}{\expandafter\MakeLowercase\expandafter{\target}\xspace}
\newcommand{\ec}{X\xspace}
\newcommand{\inec}{\expandafter\MakeLowercase\expandafter{\ec}\xspace}
\newcommand{\scc}{C\xspace}
\newcommand{\inscc}{\expandafter\MakeLowercase\expandafter{\scc}\xspace}

\newcommand{\mecalg}{\ProcNameSty{allMECs}}
\newcommand{\allsccalg}{\ProcNameSty{SCCs}}
\newcommand{\sccalg}{\ProcNameSty{SmallestBSCC}}
\newcommand{\reach}{\ProcNameSty{GraphReach}}
\newcommand{\good}{\DataSty{goodEC}}
\newcommand{\winning}{\DataSty{winMEC}}
\newcommand{\badv}{B\xspace}
\newcommand{\rev}{\mathit{RevG}}
\newcommand{\blue}{\mathit{Bl}}

\newcommand{\OutDeg}{\mathit{Outdeg}}
\newcommand{\InDeg}{\mathit{Indeg}}

\newcommand{\remove}{\ProcNameSty{Remove}}
\newcommand{\bad}{\ProcNameSty{Bad}}
\newcommand{\construct}{\ProcNameSty{Construct}}
\newcommand{\bits}{\mathit{bits}\xspace}
\newcommand{\ds}{\mathit{D}\xspace}

\newif\iffullversion
\newcommand{\infull}[1]{\iffullversion #1\fi}
\newcommand{\inshort}[1]{\iffullversion \else #1\fi}

\fullversiontrue

\declaretheorem[numberwithin=section]{theorem}
\declaretheorem[numberlike=theorem]{lemma}
\declaretheorem[numberlike=theorem]{proposition}
\declaretheorem[numberlike=theorem]{corollary}
\declaretheorem[numberlike=theorem]{definition}
\declaretheorem[numberlike=theorem]{observation}
\declaretheorem[numbered=no,name=Observation]{observation*}
\declaretheorem[numberlike=theorem]{reduction}
\declaretheorem[numberlike=theorem]{conjecture}
\declaretheorem[numberlike=theorem]{invariant}
\declaretheorem[numberlike=theorem]{remark}

\DontPrintSemicolon
\SetAlCapFnt{\normalfont\scshape}
\SetProcNameSty{texttt}
\SetFuncSty{texttt}

\usepackage{authblk}

\begin{document}
\title{Model and Objective Separation with Conditional Lower Bounds: Disjunction is Harder than Conjunction}
\author[1]{Krishnendu Chatterjee}
\affil[1]{IST Austria}
\author[2]{Wolfgang Dvo{\v r}{\' a}k}
\author[2]{Monika Henzinger}
\author[2]{Veronika~Loitzenbauer}
\affil[2]{University of Vienna, Faculty of Computer Science}

\maketitle

\begin{abstract}
Given a model of a system and an objective, the model-checking question asks 
whether the model satisfies the objective. We study polynomial-time problems in 
two classical models, graphs and Markov Decision Processes (MDPs), with respect to 
several fundamental $\omega$-regular objectives, e.g., Rabin and Streett objectives. 
For many of these problems the best-known upper bounds are quadratic or cubic, 
yet no super-linear lower bounds are known. 
In this work our contributions are two-fold: First, we present several improved 
algorithms, and second, we present the first conditional super-linear lower bounds 
based on widely believed assumptions about the complexity of CNF-SAT and combinatorial Boolean 
matrix multiplication. 
A separation result for two models with respect to an objective means a conditional 
lower bound for one model that is strictly higher than the existing upper bound 
for the other model, and similarly for two objectives with respect to a model. 
Our results establish the following separation results: 
(1) A separation of models (graphs and MDPs) for disjunctive queries of 
reachability and B\"uchi objectives.
(2) Two kinds of separations of objectives, both for graphs and MDPs, namely, 
(2a) the separation of dual objectives such as \infull{reachability/safety (for 
disjunctive questions) and }Streett/Rabin objectives, and (2b) the separation of
conjunction and disjunction of multiple objectives of the same type such as safety, B\"uchi, and coB\"uchi. 
In summary, our results establish the first model and objective separation 
results for graphs and MDPs for various classical $\omega$-regular objectives. 
Quite strikingly, we establish conditional lower bounds for the disjunction of
objectives that are strictly higher than the existing upper bounds for the 
conjunction of the same objectives.

\end{abstract}

\section{Introduction}
The fundamental problem in formal verification is the \emph{model-checking} 
question that given a model of a system and a property asks whether the model 
satisfies the property. 
The model can be, for example, a standard graph, or a probabilistic extension of 
graphs, and the property describes the desired behaviors\infull{ (or infinite paths)}
of the model. 
For several basic model-checking questions, though polynomial-time 
algorithms are known, the best-known existing upper bounds are quadratic or 
cubic, yet no super-linear lower bounds are known.
In graph algorithmic problems unconditional super-linear lower bounds are very 
rare when polynomial-time solutions exist.
However, recently there have been many interesting results that establish 
\emph{conditional lower bounds}~\cite{AbboudW14,AbboudWY15,AbboudBW15a}.
These are lower bounds based on the assumption that 
for some well-studied problem such as 3-SUM~\cite{GajentaanO12} or All-Pairs 
Shortest Paths~\cite{WilliamsW10,RodittyZ11} no (polynomially\footnote{In particular
improvements by polylogarithmic factors are not excluded.}) faster algorithm 
exists (compared to the best known algorithm).
The lower bounds in this work assume
(A1)~there is no combinatorial\footnote{Combinatorial here means avoiding fast matrix multiplication~\cite{LeGall14}, see also the 
discussion in~\cite{HenzingerKNS15}.} algorithm with running time of
$O(n^{3-\varepsilon})$ for any $\varepsilon > 0$
to multiply two $n \times n$ Boolean matrices;
or (A2)~for all $\varepsilon>0$ there exists a $k$ such that there is no algorithm 
for the $k$-CNF-SAT problem that runs in $2^{(1-\varepsilon) \cdot n} \cdot \operatorname{poly}(m)$ time, where $n$ is the number of variables and $m$ the number of clauses. 
These two assumptions have been used to establish lower bounds for 
several well-studied problems, such as dynamic graph algorithms~\cite{AbboudW14,AbboudWY15}, 
measuring the similarity of strings~\cite{AbboudWW14,Bringmann14,
BringmannK15,BackursI15,AbboudBW15b}, context-free grammar
parsing~\cite{Lee02,AbboudBW15a}, and verifying first-order graph 
properties~\cite{PatrascuW10,Williams14}. 
No relation between conjectures (A1) and (A2) is known. 
In this work we present conditional lower bounds that are super-linear 
for fundamental model-checking problems.

\smallskip\noindent{\em Models.} 
The two most classical models in formal verification are 
\emph{standard graphs} and \emph{Markov decision processes \upbr{MDPs}}. 
MDPs are probabilistic extensions of graphs, 
and an MDP consists of a finite 
directed graph $(V,E)$ with a partition of the vertex set~$V$ into 
player~1 vertices $\vo$ and random vertices $\vr$
and a probabilistic transition function that specifies for vertices in $\vr$
a probability distribution over their successor vertices. 
Let $n = |V|$ and $m = |E|$. 
An infinite path in an MDP is obtained by the following process. 
A token is placed on an initial vertex and the token is moved indefinitely 
as follows: At a vertex $v \in \vo$ a choice is made to move the token along 
one of the outedges of~$v$, and at a vertex $v \in \vr$ the token is moved 
according to the probabilistic transition function. 
Note that if $\vr=\emptyset$, then we have a standard graph, and 
if $\vo=\emptyset$, then we have a Markov chain.
Thus MDPs generalize standard graphs and Markov chains.

\smallskip\noindent{\em Objectives.}
Objectives (or properties) are subsets of infinite paths that specify the 
desired set of paths. \inshort{Let $\target \subseteq V$ be a set of target
vertices.}
The most basic objective is \emph{reachability} where\infull{, given a set 
$\target \subseteq V$ of \emph{target} vertices,} an infinite path satisfies the 
objective if the path visits a vertex of $\target$ {\em at least once}.
The dual objective to reachability is \emph{safety} where\infull{, given a set 
$\target \subseteq V$ of \emph{target} vertices,} an infinite path satisfies the 
objective if the path does \emph{not} visit any vertex of $\target$.
The next extension of a reachability objective is the 
\emph{Büchi} objective that requires the set\infull{ of target vertices}\inshort{~$T$} 
to be reached \emph{infinitely often}. Its dual, the \emph{coBüchi} objective, 
requires the set\infull{ of target vertices}\inshort{~$T$} to be reached only \emph{finitely often}.
A natural extension of single objectives are \emph{conjunctive} and \emph{disjunctive} 
\emph{objectives}~\cite{FijalkowH12,Wolper00,ChatterjeeHP07}. For two objectives
$\obj_1$ and $\obj_2$ their conjunctive objective is equal to $\obj_1 \cap \obj_2$
and their disjunctive objective is equal to $\obj_1 \cup \obj_2$.
The conjunction of reachability (resp.\ Büchi) objectives is known as 
\emph{generalized reachability \upbr{resp.\ Büchi}}~\cite{FijalkowH12,Wolper00}.
A very central and canonical class of objectives in formal verification are 
\emph{Streett} (strong fairness) objectives and their dual \emph{Rabin} objectives~\cite{Thomas97}.
A \emph{one-pair Streett} objective for two sets of vertices $L$ and $U$ specifies
that if the Büchi objective for target set $L$ is satisfied, then also the Büchi
objective for target set $U$ has to be satisfied; in other words, a one-pair 
Streett objective is the disjunction of a coBüchi objective (with target set 
$L$) and a Büchi objective (with target set $U$).
The dual \emph{one-pair Rabin} objective for two vertex sets $L$ and $U$ 
is the conjunction of a Büchi objective with target set $L$ and a
coBüchi objective with target set $U$.
A Streett objective is the conjunction of $k$ one-pair Streett objectives
and its dual Rabin objective is the disjunction of $k$ one-pair Rabin objectives.

\smallskip\noindent{\em Algorithmic questions.}
\infull{The algorithmic question given a model and an objective is as follows: 
(a)~for standard graphs, the model-checking question asks whether there is a 
path that satisfies the objective; and 
(b)~for MDPs, the basic model-checking question asks 
}
\inshort{Given a model and an objective, the algorithmic question 
(a)~for standard graphs is whether there is a path that satisfies
the objective and (b)~for MDPs is}
whether there is a 
\infull{\emph{policy} (or a }\emph{strategy} that resolves the 
non-deterministic choices of outgoing 
edges\infull{)} for player~1 to ensure that the objective is satisfied with 
probability~1.
Observe that if we consider the model-checking question for MDPs with 
$\vr=\emptyset$, then it exactly corresponds to the model-checking question 
for standard graphs.
Given $k$ objectives, the \emph{conjunctive query} question asks whether
there is a \infull{policy}\inshort{strategy} for player~1 to ensure that \emph{all}
the objectives 
are satisfied with probability~1, and the \emph{disjunctive query}
question asks whether there is a \infull{policy}\inshort{strategy} for player~1 to ensure that \emph{one}
of the objectives is satisfied with probability~1.
Conjunctive queries coincide with conjunctive objectives on graphs and MDPs,
while disjunctive queries coincide with disjunctive objectives on graphs but not MDPs\infull{ (see 
Observations~\ref{obs:conj} and~\ref{obs:disjgraph})}.

\smallskip\noindent{\em Significance of model and objectives.} 
Standard graphs are the model for non-deterministic systems, and provide the 
framework to model hardware and software systems~\cite{SPIN,NUSMV}, as well
as many basic logic-related questions such as automata emptiness. 
MDPs model systems with both non-deterministic and probabilistic behavior; 
and provide the framework for a wide range of applications from randomized 
communication and security protocols, to stochastic distributed systems, 
to biological systems~\cite{prism,baierbook}.
In verification, reachability objectives are the most basic objectives for 
safety-critical systems.
In general all properties that arise in verification (such as liveness, 
fairness) are $\omega$-regular languages ($\omega$-regular languages extend 
regular languages to infinite words), and every $\omega$-regular language can 
be expressed as a Streett objective (or a Rabin objective). Important special
cases of Streett (resp.\ Rabin) objectives are Büchi and coBüchi objectives~\cite{ChatterjeeH14}.
Thus the algorithmic questions we consider are the most basic questions in 
formal verification.

\smallskip\noindent{\em Model separation and objective separation questions.} 
In this work our results (upper and conditional lower bounds) aim to establish 
the following two fundamental separations:

\begin{itemize}
\item {\em Model separation.} 
Consider an objective where the algorithmic question for both graphs and MDPs 
can be solved in polynomial time, and establish a conditional lower bound for 
MDPs that is strictly higher than the best-known upper bound for graphs. 
\infull{In other words, the conditional lower bound would separate the model of 
graphs and MDPs for problems (i.e., w.r.t.\ the objective) that can be solved in 
polynomial time.}

\item {\em Objective separation.} 
Consider a model (either graphs or MDPs) with two different objectives and 
show that, though the algorithmic question for both objectives can be solved in 
polynomial time, there is a conditional lower bound for one objective that 
is strictly higher than the best-known upper bound for the other objective.
\infull{In other words, the conditional lower bound would separate the two objectives 
w.r.t.\ the model though they both can be solved in polynomial time.}
\end{itemize}
To the best of our knowledge, there is no previous work that establish any 
model separation or objective separation result in the literature.

\smallskip\noindent{\em Our results.} 
In this work we present improved algorithms as well as 
the first conditional lower bounds that are super-linear for algorithmic 
problems in model checking that can be solved in polynomial time, and together
they establish both model separation and objective separation results. 
An overview of the results for the different objectives is given in Table~\ref{tab:comparison},
where our results are highlighted in boldface.
We use  
\textsc{MEC} to refer to the time to compute the maximal end-component 
decomposition of an MDP. An end-component is a\infull{ (non-trivial)}
strongly connected sub-MDP
that has no outgoing edges for random vertices. We have $\textsc{MEC} = O(\min(n^2, 
m^{1.5}))$~\cite{ChatterjeeH14}\infull{ and assume $\textsc{MEC} = \Omega(m)$ 
and $m\ge n$}.
Moreover, we use 
$k$ to denote the number of combined objectives in the case of conjunction 
or disjunction of multiple objectives and 
$b$ to denote the total number of elements in all the target sets
that specify the objectives.
We first describe Table~\ref{tab:comparison} and our main results and then 
discuss the significance of our results for model and objective separation.

\begin{table*}[!t]
\renewcommand{\arraystretch}{1.3}
\inshort{\nocaptionrule} \caption{Upper and lower bounds.
Our results are boldface and respective results are referred.}\label{tab:comparison}
\centering
\small\scriptsize
\begin{tabular}{@{}*{2}{l}*{4}{c}@{}}
\toprule
 &  & \multicolumn{2}{c}{Graphs} & \multicolumn{2}{c}{MDPs} \\
\cmidrule{3-4}\cmidrule{5-6}
& & upper bound & lower bound$^*$ & upper bound & lower bound$^*$ \\
\midrule
Reach & Conj. & \multicolumn{2}{c}{\NP-c~\cite{ChatterjeeAM13}}\phantom{abcdef} & \multicolumn{2}{c}{\PSPACE-c~\cite{FijalkowH12}} \\
\cmidrule{2-6}
& Disj.\ Obj. & \multicolumn{2}{c}{\multirow{2}{*}{$\Theta(m + b)$\phantom{abcdef}}}  
& $O(\textsc{MEC} + b)$~\cite{ChatterjeeJH03,ChatterjeeH14} \\
& Disj.\ Qu. & \multicolumn{2}{c}{} & $\mathbf{O(k \cdot m + \textsc{\bf MEC})}$ \inshort{[Th.~\ref{thm:reach_alg}]} & 
$\mathbf{k\cdot n^{2-o(1)}}$ \inshort{[Th.~\ref{thm:reach_STChard}]}, $\mathbf{m^{2-o(1)}}$ \inshort{[Th.~\ref{thm:reach_OVChard}]}\\
\midrule
Safety  & Conj. & \multicolumn{2}{c}{$\Theta(m + b)$\phantom{abcdef}} & 
\multicolumn{2}{c}{$\Theta(m + b)$}\\
\cmidrule{2-6}
& Disj.\ Obj. & \multirow{2}{*}{$O(k \cdot m)$} & 
\multirow{2}{*}{$\mathbf{k\cdot n^{2-o(1)}}$ \inshort{[Th.~\ref{thm:safety_STChard}]}}
& \multicolumn{2}{c}{\PSPACE-c~\cite{FijalkowH12}} \\
& Disj.\ Qu. & & & $O(k \cdot m)$ & $\mathbf{k\cdot n^{2-o(1)}}$ \inshort{[Th.~\ref{thm:safety_STChard}]}, $\mathbf{m^{2-o(1)}}$ \inshort{[Th.~\ref{thm:safety_OVChard}]}\\
\midrule
B{\"u}chi  & Conj. & \multicolumn{2}{c}{$\Theta(m + b)$\phantom{abcdef}} & $O(\textsc{MEC} + b)$ \\
\cmidrule{2-6}
& Disj.\ Obj. & \multicolumn{2}{c}{\multirow{2}{*}{$\Theta(m + b)$\phantom{abcdef}}} 
& $O(\textsc{MEC} + b)$~\cite{ChatterjeeJH03,ChatterjeeH14} \\
& Disj.\ Qu. & \multicolumn{2}{c}{} & $\mathbf{O(k \cdot m + \textsc{\bf MEC})}$ \inshort{[Th.~\ref{thm:reach_alg}]} & 
$\mathbf{k\cdot n^{2-o(1)}}$ \inshort{[Cor.~\ref{cor:STC}]}, $\mathbf{m^{2-o(1)}}$ \inshort{[Cor.~\ref{cor:OVC}]}\\
\midrule
coB{\"u}chi & Conj. & \multicolumn{2}{c}{$\Theta(m + b)$\phantom{abcdef}} & $O(\textsc{MEC} + b)$ \\
\cmidrule{2-6}
& Disj.\ Obj. & \multirow{2}{*}{$O(k \cdot m)$} & 
\multirow{2}{*}{$\mathbf{k\cdot n^{2-o(1)}}$ \inshort{[Cor.~\ref{cor:STC}]}}
& $\mathbf{O(k \cdot m + \textsc{\bf MEC})}$ \inshort{[Th.~\ref{thm:cobuchi_alg}]} & $\mathbf{k\cdot n^{2-o(1)}}$ \inshort{[Cor.~\ref{cor:STC}]}, $\mathbf{m^{2-o(1)}}$ \inshort{[Cor.~\ref{cor:OVC}]}\\
& Disj.\ Qu. & & & $\mathbf{O(k \cdot m + \textsc{\bf MEC})}$ \inshort{[Th.~\ref{thm:cobuchi_alg}]} & $\mathbf{k\cdot n^{2-o(1)}}$ \inshort{[Cor.~\ref{cor:STC}]}, $\mathbf{m^{2-o(1)}}$ \inshort{[Cor.~\ref{cor:OVC}]}\\
\cmidrule{2-6}
Singleton & Disj.\ Obj.& \multicolumn{2}{c}{\multirow{2}{*}{$\mathbf{\Theta(m)}$ \inshort{[Th.~\ref{thm:singleton_alg}]}\phantom{abcdef}}} & & $\mathbf{m^{2-o(1)}}$ \inshort{[Cor.~\ref{cor:OVC}]}\\
& Disj.\ Qu. & & & & $\mathbf{m^{2-o(1)}}$ \inshort{[Cor.~\ref{cor:OVC}]}\\
\midrule
Streett & & \multicolumn{2}{l}{$O(\min(n^2, m \sqrt{m \log n}, k m) + b \log n)$
\cite{HenzingerT96,ChatterjeeHL15}}
& \multicolumn{2}{l}{$\mathbf{O(\textbf{min}(n^2, m \sqrt{m\,\textbf{log}\,n}) + b\,\textbf{log}\,n)}$ \inshort{[Th.~\ref{thm:streett_alg}]}}\\
\midrule
Rabin & & $O(k \cdot m)$ & 
$\mathbf{k\cdot n^{2-o(1)}}$ \inshort{[Cor.~\ref{cor:STC}]}
& $O(k \cdot \textsc{MEC})$ & 
$\mathbf{k\cdot n^{2-o(1)}}$ \inshort{[Cor.~\ref{cor:STC}]}, $\mathbf{m^{2-o(1)}}$ \inshort{[Cor.~\ref{cor:OVC}]}\\
\bottomrule
\multicolumn{6}{l}{$^*$ $\mathbf{k\cdot n^{2-o(1)}}$ lower bounds  are based on the \inshort{combinatorial Boolean Matrix Multiplication}\infull{BMM} Conjecture / Strong Triangle Conjecture~(A1)
} \\ 
\multicolumn{6}{l}{\phantom{$^*$} $\mathbf{m^{2-o(1)}}$ lower bounds are based on the Orthogonal Vectors Conjecture
/ Strong ETH~(A2)
}
\end{tabular}
\end{table*}

\begin{enumerate}
\item
\emph{Conjunctive and Disjunctive Reachability \upbr{and Büchi} Problems.} 
    First, we consider conjunctive and disjunctive reachability objectives and 
    queries. Recall that conjunctive objectives and queries in general
    and disjunctive objectives and queries on graphs coincide. For reachability
    further the disjunctive objective can be reduced to a single objective\infull{ (see 
    Observation~\ref{obs:disjobjreach})}.
    The following results are known: the algorithmic question for conjunctive 
    reachability objectives is \NP-complete for 
    graphs~\cite{ChatterjeeAM13}, and \PSPACE-complete for MDPs~\cite{FijalkowH12};
    and the disjunctive objective 
    can be solved in linear time for graphs
    and in $O(\min(n^2, m^{1.5}) + b)$ time in MDPs~\cite{ChatterjeeJH03,ChatterjeeH14}.
    We present three results for disjunctive reachability queries in MDPs: 
    (i)~We present an $O(k m + \textsc{MEC})$-time algorithm\infull{\footnote{This implies an $O(\textsc{MEC} + b)$-time algorithm for disjunctive objective but does not 
    improve the running time for this case.}}.
    (ii)~We show that under assumption (A1) there does not exist a combinatorial
    $O(k \cdot n^{2-\varepsilon})$ algorithm for any $\varepsilon > 0$.
    (iii)~We show that for $k = \Omega(m)$ 
    there does not exist an $O(m^{2-\varepsilon})$
     time algorithm for any $\varepsilon > 0$ under assumption (A2). 
     Hence we establish an upper bound and matching conditional lower bounds
     based on (A1) and (A2).

		Disjunctive Büchi objectives (on graphs and MDPs) 
		can be reduced in linear time to disjunctive
		reachability objectives and vice versa, therefore 
		the same results apply to disjunctive Büchi problems\infull{ (see Observation~\ref{obs:Reachability_Buchi})}.
		The basic algorithm for conjunctive Büchi objectives runs in time $O(m+b)$
		on graphs and in time $O(\textsc{MEC}+b)$ on MDPs.
 
\item
\emph{Conjunctive and Disjunctive Safety Problems.}
    Second, we consider conjunctive and disjunctive safety objectives and queries.
    The following results are known: the conjunctive problem can be reduced
    to a single objective and can be solved in linear 
    time, both in graphs and MDPs (see e.g.~\cite{ChatterjeeDH10}); disjunctive queries for MDPs can be solved in $O(k\cdot m)$ time; and disjunctive objectives for MDPs
    are \PSPACE-complete~\cite{FijalkowH12}.
    We present two results: 
    (i)~We show that for the disjunctive problem in graphs 
    under assumption (A1) there does not exist a combinatorial
    $O(k \cdot n^{2-\varepsilon})$ algorithm for any $\varepsilon > 0$.
    This implies the same conditional lower bound for disjunctive queries
    and objectives in MDPs and matches the upper bound for graphs 
    and disjunctive queries in MDPs.
    (ii)~We present, for $k = \Omega(m)$,
    an $\Omega(m^{2-o(1)})$ lower bound for disjunctive 
    objectives and queries in MDPs under assumption (A2). 
    Again this lower bound matches the upper bound of $O(k \cdot m)$
    for disjunctive queries.

\item
\emph{Conjunctive and Disjunctive coBüchi Problems.}
	For coBüchi, a conjunctive objective can be reduced to a single objective.
	For single objectives the basic algorithm runs in time $O(\textsc{MEC}+b)$ on MDPs
	and in time $O(m+b)$ on graphs.
	\infull{Since the conditional lower bounds for disjunctive safety objectives and 
	queries actually already apply for the non-emptiness of the winning set, 
	the reductions also hold for coBüchi (see 
	Observation~\ref{obs:nonempty_safety_coBuchi}).}
	\inshort{Our conditional lower bounds for disjunctive safety objectives and 
	queries also hold for coBüchi objectives.}
	Here the running times and the conditional lower bounds are matching for both
	disjunctive queries and disjunctive objectives.
	For the conditional lower bound based on assumption (A2)
	only \emph{singleton coBüchi} objectives, i.e., coBüchi objectives with target 
	sets of cardinality one, are needed, 
	therefore the bound already holds for this case. 
	We additionally present two results:
	(i)~We present $O(k m + \textsc{MEC})$-time algorithms for 
    disjunctive queries and objectives in MDPs.
   (ii)~We present a linear time algorithm for 
   disjunctive singleton coBüchi objectives in graphs.
	
\item
\emph{Rabin and Streett objectives.}
    Finally, we consider Rabin and Streett objectives.
    The basic algorithm for Rabin objectives runs in time $O(k \cdot m)$
    on graphs and in time $O(k \cdot \textsc{MEC})$ on MDPs.
    As disjunctive coBüchi objectives are a special case of Rabin objectives,
    the conditional lower bounds for coBüchi objectives 
    of $\Omega(k \cdot n^{2-o(1)})$ on graphs and 
    additionally $\Omega(m^{2-o(1)})$ on MDPs extend to Rabin objectives. 
    The conditional lower bound for graphs is matching (for combinatorial algorithms).
    Furthermore, we extend the results
    of~\cite{HenzingerT96,ChatterjeeHL15} from graphs 
    to MDPs to show that MDPs with Streett objectives can be solved in 
    $O(\min(m \sqrt{m \log n}, n^2) + b \log n)$ time. 
\end{enumerate}

\begin{table}[!t]
\renewcommand{\arraystretch}{1.3}
\inshort{\nocaptionrule} \caption{Model Separation.}\label{tab:model}
\centering
\small\scriptsize
\begin{tabular}{@{}lll@{}}
\toprule
 & upper bound Graphs & lower bounds MDPs\\
\midrule
Reach/Büchi Disj.\ Qu. & $m + nk$ & $\mathbf{k\cdot n^{2-o(1)},m^{2-o(1)}}$\\
coBüchi \infull{Singleton }\inshort{Singl.\ }Disj.\infull{\ Obj./Qu.}& $\mathbf{m}$
& $\mathbf{m^{2-o(1)}}$\\
\bottomrule
\end{tabular}
\end{table}

\begin{table}[!t]
\renewcommand{\arraystretch}{1.3}
\inshort{\nocaptionrule} \caption{Dual Objective Separation for Graphs.}\label{tab:obj}
\centering
\small\scriptsize
\begin{tabular}{@{}llll@{}}
\toprule
& upper bound & lower bound & \\
\midrule
Reach Disj. & $m + nk$ & $\mathbf{k\cdot n^{2-o(1)}}$ & Safety Disj.\\
B{\"u}chi Disj. & $m + nk$ & $\mathbf{k\cdot n^{2-o(1)}}$ & coB{\"u}chi Disj.\\
\midrule
B{\"u}chi Conj. & $m + nk$ & $\mathbf{k\cdot n^{2-o(1)}}$ & coB{\"u}chi Disj.\\
Streett & $n^2 + nk \log n$
& $\mathbf{k\cdot n^{2-o(1)}}$ & Rabin\\
\bottomrule
\end{tabular}
\end{table}

\begin{table}[!t]
\renewcommand{\arraystretch}{1.3}
\inshort{\nocaptionrule} \caption{Dual Objective Separation for MDPs.}\label{tab:objMDP}
\centering
\small\scriptsize
\inshort{\setlength\tabcolsep{4pt}}
\begin{tabular}{@{}llll@{}}
\toprule
& upper bound & lower bound & \\
\midrule
B{\"u}chi Disj.\ \infull{Obj.}\inshort{O.} & $\min(n^2, m^{1.5}) + nk$ & 
$\mathbf{k\cdot n^{2-o(1)},m^{2-o(1)}}$ & \infull{coB{\"u}chi }\inshort{coB.\ }Disj.\ \infull{Obj.}\inshort{O.} \\
\midrule
B{\"u}chi Conj. & $\min(n^2, m^{1.5}) + nk$ & 
$\mathbf{k\cdot n^{2-o(1)},m^{2-o(1)}}$ & \infull{coB{\"u}chi }\inshort{coB.\ }Disj.\ \infull{Obj.}\inshort{O.} \\ 
Streett & $\mathbf{\inshort{\widetilde{O}(}\textbf{min}(n^2, \inshort{m^{1.5}}\infull{m \sqrt{m\,\textbf{log}\,n}}) \infull{ + nk\,\textbf{log}\,n}\inshort{)}}$ & $\mathbf{k\cdot n^{2-o(1)},m^{2-o(1)}}$ & Rabin\\
\bottomrule
\end{tabular}
\end{table}

\smallskip\noindent{\em Significance of our results.} 
We now describe the model and objective separation results that are obtained 
from the results we established.
\begin{enumerate}
\item
\emph{Model Separation.} 
    Table~\ref{tab:model} shows our results that 
    separate graphs and MDPs regarding their complexity for certain 
    objectives and queries under assumptions (A1) and (A2).
    First, 
    for reachability and Büchi objectives disjunction in graphs is in linear time 
    while in MDPs we have $\Omega(kn^{2-o(1)})$ and $\Omega(m^{2-o(1)})$ conditional 
    lower bounds for disjunctive queries. 
    Second, \infull{for coBüchi we have a separation when restricted to the 
    class where each target set is a singleton. 
    For these objectives disjunction in graphs}\inshort{in graphs 
    the disjunction of coBüchi objectives where each target set is a singleton} 
    is in linear time while 
    we establish an $\Omega(m^{2-o(1)})$ conditional lower bound for MDPs for both 
    disjunctive objectives and queries.
    
\item
\emph{Objective Separation.}
    \infull{Further we identify complexity separations between different objectives. 
    Here w}\inshort{W}e consider two aspects, separations between dual objectives like Büchi and 
    coBüchi (Tables~\ref{tab:obj} and~\ref{tab:objMDP}), and separations between 
    conjunction and disjunction of objectives (Table~\ref{tab:condis}).
    We compare dual objectives in two ways: (i) we show that single objectives
    that are dual to each other behave differently when we consider 
    disjunction for each of them and (ii) we compare conjunctive objectives
    and their dual disjunctive objectives. For (ii) we have that 
    conjunctive Büchi objectives are dual to disjunctive coBüchi objectives,
    and Streett objectives\infull{, the conjunction of 1-pair Streett objectives,} are 
    dual to Rabin objectives\infull{, the disjunction of 1-pair Rabin objectives}.
    
    \begin{enumerate}
\item
\emph{Separating Dual Objectives in Graphs.}
		\infull{First, we consider reachability and safety objectives. }In graphs we have 
		that for reachability objectives disjunction
		is in linear time while for disjunctive safety objectives we establish an 
		$\Omega(k n^{2-o(1)})$ 
		lower bound under assumption (A1). \inshort{Analogous results hold for Büchi and coBüchi objectives.}\infull{Analogously, we have 
		disjunctive Büchi objectives are in linear time on graphs
		while we establish an $\Omega(k n^{2-o(1)})$ conditional lower bound for disjunction of coBüchi objectives.}
		Further, conjunctive Büchi objectives are in linear time and thus can 
		be separated from their dual objective, the disjunctive coBüchi objectives.
		Finally, for Streett objectives in graphs with $b = O(n^2 / \log n)$ 
		we have an $O(n^2)$ algorithm while we establish an 
		$\Omega(n^{3-o(1)})$ lower bound for Rabin objectives when $k = \Theta(n)$.
		
\item
\emph{Separating Dual Objectives in MDPs.}	
		\infull{First, consider Büchi and coBüchi objectives in MDPs.}
		On MDPs disjunctive Büchi objectives are in time $O(\textsc{MEC} + b)$, which is in 
		$O(\min(n^2, m^{1.5}) + n k)$, while for coBüchi objectives
		we show $\Omega(k n^{2-o(1)})$ and $\Omega(m^{2-o(1)})$ conditional lower bounds for 
		both disjunctive queries and disjunctive objectives. This separates the 
		two objectives for both sparse and dense graphs.
		Further conjunctive Büchi objectives can be solved in $O(\textsc{MEC} + b)$
		time and thus there is also a separation between disjunctive coBüchi
		objectives and their dual.
		Finally, for Streett objectives in MDPs with $b = O(\min(n^2, m^{1.5}) / \log n)$ 
		we show both an $O(n^2)$-time and an $O(m^{1.5})$-time algorithm while 
		we establish $\Omega(n^{3-o(1)})$ and 
		$\Omega(m^{2-o(1)})$ conditional lower bounds for Rabin objectives
		when $k = \Theta(n)$.
		
\item
\emph{Separating Conjunction and Disjunction in Graphs and MDPs.}
		Except for reachability, i.e., in particular for all considered 
		polynomial-time problems, we observe that the disjunction of objectives is
		computationally harder than the conjunction of these objectives (under assumptions (A1), (A2)).
		First, for safety objectives conjunction is in linear time even for MDPs
		while for disjunctive queries (disjunctive objectives are $\PSPACE$-complete)
		we present $\Omega(kn^{2-o(1)})$ and $\Omega(m^{2-o(1)})$ conditional lower bounds, where the 
		first bound also holds for graphs.
		Second, for Büchi and coBüchi objectives conjunction is in $O(\textsc{MEC} + b)$ on MDPs
		(and $O(m+b)$ on graphs) while 
		we show $\Omega(kn^{2-o(1)})$ and $\Omega(m^{2-o(1)})$ conditional lower bounds for 
		disjunctive coBüchi objectives and disjunctive Büchi / coBüchi queries on MDPs.
		\infull{The $\Omega(m^{2-o(1)})$ bound even holds for the disjunction of singleton coBüchi
		objectives. }Further, for coBüchi objectives our $\Omega(kn^{2-o(1)})$ bound also holds on graphs,
		which separates conjunction and disjunction also in this setting.
		Third, \infull{we can also see the results for Streett and Rabin objectives 
		as a separation between conjunction and disjunction. Recall that Streett 
		objectives are the conjunction of one-pair Streett objectives and 
		Rabin objectives are the disjunction of one-pair Rabin objectives.
		Further, both Büchi and coBüchi objectives are special cases of each of 
		one-pair Streett and one-pair Rabin objectives. In particular the following
		separations are easy observations or corollaries of our results:
		For the disjunction of one-pair Streett objectives the same conditional 
		lower bounds (and the same upper bound, see Observation~\ref{obs:disjStreett}) as 
		for the disjunction of coBüchi objectives apply.
		Thus the disjunction of one-pair Streett objectives is harder than the 
		conjunction of one-pair Streett objectives (under assumptions (A1)/(A2)).
		The conjunction of one-pair Rabin objectives can be solved in the same time 
		as conjunctive Büchi objectives. Thus also the disjunction of one-pair 
		Rabin objectives is harder than their conjunction.}
		\inshort{corollaries of our results are that
		for each of one-pair Streett objectives and one-pair Rabin objectives their
		disjunction is harder than their conjunction.}
	\end{enumerate}
\end{enumerate}

\begin{table}[!t]
\renewcommand{\arraystretch}{1.3}
\inshort{\nocaptionrule} \caption{Separating Conjunction and Disjunction.}\label{tab:condis}
\centering
\small\scriptsize
\inshort{\setlength\tabcolsep{5pt}}
\begin{tabular}{@{}lllll@{}}
\toprule
& & Conjunction & Disjunction \\
\midrule 
Safety & Graphs & $m + nk$ & $\mathbf{k\cdot n^{2-o(1)}}$ \\
& MDP Qu. & $m + nk$ & $\mathbf{k\cdot n^{2-o(1)},m^{2-o(1)}}$\\
\midrule
B{\"u}chi & MDPs Qu. & $\min(n^2, m^{1.5}) + nk$ & 
$\mathbf{k\cdot n^{2-o(1)},m^{2-o(1)}}$\\
\midrule
coB{\"u}chi & Graphs & $m + nk$ & $\mathbf{k\cdot n^{2-o(1)}}$ \\
& MDPs\infull{ Obj./Qu.} & $\min(n^2, m^{1.5}) + nk$ & $\mathbf{k\cdot n^{2-o(1)},m^{2-o(1)}}$\\
\midrule
1-pair Streett & Graphs & $n^2 + nk \log n$
& $\mathbf{k\cdot n^{2-o(1)}}$ \\
& MDPs\infull{ Obj./Qu.} & $\mathbf{\inshort{\widetilde{O}(}\textbf{min}(n^2, \inshort{m^{1.5}}\infull{m \sqrt{m\,\textbf{log}\,n}}) \infull{ + nk\,\textbf{log}\,n}\inshort{)}}$ & $\mathbf{k\cdot n^{2-o(1)},m^{2-o(1)}}$\\
\midrule
1-pair Rabin & Graphs & $m + nk$ & $\mathbf{k\cdot n^{2-o(1)}}$ \\
& MDPs\infull{ Obj./Qu.} & $\min(n^2, m^{1.5}) + nk$ & $\mathbf{k\cdot n^{2-o(1)},m^{2-o(1)}}$ \\
\bottomrule
\end{tabular}
\end{table}

\smallskip\noindent{\em Remark about Streett and Rabin objective separation.} 
One remarkable aspect of our objective separation result is that we achieve it
for Rabin and Streett objectives (both in graphs and MDPs), which are dual. 
In more general models such as games on graphs, Rabin objectives are 
\NP-complete and Streett objectives are \coNP-complete~\cite{EmersonJ99}.
In graphs and MDPs, both Rabin and Streett objectives can be solved in 
polynomial time.  
Since Rabin and Streett objectives are dual, and they belong to the 
complementary complexity classes (either both in P, or one is \NP-complete, 
other \coNP-complete), they were considered to be equivalent for algorithmic 
purposes for graphs and MDPs. 
Quite surprisingly we show that under some widely believed assumptions, both 
for MDPs and graphs, Rabin objectives are algorithmically harder than Streett objectives.

\smallskip\noindent{\em Technical contributions.}

\emph{Algorithms.}
	(1)~We show that given the MEC-decomposition of an MDP, the \emph{almost-sure reachability problem}
	can be solved in linear time on the MDP where each MEC is contracted to 
	a player~1 vertex. This yields to the improved algorithms 
	for disjunctive queries of reachability and Büchi objectives on MDPs.
	(2)~For \emph{MDPs} with \emph{disjunctive coBüchi} objectives and disjunctive queries of 
	coBüchi objectives we use the MEC-decomposition in a different way;
	namely, we show that it is sufficient to do a linear-time computation in
	each MEC per coBüchi objective to solve both disjunctive questions.
	(3)~Further we show that for \emph{graphs} with a 
	\emph{disjunctive coBüchi} objective for which the target set of each of the
	single coBüchi objectives has \emph{cardinality one}
	the problem can be solved with a breadth-first search like algorithm in linear time.
	(4)~Finally, we provide faster algorithms for \emph{MDPs with Streett objectives}.
	The straight-forward algorithm repeatedly
	computes MEC-decompositions in a black-box manner; we show that one can 
	open this black-box and combine the current
	best algorithms for MEC-decomposition~\cite{ChatterjeeH14} and graphs with 
	Streett objectives~\cite{HenzingerT96,ChatterjeeHL15}
	to achieve almost the same running time for MDPs with Streett objectives
	as for graphs.

\emph{Conditional Lower Bounds.}
(a)~Conjecture (A1)
is equivalent to the conjecture that there is no combinatorial $O(n^{3-\varepsilon})$
time algorithm to detect whether an $n$-vertex graph contains a triangle~\cite{WilliamsW10}.
	We show that triangle-detection in graphs can be linear-time reduced 
	to \emph{disjunctive queries of almost-sure reachability} in MDPs and thus 
	that the latter is hard assuming (A1).
(b)~For the hardness under (A2) we consider the intermediate problem 
	Orthogonal Vectors, which is known to be hard under (A2)~\cite{Williams05},
	and linear-time reduce it to \emph{disjunctive queries of almost-sure 
	reachability} in MDPs.
(c)~For \emph{disjunctive safety problems} we give a linear-time reduction from  
	triangle-detection that only requires player~1 vertices and thus
	hardness also holds in graphs when assuming (A1).
(d)~However, the reduction we give from Orthogonal Vectors
	to \emph{disjunctive safety problems} requires random vertices and thus hardness
	under (A2) only holds on MDPs.
(e)~\infull{Based on the hardness results for {almost-sure reachability} and
	safety, w}\inshort{W}e then exploit reductions between the different types of
	objectives to obtain the hardness results for Büchi, coBüchi,
	and Rabin.

\smallskip\noindent{\em Outline.}
In Section~\ref{sec:prelim} we provide formal definitions, describe the connections
between different objectives, and state the 
conjectures on which the conditional lower bounds are based.
Section~\ref{sec:reach} is about disjunctive reachability queries on MDPs; we 
first present the improved algorithm and then the conditional lower bounds.
In Section~\ref{sec:safety} we describe the conditional lower bounds for disjunctive
safety problems on graphs and MDPs. In Section~\ref{sec:streett} we provide
the improved algorithms for MDPs with Streett objectives. In Section~\ref{sec:rabin}
we show how the conditional lower bounds extend from reachability and safety
to Büchi, coBüchi, and Rabin and present algorithms for MDPs with Rabin objectives and 
for MDPs with disjunctive objectives and queries of Büchi and coBüchi objectives.
In Section~\ref{sec:singleton} we describe the linear time algorithm for disjunctive
coBüchi objectives on graphs for the special case when all target sets are singletons.
We conclude in Section~\ref{sec:conclusion}.
\section{Preliminaries}\label{sec:prelim}
\smallskip\noindent{\em Markov Decision Processes \upbr{MDPs} and Graphs.} 
An MDP~$\mdp = ((V, E), \allowbreak(\vo, \vr), \allowbreak\trans)$ consists of a finite directed 
graph with vertices $V$ and edges $E$ with a partition of the vertices into 
\emph{player~1 vertices} $\vo$ and \emph{random vertices} $\vr$ and a 
probabilistic transition function~$\trans$. We call an edge $(u,v)$ with $u \in \vo$
\emph{player~1 edge} and an edge $(v, w)$ with $v \in \vr$ a \emph{random edge}.
The probabilistic transition function is a function from $\vr$ to $\mathcal{D}(V)$, 
where $\mathcal{D}(V)$ is the set of probability distributions over $V$ and 
a random edge $(v, w) \in E$ if and only if $\trans(v)[w] > 0$.\infull{ For the purpose
of this paper we assume for simplicity that, for each random vertex $v$, 
$\trans(v)[w]$ is the uniform distribution over all $w \in V$ with $(v, w) \in E$; this is w.l.o.g.\ as we are only ask whether a probability is zero or one 
(qualitative analysis) or zero or larger than zero.} Graphs are a special case
of MDPs with $\vr = \emptyset$.
\infull{

\smallskip\noindent{\em Sub-MDPs and Maximal End-Components.}
A sub-MDP of an MDP $\mdp$ induced by a vertex set $\ec\infull{ \subseteq V}$ 
is defined as $\mdp[\ec] = ((\ec, E \cap (\ec \times \ec), (\vo \cap \ec, 
\vr \cap \ec), \trans')$, where $\trans': \ec \rightarrow \mathcal{D}(\ec)$ 
is for each $v \in \vr \cap \ec$ the uniform distribution over all $w \in \ec$ 
with $(v, w) \in E$.
An \emph{end-component} \upbr{EC} of an MDP $\mdp$ is a set of 
vertices $\ec\infull{ \subseteq V}$ such that \upbr{a} the induced sub-MDP 
$\mdp[\ec]$
is strongly connected, \upbr{b} all outgoing edges in $E$ of vertices in 
$\ec \cap \vr$ are contained in $\mdp[\ec]$, and \upbr{c} $\mdp[\ec]$ contains at least one edge.
An end-component is a \emph{maximal end-component} \upbr{MEC} if it is maximal 
under set inclusion.
An end-component is \emph{trivial} if it consists of a single vertex \upbr{with
a self-loop}, otherwise it is \emph{non-trivial}.
The \emph{MEC-decomposition} of an MDP consists of all MECs of the MDP and the 
set of vertices that do not belong to any MEC.
}

\smallskip\noindent{\em Plays and Strategies.} 
A \emph{play} or infinite path in $\mdp$ is an infinite sequence $\pat = \langle v_0, 
v_1, v_2, \ldots \rangle$ such that $(v_i, v_{i+1}) \in E$ for all $i \in \mathbb{N}$;
we denote by $\Path$ the set of all plays.
A player~1 \emph{strategy}~$\str: V^* \cdot \vo \rightarrow V$ is a function that 
assigns to every finite prefix~$\pat \in V^* \cdot \vo$ of a play that ends in a 
player~1 vertex~$v$ a successor vertex $\str(\pat) \in V$ such that 
there exists an edge $(v, \str(\pat)) \in E$; we denote by $\Str$ the 
set of all player~1 strategies. A strategy is \emph{memoryless} if we have 
$\str(\pat) = \str(\pat')$ for any $\pat, \pat' \in V^* \cdot \vo$ that 
end in the same vertex $v \in \vo$.

\smallskip\noindent{\em Objectives and Almost-Sure Winning Sets.} 
An \emph{objective} $\obj$
is a subset of $\Path$ said to be winning for player~1. We say that a 
play $\pat \in \Path$ \emph{satisfies the objective} if $\pat \in \obj$. 
For any measurable set of plays $A \subseteq \Path$
we denote by $\pr{\str}{v}{A}$ the probability that a play starting at $v \in V$
belongs to $A$ when player~1 plays strategy~$\str$. 
A strategy~$\str$ is \emph{almost-sure\infull{ \upbr{a.s.}} winning} from a vertex $v \in V$
for an objective $\obj$ if $\pr{\str}{v}{\obj} = 1$. In graphs the existence 
of an almost-sure winning strategy corresponds to the existence of a play in the objective. The \emph{almost-sure 
winning set} $\as{\mdp, \obj}$
of player~1 is the set of vertices for which player~1 has an
almost-sure winning strategy. \infull{Computing the almost-sure winning set for some 
objective is also called \emph{qualitative analysis} of MDPs. Below we define 
the objectives used in this work. }Let $\Inf(\pat)$ for $\pat \in \Path$ denote
the set of vertices that occurs infinitely often in $\pat$.
\begin{description}
    \item[Reachability] For a vertex 
    set~$\target \subseteq V$ the reachability 
	objective is the set of infinite paths that contain a vertex of $\target$, i.e., 
	$\reacht{\target} = \set{\sseq \in \Path \mid \exists j \ge 0: v_j \in \target}$.
	
	\item[Safety] For a vertex set~$\target \subseteq V$ the safety 
	objective is the set of infinite paths that \emph{do not contain} any vertex 
	of $\target$, i.e., 
	$\objsty{Safety}{\target} = \set{\sseq \in \Path \mid \forall j \ge 0: v_j \notin \target}$.
	
    \item[B{\"u}chi] For a vertex set~$\target \subseteq V$ the B{\"u}chi 
	objective is the set of infinite paths in which a vertex of $\target$ occurs
	\emph{infinitely often}, i.e., 
	$\objsty{B{\"u}chi}{\target} = \set{\pat \in \Path \mid \Inf(\pat) \cap \target \ne \emptyset}$.
	
    \item[coB{\"u}chi] For a 
    vertex set~$\target \subseteq V$ the coB{\"u}chi 
	objective is the set of infinite paths for which \emph{no} vertex of $\target$ 
	occurs \emph{infinitely often}, i.e., 
	$\objsty{coB{\"u}chi}{\target} = \set{\pat \in \Path \mid \Inf(\pat) \cap \target = \emptyset}$.
	
    \item[Streett] Given a set $\SP$ of
	$k$ pairs $(L_i, U_i)$ of vertex sets $L_i, U_i
	\subseteq V$ with $1 \le i \le k$, the Streett objective is the set of infinite paths 
	for which it holds \emph{for each} $1 \le i \le k$ that whenever a vertex of $L_i$ 
	occurs infinitely often, then a vertex of $U_i$ occurs infinitely often, i.e., 
	$\streett{\SP} = \set{\pat \in \Path \mid
	L_i \cap \Inf(\pat) = \emptyset \text{ or } U_i \cap \Inf(\pat) \ne \emptyset 
	\text{ for all } 1 \le i \le k}$.
	
    \item[Rabin] Given a set $\RP$ of $k$ pairs $(L_i, U_i)$ of vertex sets $L_i, U_i 
	\subseteq V$ with $1 \le i \le k$, the Rabin objective is the set of infinite paths 
	for which there \emph{exists} an $i$, $1 \le i \le k$, such that a vertex of $L_i$ 
	occurs infinitely often but no vertex of $U_i$ occurs infinitely often, i.e., 
	$\objsty{Rabin}{\RP} = \set{\pat \in \Path \mid 
	L_i \cap \Inf(\pat) \ne \emptyset \text{ and } U_i \cap \Inf(\pat) = \emptyset 
	\text{ for some } 1 \le i \le k}$.
\end{description}

Given $c$ objectives $\obj_1, \ldots, \obj_c$, the \emph{conjunctive objective} 
$\obj = \obj_1 \cap \ldots \cap \obj_c$ is given by the intersection of the $c$
objectives, and the \emph{disjunctive objective} $\obj = 
\obj_1 \cup \ldots \cup \obj_c = \bigvee_{i = 1}^c \obj_i$ 
is given by the union of the $c$
objectives. For the \emph{conjunctive query} of $c$ objectives 
$\obj_1, \ldots, \obj_c$ we define the (almost-sure) winning set to be 
the set of vertices that have one strategy that is (almost-sure) winning for \emph{each of}
the objectives $\obj_1, \ldots, \obj_c$.
Analogously, a vertex is in 
the (almost-sure) winning set $\bigvee_{i=1}^c \as{\mdp, \obj_i}$ for the \emph{disjunctive query} of the $c$ 
objectives if it is in a (almost-sure) winning set for \emph{at least 
one} of the $c$ objectives (i.e. we take the union of the winning sets).

\infull{
Below we present several observations that interlink different types of objectives.

\begin{observation}\label{obs:conj}
	The almost-sure winning set for a \emph{conjunctive objective} is the same as for 
	the corresponding \emph{conjunctive query}.
\end{observation}
\begin{proof}
	We have for any $v \in V$ and $\str \in \Str$ and any two objectives $\obj_1$, $\obj_2$ that $\pr{\str}{v}{\obj_1 \land \obj_2} = 1$ 
	iff $\pr{\str}{v}{\obj_1} = 1$ and $\pr{\str}{v}{\obj_2} = 1$.
\end{proof}

\begin{observation}\label{obs:disjgraph}
	On \emph{graphs} \upbr{i.e. $\vr = \emptyset$}
	the winning set for a \emph{disjunctive objective} is the same as for the 
	corresponding \emph{disjunctive query}.
\end{observation}
\begin{proof}
	For any two objectives $\obj_1$, $\obj_2$ we have for each $\pat \in \Path$ 
	that $\pat \in (\obj_1 \cup \obj_2)$ iff $\pat \in \obj_1$ or $\pat \in \obj_2$.
\end{proof}

\begin{observation}\label{obs:disjobjreach}
	The \emph{disjunctive objective} of \emph{B{\"u}chi} \upbr{resp.\ reachability} 
	objectives is the same as the B{\"u}chi \upbr{resp.\ reachability} objective of the 
	union of the target sets.
\end{observation}
\begin{proof}
We show the claim for Büchi, the proof for reachability is analogous. For two 
target sets $T_1, T_2 \subseteq V$ we have 
$\set{\pat \in \Path \mid \Inf(\pat) \cap \target_1 \ne \emptyset} \cup 
\set{\pat \in \Path \mid \Inf(\pat) \cap \target_2 \ne \emptyset}=
\set{\pat \in \Path \mid \Inf(\pat) \cap (\target_1 \cup \target_2) \ne \emptyset}$.	
\end{proof}

\begin{observation}\label{obs:conjsafety}
	The \emph{conjunctive objective} of \emph{coB{\"u}chi} \upbr{resp.\ safety} 
	objectives is the same as the coB{\"u}chi \upbr{resp.\ safety} objective of the 
	union of the target sets.
\end{observation}
\begin{proof}
We show the claim for coBüchi, the proof for safety is analogous. For two 
target sets $T_1, T_2 \subseteq V$ we have 
$\set{\pat \in \Path \mid \Inf(\pat) \cap \target_1 = \emptyset} \cap 
\set{\pat \in \Path \mid \Inf(\pat) \cap \target_2 = \emptyset}=
\set{\pat \in \Path \mid \Inf(\pat) \cap (\target_1 \cup \target_2) = \emptyset}$.	
\end{proof}

By definition each path winning for a safety objective 
is also winning
for the corresponding coBüchi objective while the converse is not always true.
However, when it comes to the non-emptiness of winning sets these two objectives become equivalent.

\begin{observation}\label{obs:nonempty_safety_coBuchi}
	For a fixed MDP $\mdp$ the winning set for $\objsty{Safety}{\target}$ is non-empty
	iff the winning set for $\objsty{coB{\"u}chi}{\target}$ is non-empty.
	This equivalence extends also to conjunctions and disjunctions of safety and coBüchi objectives.
\end{observation}
\begin{proof}
	By \cite[p.~891]{CourcoubetisY95} (see also Section~\ref{sec:gec})
	the winning set for $\objsty{Safety}{\target}$ resp.\ $\objsty{coB{\"u}chi}{\target}$
	is non-empty if and only if there exists an end-component $\ec$ with $\ec \cap \target = \emptyset$.
\end{proof}

\begin{observation}\label{obs:Reachability_Buchi}
	Disjunctive \upbr{Obj./Qu.} Reachability in MDPs can be linear time reduced to 
	disjunctive \upbr{Obj./Qu.} Büchi-Objectives in MDPs and vice versa.
\end{observation}
\begin{proof}
	Reachability $\Rightarrow$ Büchi: For each target set $\target$ 
	replace each $t \in \target$ with two vertices: $t_\text{in} \in \vo$ and 
	$t_\text{out}$, where $t_\text{out}$ belongs to the same player as $t$.
	Assign all incoming edges of $t$ to $t_\text{in}$
	and all outgoing edges of $t$ to $t_\text{out}$, and add the edge $(t_\text{in},
	t_\text{out})$ and the self-loop $(t_\text{in}, t_\text{in})$. Let the 
	corresponding target set for Büchi be the union of $t_\text{in}$ for all 
	$t \in \target$. A vertex $t_\text{in}$ in the modified MDP can be visited 
	infinitely often  almost surely iff in the original MDP the vertex $t$ can be
	reached almost surely.
	
	Büchi $\Rightarrow$ Reachability: For each target set $\target$ 
	replace each $t \in \target$ with three vertices: $t_\text{in} \in \vr$,
	$t_r \in \vo$, and $t_\text{out}$, where $t_\text{out}$ belongs to the 
	same player as $t$. Assign all incoming edges of $t$ to $t_\text{in}$
	and all outgoing edges of $t$ to $t_\text{out}$, and add the edges 
	$(t_\text{in}, t_\text{out})$, $(t_\text{in}, t_r)$, and $(t_r, t_\text{out})$.
	Let the corresponding target set for Reachability be the union of $t_r$ for all 
	$t \in \target$. A vertex $t_r$ in the modified MDP can be reached almost 
	surely iff in the original MDP the vertex $t$ can almost surely
	be visited infinitely often.
\end{proof}

}

\infull{
\begin{observation}\label{obs:ConjBuchiStreet}
	Conjunctive Büchi \upbr{resp.\ coBüchi} objectives are special instances of Streett objectives.
\end{observation}
\begin{proof}
	For Büchi let $L_i = V$ and $U_i = \target_i$, for coBüchi let $L_i = \target_i$
	and $U_i = \emptyset$.
\end{proof}

\begin{observation}\label{obs:DisjBuchiRabin}
	Disjunctive Büchi \upbr{resp.\ coBüchi} objectives are special instances of Rabin objectives.
\end{observation}
\begin{proof}
	For Büchi let $L_i = \target_i$ and $U_i = \emptyset$, for coBüchi let 
	$L_i = V$ and $U_i = \target_i$.
\end{proof}
}

\subsection{Conjectured Lower Bounds}\label{sec:conjectures}

While classical complexity results are based on standard complexity-theoretical assumptions,
e.g., $\P \ne \NP$, polynomial lower bounds are often based on widely believed,
conjectured lower bounds 
about well studied algorithmic problems.
Our lower bounds will be conditioned on the popular conjectures discussed below.

First, we consider conjectures on Boolean matrix multiplication~\cite{WilliamsW10,AbboudW14} and 
triangle detection~\cite{AbboudW14} in graphs, which build the basis for our lower bounds on dense graphs. A triangle in a graph is a triple $x, y, z$ of vertices 
such that $(x, y), (y, z), (z, x) \in E$.

\begin{conjecture}[Combinatorial Boolean Matrix Multiplication Conjecture (BMM)]\label{conj:bmm}
There is no $O(n^{3 - \varepsilon})$ time combinatorial algorithm for computing the boolean product of 
two $n \times n$ matrices for any $\varepsilon > 0$.
\end{conjecture}

\begin{conjecture}[Strong Triangle Conjecture (STC)]\label{conj:triangle}
There is \infull{no $O(\min\{n^{\omega - \varepsilon},$ $m^{2 \omega/(\omega + 1) - \varepsilon}\})$ expected time algorithm and }no
$O(n^{3 - \varepsilon})$ time combinatorial algorithm that can detect whether a graph contains a triangle for any $\varepsilon > 0$\infull{, 
where $\omega < 2.373$ is the matrix multiplication exponent}. 
\end{conjecture}

\inshort{BMM is equivalent to STC~\cite{WilliamsW10}.}
\infull{By a result of Vassilevska~Williams and Ryan Williams~\cite{WilliamsW10}, we have that 
BMM is equivalent to the combinatorial part of STC.} 
\infull{Moreover, if we do not restrict ourselves to combinatorial algorithms, STC still gives a super-linear lower bound.}\inshort{A weaker assumption, without the 
restriction to combinatorial algorithms, is that detecting a triangle in a graph
takes super-linear time.}

Second, we consider the Strong Exponential Time Hypothesis~\cite{ImpagliazzoPZ01,CalabroIP09}
and the Orthogonal Vectors Conjecture~\cite{AbboudWW15}, the former dealing with satisfiability in propositional logic and the latter with 
the \emph{Orthogonal Vectors Problem}.

\emph{The Orthogonal Vectors Problem} (OV). Given two sets $S_1, S_2$ of $d$-bit 
vectors with $|S_i |\leq N$, $d \in \Theta(\log N)$, are there $u \in S_1$ 
and $v \in S_2$ such that $\sum_{i=1}^{d} u_i \cdot  v_i = 0$?

\begin{conjecture}[Strong Exponential Time Hypothesis (SETH)]\label{conj:seth}
For each  $\varepsilon >0$ there is a $k$ such that k-CNF-SAT on $n$ variables and $m$ clauses
cannot be solved in $O(2^{(1-\varepsilon)n} \operatorname{poly}(m))$ time.
\end{conjecture}

\begin{conjecture}[Orthogonal Vectors Conjecture (OVC)]\label{conj:ov}
There is no $O(N^{2-\varepsilon})$ time algorithm for \inshort{OV}\infull{the Orthogonal Vectors Problem} for any $\varepsilon > 0$.
\end{conjecture}

\infull{By a result of Williams~\cite{Williams05} 
we know that SETH implies OVC}\inshort{SETH implies OVC~\cite{Williams05}}, i.e.,
whenever a problem is hard assuming OVC, it is also hard when assuming SETH.
Hence, it is preferable to use OVC for proving lower bounds.
Finally, to the best of our knowledge, no relations between the former two conjectures and the latter two conjectures 
are known.
\begin{remark}
	The conjectures that no \emph{polynomial} improvements over the best known
	running times are possible do not exclude improvements by sub-polynomial 
	factors such as poly-logarithmic factors or factors of, e.g., $2^{\sqrt{\log n}}$
	as in \cite{Williams14a}.
\end{remark}

\infull{
\section{Reachability in MDPs}\label{sec:reach}

First let us briefly discuss reachability on Graphs.
The winning set for disjunctive reachability can simply be computed by union all target sets and then starting
a breadth-first search which is in $O(m)$.
On the other hand, the problem becomes $\NP$-complete when considering conjunctive reachability~\cite{FijalkowH12},
as with conjunction one can require a path to contain several vertices
and in particular one can embed the well-known $\NP$-hard problem of Hamiltonian path.

Turning to MDPs, notice that in MDPs based on acyclic graphs almost-sure reachability is equivalent to 
computing the winning set for a player with reachability objectives in a 2-player graph-game where 
all the random vertices are owned by the opponent 
(as random will play the optimal strategy for the opponent with non-zero probability).
As computing the winning set for conjunctive reachability in the 2-player graph-game is $\PSPACE$-hard~\cite{FijalkowH12}
even for acyclic graphs, we have that conjunctive almost-sure reachability in MDPs is $\PSPACE$-hard as well.
Moreover, as we will show later,  compared to graphs, also disjunctive reachability becomes harder, 
i.e., we will provide polynomial lower-bound based on popular conjectures.

In the first part of this  section we present an improved algorithm for disjunctive reachability queries in MDPs.
As disjunctive reachability objectives can be easily reduced to a single reachability objective
by taking the union of all target sets, the algorithm mentioned above is also
an algorithm for disjunctive reachability objectives (by setting $k=1$).
In the second part we present two lower bounds for disjunctive reachability queries,
an $\Omega(n^{3-o(1)})$ lower bound based on STC and an $\Omega(m^{2-o(1)})$ lower bound based on OVC (resp.\ SETH).

\subsection{Algorithm for Disjunctive Reachability Queries in MDPs}
In this section we present an algorithm to compute the almost-sure winning
set for disjunctive reachability queries in MDPs. 
In particular we show the following theorem:
\begin{theorem}\label{th:timedrmdp}
	For an MDP~$\mdp$ and target sets $\target_i \subseteq V$ 
for $1 \le i \le k$ the almost-sure winning set for disjunctive reachability
queries can be computed in $O(k m + \textsc{MEC})$ time, where MEC is the time 
needed to compute a MEC-decomposition.
\end{theorem}
A vertex~$v$ is in the almost-sure winning
set if player~1 has a strategy to reach one of the $k$ target sets $\target_i$
with probability~1 starting from $v$.
Note that the sets $\target_i$ are not absorbing in contrast to what is often 
assumed for the reachability objective in MDPs. The trivial algorithm would be 
to invoke an algorithm for almost-sure reachability in MDPs $k$ times (for one 
target set $\target_i$ at a time, temporarily making the set $\target_i$ absorbing
if necessary). The 
crucial observation to improve upon this is that given an MDP without 
non-trivial end-components, almost-sure reachability in MDPs can be solved in 
linear time.

We further observe that, for each target set,
either all vertices of an end-component are winning (almost-surely) or none.
Thus if we know the MEC-decomposition of an MDP, 
we can contract the MECs 
to single vertices with self-loops and solve almost-sure reachability on the 
derived MDP. This derived MDP does not have non-trivial end-components,
therefore given the MEC decomposition, the problem can be solved in linear time
per target set. Our algorithm implies that
almost-sure reachability (i.e.\ $k=1$) can be solved in the same asymptotic 
time needed to determine the MEC-decomposition of an MDP.

\begin{definition}[Contraction of MECs]
	Contracting a MEC $\ec$ in an MDP $\mdp$ 
	creates a modified MDP $\mdp'$ from $\mdp$ where the vertices of $\ec$ are 
	replaced by a single vertex~$u$ that belongs to player~1 and the edges to or 
	from a vertex in $\ec$ are replaced with edges to or from, respectively, the 
	vertex~$u$; parallel edges are omitted from $P'$, for parallel random edges
	the probabilities are added up.
\end{definition}

\begin{observation}[\cite{ChatterjeeH11}]\label{obs:mecfree}
	The MDP $\mdp'$ that is constructed from the MDP $\mdp$ by contracting all 
	MECs of $\mdp$ does not contain any non-trivial end-components.
\end{observation}
\begin{proof}
Assume by contradiction that the MDP $\mdp'$ contains an end-component~$\ec'$ 
with at least two vertices. Let $\ec$ be the set of vertices corresponding 
to the vertices of $\ec'$ in the original MDP~$\mdp$. Then $\ec$ is an 
end-component in $\mdp$, a contradiction to the definition of~$\mdp'$.
\end{proof}

In the derived MDP we basically apply, for each target set,
one iteration of the classical almost-sure reachability algorithm but with a slightly
modified random attractor computation defined below. The classical algorithm
repeatedly executes the following two steps:
1) Compute the vertices $S$ from which player~1 can reach the target set~$\target$.
2a) If $S = V$, output $S$ as the (almost-sure) winning set of player~1.
2b) If $S \subsetneq V$, remove the random attractor of $V \setminus S$ 
from the graph (and from $V$) and repeat.
Intuitively, a \emph{random attractor} of a set of vertices $W$ contains the vertices
from which there is a positive probability to reach $W$ for every strategy of
player~1. The \emph{extended random attractor}, formally defined below
and used implicitly in~\cite{ChatterjeeH14}, additionally 
includes player~1 vertices for which the only player~1 strategy to avoid a 
positive probability to reach $W$ is using a self-loop of a vertex not in 
the target set.
Additionally, we explicitly avoid adding vertices in the considered target set 
to the attractor. In the classical algorithm this was achieved by making the 
target set absorbing, which would not work for the extended random attractor.

\begin{definition}[Extended Random Attractor]\label{def:extattr}
Let $E(v)$ denotes the set of vertices $u \in V$
for which $(v, u) \in E$.
In an MDP $\mdp = ((V, E), (\vo, \vr), \trans)$ 
the \emph{extended random attractor} $\ate(\mdp, W, \target)$ for sets of vertices 
$W, \target \subseteq V$ is defined as $\ate(\mdp, W, \target) = \bigcup_{j \ge 0} Z_j$ 
where $Z_0 = W \setminus \target$ and 
$Z_j$ for $j > 0$ is defined recursively as $Z_{j+1} = Z_j \cup \set{v \in \vr 
\mid E(v) \cap Z_j \ne \emptyset} \cup \set{v \in \vo \mid E(v) \subseteq Z_j \cup
\set{v}} \setminus \target$. In contrast to a random attractor \upbr{a} a set of 
vertices~$\target$ can be specified that is never included in $\ate(\mdp, W, \target)$ 
and \upbr{b} a player~1 vertex is also included in $Z_{j+1}$ if all its outgoing 
edges apart from its self-loop are contained in $Z_j$.
The extended random attractor $A = \ate(\mdp, W, \target)$
can be computed in $O(\sum_{v \in 
A} \InDeg(v) + \lvert \vo \setminus \target \rvert)$ time~\cite{Beeri80,Immerman81}.
\end{definition}

Putting the pieces together, our algorithm looks as follows: First, the 
MEC-decomposition of the input MDP~$\mdp$ is computed. Then all MECs of $\mdp$
are contracted to construct the derived MDP~$\mdp'$, which does not contain 
any non-trivial MECs. For each target set we execute one iteration of the classical
algorithm, replacing the usual random attractor with the extended random attractor.
The union of the winning sets determined for each target set then gives the 
winning set of player~1 for disjunctive reachability.

\begin{algorithm}
	\SetAlgoRefName{DisjReachMDP}
	\caption{Disjunctive Query Reachability in MDPs}
	\label{alg:drmdp}
	\SetKwInOut{Input}{Input}
	\SetKwInOut{Output}{Output}
	\BlankLine
	\Input{an MDP $\mdp = ((V, E), (\vo, \vr), \trans)$ and 
	target sets $\target_i \subseteq V$ for $1 \le i \le k$
	}
	\Output
	{
	$\bigvee_{1 \le i \le k} \as{\mdp, \reacht{\target_i}}$
	}
	\BlankLine
	compute MEC decomposition of $\mdp$\;
	let $P'$ be $P$ with all MECs contracted\;
	let $\target_i'$ for $1 \le i \le k$ be the set
	of vertices of $P'$ that represent some vertex of $\target_i$\;
	$W' \gets \emptyset$\;
	\For{$i \gets 1$ \KwTo $k$}{
		$S' \gets \reach(\mdp', \target'_i)$\;
		$A' \gets \ate(\mdp', V' \setminus S', \target'_i)$\;
		$W' \gets W' \cup V' \setminus A'$\;
	}
	let $W$ be the vertices in $W'$ after undoing contraction\;
	\Return{$W$}\;
\end{algorithm}

\begin{proposition}[Runtime]\label{prop:timedrmdp}
	Algorithm~\ref{alg:drmdp} runs in time $O(k m + \textsc{MEC})$.
\end{proposition}
\begin{proof}
	Contracting all MECs can be done in time $O(m)$ as we have to consider each 
	edge (and vertex) at most twice. The for-loop is executed $k$ times.
	Within the for-loop both the vertices~$S$ that can reach $\target_i$ and the 
	extended random attractor $A = \ate(\mdp', V \setminus S, \target_i)$ can be found 
	in linear time, that is, in $O(k m)$ time over all iterations of the for-loop.
	Undoing the contraction takes again at most $O(m)$ time.
\end{proof}

\begin{proposition}[Correctness]\label{prop:corrdrmdp}
For an MDP~$\mdp$ and target sets $\target_i \subseteq V$ 
for $1 \le i \le k$ Algorithm~\ref{alg:drmdp} returns the set
$\bigvee_{1 \le i \le k} \as{\mdp, \reacht{\target_i}}$.
\end{proposition}
\begin{proof}
We assume that in the MDP $\mdp$ each vertex 
has at least one outgoing edge and each random vertex has at least one outgoing 
edge that is not a self-loop. This is w.l.o.g.\ because $\bigvee_{1 \le i \le k} \as{\mdp,\reacht{\target_i}}$ does not change if we replace each vertex without 
outgoing edges by a vertex with a self-loop and treat a random vertex whose
only outgoing edge is a self-loop as a player~1 vertex.

First note that by definition a vertex is in $\bigvee_{1 \le i \le k} \as{\mdp, \reacht{\target_i}}$
if and only if it is in $\as{\mdp, \reacht{\target_i}}$ for some $1 \le i \le k$. 
Hence we can consider the $k$ target
sets separately by showing that in the $i$-th iteration of the for-loop of 
Algorithm~\ref{alg:drmdp} the set $\as{\mdp, \reacht{\target_i}}$ is identified.

Let $\mdp'$ be the MDP derived from the MDP $\mdp$ by contracting all
MECs of $\mdp$ and let $\target_i'$ be the set of contracted vertices that 
represent some vertex of $\target_i$ as in Algorithm~\ref{alg:drmdp}.
We use the superscript $'$ to denote sets related to the MDP $\mdp'$ and 
omit the superscript for sets related to the original MDP $\mdp$.
 Note that since only strongly connected subgraphs are 
contracted in $\mdp'$, it clearly holds that a vertex $v \in V$ can reach 
another vertex $u \in V$ if and only if the vertex $v' \in 
V'$ corresponding to $v$ can reach the vertex $u' \in V'$ corresponding to $u$.

Fix some iteration $i$ and let 
$S' = \reach(\mdp', \target'_i)$, let $A' = \ate(\mdp', V' \setminus S', 
\target'_i)$, and let $W'_i = V' \setminus A'$,
that is, $W'_i$ is the set added to $W'$ in the $i$-th iteration of the 
for-loop of Algorithm~\ref{alg:drmdp}. Let the same letters without superscript
denote the corresponding sets of vertices after reverting the contraction
of the MECs of $\mdp$.
We prove the lemma by first showing $\as{\mdp, \reacht{\target_i}} \subseteq 
W_i$ and then $W_i \subseteq \as{\mdp, \reacht{\target_i}}$.\smallskip

We prove  $\as{\mdp, \reacht{\target_i}} \subseteq W_i$ by showing $A \subseteq 
V \setminus \as{\mdp, \reacht{\target_i}}$ by induction on the 
recursive definition of $A' = \ate(\mdp', V' \setminus S', \target'_i) =
\cup_{j \ge 0} Z'_j$, where the sets $Z'_j$ are defined as in 
Definition~\ref{def:extattr} and the sets $Z_j$ are the corresponding
sets after reverting the contraction of the MECs of $\mdp$. Since the attractor
computation is done on $\mdp'$, each
set $Z_j$ either contains all vertices of a MEC of $\mdp$ or none.
Clearly $A \cap \target_i = \emptyset$ as vertices
in $\target_i'$ are explicitly excluded from $A'$.
Player~1 cannot reach $\target_i$ almost surely from the vertices in
$Z_0 = V \setminus S$ because these vertices cannot reach any vertex in 
$\target_i$. Assume the claim holds for $Z_j$, i.e., for all vertices $z \in Z_j$
and any strategy $\str$ of player~1 we have $\pr{\str}{z}{\mdp, 
\reacht{\target_i}} < 1$. 
By the definition of $Z'_{j+1}$, for a random vertex $v'$ in $Z'_{j+1} \setminus Z'_j$
there is a positive probability to reach a vertex in $Z'_j$; thus, 
$\pr{\str'}{v'}{\mdp', \reacht{\target_i'}} < 1$ for any strategy $\str'$ 
of player~1. Random vertices in $\mdp'$ were not contracted, thus the same
argument holds for $Z_{j+1}$ and $\mdp$. 
A player~1 vertex $\inec'$ in $Z'_{j+1} \setminus Z'_j$ corresponds to either a 
player~1 vertex $\inec$ or a MEC $\ec$ in $Z_{j+1} \setminus Z_j$. In both cases
all the edges from $\inec$ resp.\ $\ec$ lead to vertices in $Z_j$ or to $\inec$
resp.\ $\ec$ itself. Hence since $\inec \notin \target_i$ resp.\ $\ec \cap \target_i 
= \emptyset$, we also have $\pr{\str}{\inec}{\reacht{\target_i}} < 
1$ for any strategy $\str$ of player~1 and $\inec$ resp.\ all $\inec \in \ec$. 
\smallskip

We next show $W_i \subseteq  \as{\mdp, \reacht{\target_i}}$.
Let $G[W_i] = (W_i, E \cap (W_i \times W_i))$ be the subgraph induced by 
the vertices in $W_i$.
We establish two properties:
(1) all outgoing edges of random vertices $\vr \cap W_i$  lead to vertices in $W_i$, and
(2) all vertices in $W_i \setminus \target_i$ can reach $\target_i$ in $G[W_i]$.
The claim follows from these two properties using the same proof as for the 
classical algorithm for almost-sure reachability in MDPs (see below).

\begin{enumerate}
  \item[(1)] For vertices in $\vr$ we distinguish whether they are contained in a MEC of $\mdp$
	or not. In the first case property~(1) follows from the fact that a MEC has no
	outgoing random edges and every MEC is either completely contained in $W_i$
	or completely contained in $V \setminus W_i$.
	In the second case property~(1) follows from the definition of an extended 
	random extractor because a vertex in $\vr \cap W'_i$ with an edge to a vertex in 
	$A$ would have been included in $A$.

  \item[(2)] To show property~(2) we will use that by Observation~\ref{obs:mecfree} the 
	MDP $\mdp'$ does not contain any non-trivial MEC. 
	Assume by contradiction that some vertices in 
	$W_i \setminus \target_i$ cannot reach $\target_i$ in $G[W_i]$.
	Then there exists a bottom SCC $\scc$ (i.e.\ an SCC without outgoing edges, 
	possibly a single vertex)
	in $G[W_i]$ with $\scc \cap \target_i = \emptyset$. Note that 
	every MEC in $G[W_i]$ is completely contained in one of the SCCs of $G[W_i]$.
	By property~(1) $\scc$ has no outgoing random edges in $\mdp$; by this and the fact 
	that $\scc$ is strongly connected,
	the corresponding set $\scc'$ of vertices in $\mdp'$ 
	would be a non-trivial MEC in $\mdp'$ if it contained
	more than one vertex. Thus $\scc'$ can contain only one vertex~$\inscc'$ and this 
	vertex has either no outgoing 
	edge or only a self-loop in $G'_{W'_i}$. If $\inscc'$ was 
	a player~1 vertex, then all its outgoing edges would go to vertices in $A'$ 
	or be a self-loop, hence $\inscc'$ would have been included in the attractor $A'$.
	If $\inscc'$ was a random vertex, then by the assumption that in $\mdp$, and thus in 
	$\mdp'$, every random vertex 
	has an outgoing edge that is not a self-loop we would get a contradiction to
	property~(1). Thus no such bottom SCC $\scc$ can exist, that is, every bottom SCC
	of $G[W_i]$ contains a vertex of $\target_i$ and thus property~(2) holds.
\end{enumerate}

To see that the two established properties imply 
$W_i \subseteq \as{\mdp, \reacht{\target_i}}$, let for a vertex $u \in W_i$
be $d(u)$ the shortest path distance to a vertex in $\target_i$. Consider the 
following strategy~$\str$ of player~1: For a player~1 vertex $u$, 
choose an edge to a vertex $v$ such that $d(v) < d(u)$. For a random vertex $u$, 
there is always an edge to a vertex $v$ such that $d(v) < d(u)$.
Let $\ell = \lvert W_i \rvert$ and let $\alpha$ be the minimum positive
transition probability in the MDP~$\mdp$. For all vertices $v \in W_i$
the probability that $\target_i$ is reached within $\ell$ steps 
is at least $\alpha^\ell$, that is, the probability that $\target_i$ is not 
reached within $b \cdot \ell$ steps is at most $(1-\alpha^\ell)^b$, 
which goes to $0$ as $b$ goes to $\infty$. Thus for all $v \in W_i$ strategy~$\str$
ensures that $\target_i$ is reached with probability~1.
\end{proof}

\subsection{Conditional Lower Bounds for Disjunctive Reachability in MDPs}
\label{subsec:reach_lowerbounds}

}
\inshort{
\subsection{Disjunctive Reachability in MDPs}
\label{subsec:reach_lowerbounds}
}

\infull{
Here we complement the above algorithm by conditional lower bounds for disjunctive reachability queries in MDPs.
These lower bound will be based on the conjectures STC, SETH, and OVC 
introduced in Section~\ref{sec:conjectures}.
}

We first present our lower bound for dense MDPs based on STC.
\begin{theorem}\label{thm:reach_STChard}
  There is no combinatorial $O(n^{3-\epsilon})$ or $O((k\cdot n^2)^{1-\epsilon})$ algorithm \upbr{for any $\epsilon > 0$} for disjunctive reachability queries in MDPs under Conjecture~\ref{conj:triangle} \upbr{i.e., unless STC and BMM fail}. 
  \infull{In particular, there is no such algorithm deciding whether the winning set is non-empty
  or deciding whether a specific vertex is in the winning set.}
  The bounds hold for dense MDPs with $m = \Theta(n^2)$.
\end{theorem}

The above theorem is by the following reduction from the triangle detection problem.
\begin{reduction}\label{red:TriangletoMDPReach}
 Given an instance of triangle detection, i.e., a graph $G=(V,E)$,
 we build the following MDP $\mdp$. 
 \begin{itemize}
  \item  The vertices $V'$ of $\mdp$ are given by four copies 
  $V^1, V^2, V^3, V^4$ of $V$, a start vertex~$s$, and absorbing vertices 
  $F=\{g_v \mid v \in V\}$. The edges $E'$ of $\mdp$ are defined as follows:
  There is an edge from $s$ to the first copy~$v^1 \in V^1$ of every $v \in V$ 
  and the last copy~$v^4 \in V^4$ of every $v \in V$ is connected to its first copy~$v^1$
  and its corresponding absorbing vertex~$g_v \in F$; further for $1 \le i \le 3$
  there is an edge from $v^i$ to $u^{i+1}$ iff $(v,u) \in E$.
	  
  \item The set of vertices $V'$ is partitioned into player~1 vertices $V'_1=\{s\} \cup V^1 \cup V^2 \cup V^3\cup F$
	and random vertices $V'_R= V^4$.
	Moreover, the probabilistic transition function for each vertex $v \in V'_R$ 
	chooses among $v$'s successors with equal probability $1/2$ each.
	
 \end{itemize}
\end{reduction}
\begin{figure}
 \centering
 \begin{tikzpicture}[yscale=0.8, xscale=1.6,>=stealth]
 \small
		\draw  (0,-1.5)node[player](s){$s$};
  		\path 	(1,0)node[player](a1){$a^1$}
			++(0,-1.5)node[player](b1){$b^1$}
			++(0,-1.5)node[player](c1){$c^1$}
			;
  		\path 	(2,0)node[player](a2){$a^2$}
			++(0,-1.5)node[player](b2){$b^2$}
			++(0,-1.5)node[player](c2){$c^2$}
			;
  		\path 	(3,0)node[player](a3){$a^3$}
			++(0,-1.5)node[player](b3){$b^3$}
			++(0,-1.5)node[player](c3){$c^3$}
			;
		\path 	(4,0)node[random](a4){$a^4$}
			++(0,-1.5)node[random](b4){$b^4$}
			++(0,-1.5)node[random](c4){$c^4$}
			;
		\path 	(5,0)node[player](g1){$g_a$}
			++(0,-1.5)node[player](g2){$g_b$}
			++(0,-1.5)node[player](g3){$g_c$}
			;	
			
		\path [->, thick]
			(s) edge (a1)
			(s) edge (b1)
			(s) edge (c1)
			(a1) edge (b2)
			(a2) edge (b3)
			(a3) edge (b4)
			(b1) edge (c2)
			(b2) edge (c3)
			(b3) edge (c4)
			(c1) edge (a2)
			(c2) edge (a3)
			(c3) edge (a4)
			(b1) edge (a2)
			(b2) edge (a3)
			(b3) edge (a4)
			(a4) edge (g1)
			(b4) edge (g2)
			(c4) edge (g3)
			;
		\draw[->,thick, bend right, out= -45 , in=-135]	(a4) edge (a1);
		\draw[left,thick,->,rounded corners=5pt]  (b4) -- (4,-2.25) -- (1,-2.25) -- (b1);
		\draw[->,thick, bend left,out= 45 , in=135]	(c4) edge (c1);

 \end{tikzpicture}
 \infull{
 \caption{Illustration of Reduction~\ref{red:TriangletoMDPReach}, with $G=(\{a,b,c\},\{(a,b),(b,a),(b,c),(c,a)\})$. Vertices drawn as cycle are owned by player~1,
vertices drawn as diamond are random vertices.}
 }
 \inshort{
 \caption{Illustration of Reduction~\ref{red:TriangletoMDPReach}, with $G=(\{a,b,c\},$ $\{(a,b),(b,a),(b,c),(c,a)\})$. Vertices drawn as cycle are owned by player~1,
vertices drawn as diamond are random vertices.}
 }
 
 \label{fig:TriangletoMDPReach}
\end{figure}
The reduction is illustrated in Figure~\ref{fig:TriangletoMDPReach}.  \infull{
As all random choices are uniformly at random we omit the exact probabilities in the figures.

}
Next we prove that Reduction~\ref{red:TriangletoMDPReach} is indeed a valid reduction 
from triangle detection to disjunctive reachability queries in MDPs.

\begin{lemma}
A graph $G$ has a triangle iff $s$ is contained in $\bigvee_{v \in V} \as{\mdp, \reacht{\target_v}}$, where 
$\mdp$ is the MDP  given by Reduction~\ref{red:TriangletoMDPReach} and  $\target_v=\{g_v\}$ for $v \in V$.
\end{lemma}
\begin{proof}
 For the only if part assume that $G$ has a triangle with vertices $a,b,c$ and 
 let $a^i$,$b^i$,$c^i$ be the copies of $a,b,c$ in $V^i$.
 Now a strategy for player~1 in the MDP $\mdp$ to reach $g_a$ with probability~1 is as follows:
 When in $s$, go to $a^1$; when in $a^1$, go to $b^2$; when in $b^2$, go to $c^3$;
 when in $c^3$, go to $a^4$.  
 As $a,b,c$ form a triangle, all the edges required by the above strategy exist.
 When player~1 starts in $s$ and follows the above strategy the only random vertex he 
 encounters is $a^4$.
 The random choice sends him to the target vertex $g_a$ and to vertex $a^1$
 with probability $1/2$ each.
 In the former case he is done, in the latter case he continues playing his strategy and will reach $a^4$ again after three steps.
 The probability that player~1 has reached $g_a$ after $3q+1$ steps is $1-(1/2)^q$
 which converges to $1$ with $q$ going to infinity.
 Thus we have found a strategy to reach $g_a$ with probability $1$.

 For the if part assume that $s \in \bigvee_{v \in V} \as{\mdp, \reacht{\target_v}}$. 
 That is, there is an $a \in V$ such that $s \in \as{\mdp, \reacht{\target_a}}$.
 Let us consider a corresponding strategy for reaching $\target_a = \{g_a\}$.
 First, assume that the strategy would visit a vertex $v^4$ for $v \in V \setminus \{a\}$. 
 Then with probability $1/2$ player~1 would end up in the vertex $g_v$ which has no path to $g_a$, 
 a contradiction to $s \in \as{\mdp, \reacht{\target_a}}$. 
 Thus the strategy has to avoid visiting vertices $v^4$ for $v \in V \setminus \{a\}$.
 Second, as the only way to reach $g_a$ is $a^4$, the strategy has to choose $a^4$.
 But then with probability $1/2$ it will be send to $a^1$
 and there must be a path from $a^1$ to $g_a$ that doesn't not cross $V^4 \setminus \{a^4\}$.
 By the latter this path must be of the form $a^1, b^2, c^3, a^4, g_a$ for some $b,c \in V$.
 Now by the construction of $G'$ in the MDP $\mdp$ the vertices $a,b,c$ form a triangle in the original graph $G$.
\end{proof}

The size and the construction time of the MDP $\mdp$, constructed by Reduction~\ref{red:TriangletoMDPReach}, is linear in the size of the 
original graph $G$ and we have $k = \Theta(n)$ target sets.
Thus if we would have a combinatorial $O(n^{3-\epsilon})$ or 
$O((k\cdot n^2)^{1-\epsilon})$ algorithm for 
disjunctive queries of reachability objectives in MDPs for any $\epsilon > 0$, we
would immediately get a combinatorial  $O(n^{3-\epsilon})$ algorithm for 
triangle detection, which contradicts STC and BMM.

Next we present a lower bound for sparse MDPs based on OVC and SETH.

\begin{theorem}\label{thm:reach_OVChard}
  There is no $O(m^{2-\epsilon})$ or $O((k\cdot m)^{1-\epsilon})$ algorithm
  \upbr{for any $\epsilon > 0$} for 
  disjunctive reachability queries in MDPs under Conjecture~\ref{conj:ov} \upbr{i.e., unless OVC and SETH fail}.
  \infull{In particular, there is no such algorithm deciding whether the winning set is non-empty
  or deciding whether a specific vertex is in the winning set.}
\end{theorem}

To prove the above we give a reduction from OVC to disjunctive reachability queries
in MDPs.

\begin{reduction}\label{red:OVtoMDPReach}
 Given two sets $S_1, S_2$ of $d$-dimensional vectors, we build the following MDP~$\mdp$.  
 \begin{itemize}
  \item The vertices $V$ of the MDP~$\mdp$
  are given by a start vertex~$s$, vertices $S_1$ and $S_2$ representing the 
  sets of vectors, vertices $\mathcal{C}=\{c_i \mid 1 \leq i \leq d\}$ representing the 
  coordinates, and absorbing vertices $F=\{g_v \mid v \in S_2\}$.
  The edges $E$ of~$\mdp$ are defined as follows: the start vertex~$s$
  has an edge to every vertex of $S_1$ and every vertex $v \in S_2$ has an edge to $s$
  and to its corresponding absorbing vertex $g_v \in F$; further for each $x \in S_1$
  there is an edge \emph{to} $c_i \in \mathcal{C}$ iff $x_i=1$ and for each $y \in S_2$
  there is an edge \emph{from} $c_i \in \mathcal{C}$ iff $y_i=0$.
	  
  \item The set of vertices $V$ is partitioned into player~1 vertices $V_1=\{s\} \cup \mathcal{C} \cup F$
	and random vertices $V_R= S_1 \cup S_2$.
	The probabilistic transition function for each vertex $v \in V_R$ chooses among $v$'s successors
	uniformly at random.
 \end{itemize}
\end{reduction}

\begin{figure}
 \centering
 \begin{tikzpicture}[yscale=0.8,>=stealth]
\footnotesize
 \tikzstyle{vrandom}=[draw, thick, diamond, rounded corners, fill=gray!15,inner sep=0.5pt, minimum width=12pt]
  		\path 	node[player](s){$s$}
			++(2,2) node[vrandom](x1){$(\!1,\!0,\!0\!)$}
			++(0,-2)node[vrandom] (x2){$(\!1,\!1,\!1\!)$}
			++(0,-2) node[vrandom](x3){$(\!0,\!1,\!1\!)$}
			(4,2) node[player](c1){$c_1$}
			++(0,-2) node[player](c2){$c_2$}
			++(0,-2) node[player](c3){$c_3$}
			(6,2) node[vrandom](y1){$(\!1,\!1,\!0\!)$}
			++(0,-2)node[vrandom] (y3){$(\!0,\!1,\!0\!)$}
			++(0,-2) node[vrandom](y4){$(\!0,\!0,\!1\!)$}
			(8,2) node[player](g1){$g_{110}$}
			++(0,-2)node[player] (g3){$g_{010}$}
			++(0,-2) node[player](g4){$g_{001}$}
			;
		\normalsize
		\path [->, thick]
			(s) edge (x1)
			(s) edge (x2)
			(s) edge (x3)
			(x1) edge (c1)
			(x2) edge (c1)
			(x2) edge (c2)
			(x2) edge (c3)
			(x3) edge (c2)
			(x3) edge (c3)
			[<-]
			(y1) edge (c3)
			(y3) edge (c1)
			(y3) edge (c3)
			(y4) edge (c1)
			(y4) edge (c2)
			[->]
 			(y1) edge (g1)    
 			(y3) edge (g3)    
 			(y4) edge (g4)  
			;		
		\draw (7,3.1) coordinate[](rc);
		\draw (0,3.1) coordinate[](lc);
		\draw[left,thick,->,rounded corners=5pt]  (y1) -- ++(1,0) -- (rc) -- (lc) -- (s);
		\draw[left,thick,->,rounded corners=5pt]  (y3) -- ++(1,0) -- (rc) -- (lc) -- (s);
		\draw[left,thick,->,rounded corners=5pt]  (y4) -- ++(1,0) -- (rc) -- (lc) -- (s);

 \end{tikzpicture}
  \infull{
 \caption{Illustration of Reduction~\ref{red:OVtoMDPReach} 
	  for $S_1=\{(1,0,0), (1,1,1), (0,1,1)\}$ and $S_2=\{(1,1,0), (0,1,0), (0,0,1)\}$.}
 }
 \inshort{
 \caption{Illustration of Reduction~\ref{red:OVtoMDPReach} 
	  for $S_1=\{(1,0,0),$ $(1,1,1), (0,1,1)\}$ and $S_2=\{(1,1,0), (0,1,0), (0,0,1)\}$.}
 }
 \label{fig:OVtoMDPReach}
\end{figure}
The reduction is illustrated on an example in Figure~\ref{fig:OVtoMDPReach}. 

\begin{lemma}
There exist orthogonal vectors $x \in S_1$, $y \in S_2$ iff $s \in \bigvee_{v \in V} \as{\mdp, \reacht{\target_v}}$ where 
$\mdp$ is the MDP  given by Reduction~\ref{red:OVtoMDPReach} and  $\target_v=\{g_v\}$ for $v \in V$.
\end{lemma}
\begin{proof}
 For the only if part assume that there are orthogonal vectors $x \in S_1$, $y \in S_2$.
 Now a strategy for player~1 in the MDP $\mdp$ to reach $g_y$ with probability~1 is as follows:
 When in $s$, go to $x$; when in some $c \in \mathcal{C}$, go to $y$.
 As $x$ and $y$ are orthogonal, each $c_i \in \mathcal{C}$ reachable from $x$ has an edge to $y$, i.e.,
 for $x_i=1$ it must be that $y_i=0$.
 When player~1 starts in $s$ and follows the above strategy, 
 he reaches $y$ after three steps. There the random choice sends
 him to the target vertex $g_y$ and back to vertex $y$ with probability $1/2$ each.
 In the former case he is done, in the latter case
 he continues playing his strategy and will reach $y$ again after three steps.
 The probability that player~1 has reached $g_y$ after $3q$ steps is $1-(1/2)^q$,
 which converges to $1$ with $q$ going to infinity.
 Thus we have found a strategy to reach $g_y$ with probability $1$.

 For the if part assume that $s \in \bigvee_{v \in V} \as{\mdp, \reacht{\target_v}}$. 
 That is, there is an $y \in S_2$ such that $s \in \as{\mdp, \reacht{\target_y}}$.
 Let us consider a corresponding strategy for reaching $\target_y = \{g_y\}$.
 First, assume that the strategy would visit a vertex $y' \in S_2$ for $y'\not= y$.
 Then with probability $1/2$ the player would end up in the vertex $g_{y'}$ which has no path to $g_{y}$, 
 a contradiction to $s \in \as{\mdp, \reacht{\target_y}}$. 
 Thus the strategy has to avoid visiting vertices $S_2 \setminus \{y\}$.
 Second, as the only way to reach $g_y$ is $y$, the strategy has to choose $y$.
 But then with probability $1/2$ it will be send to $s$
 and thus there must be a strategy to reach $g_y$ from $s$ with probability $1$ that does not cross $S_2 \setminus \{y\}$.
 As $y$ is the only predecessor of $g_y$, there must also be such a strategy to reach $y$.
 In other words, there must be an $x \in S_1$ such that for each successor $c_i \in \mathcal{C}$ there 
 is an edge to $y$. 
 By the construction of the MDP~$\mdp$ this is equivalent to the existence of
 an $x \in S_1$  such that whenever $x_i=1$ then $y_i=0$, 
 and thus $x$ and $y$ are orthogonal vectors.
\end{proof}

The number of vertices in $\mdp$, constructed by Reduction~\ref{red:OVtoMDPReach},
is $O(N)$ and the construction can be performed in 
$O(N \log N)$ time (recall that $d \in O(\log N)$). 
The number of edges~$m$ is $O(N \log N)$ (thus we consider $\mdp$
to be a sparse MDP) and the number of target sets $k \in \Theta(N) = \theta(m/\log N)$.
Finally, if we would have an  $O(m^{2-\epsilon})$ or $O((k\cdot m)^{1-\epsilon})$ algorithm for disjunctive reachability queries in MDPs for any $\epsilon > 0$, we
would immediately get an $O(N^{2-\epsilon})$ algorithm for OV, which contradicts OVC (and thus SETH).

\infull{\section{Safety Objectives}\label{sec:safety}
It is well-known that computing the a.s.\ winning set for a single safety objective 
in an MDP is equivalent to computing
the winning set of player~1 for safety objectives in the 2-player graph-game where all the random vertices are owned by the opponent, called player~2 
(see e.g.~\cite{ChatterjeeDH10}).
A 2-player graph-game is defined as a graph with a partition 
of the vertices into player~1 vertices~$V_1$ and player~2 vertices~$V_2$. 
A player~2 strategy is defined analogous to a player~1 strategy (replacing the 
vertices~$V_1$ with the vertices~$V_2$ in the definition). 
The objective of player~2 is the dual of the objective of player~1.

Safety objectives in 2-player graph-games can be computed in $O(m)$~time by computing
a player~2 attractor (the definition of a player~1 or player~2 attractor is 
analogous to the definition of a random attractor in Definition~\ref{def:attr}).
Thus in MDPs the a.s.\ winning set for a single safety objective can be computed in $O(m)$~time by computing a random attractor,
and the a.s.\ winning set for a disjunctive query can be determined 
in $O(k\cdot m)$ time by computing $k$ random attractors and union the winning sets.
Conjunctive safety can be reduced to a single safety objective in $O(b)$ time 
by taking the union of all the sets $\target_i$.

Turning to disjunctive safety objectives, we have the same equivalence to 2-player graph-games as for single objectives (Observation~\ref{obs:SafetyGames}).
In this 2-player game the disjunctive safety objective is the 
complementary objective to the 
conjunctive reachability objective with the same sets and, as the game is determined~\cite{FijalkowH12}\footnote{A graph-game is determined if the winning 
set of player~1 is the complement of the winning set of player~2.},
the $\PSPACE$-hardness shown in~\cite{FijalkowH12} also applies to disjunctive safety objectives.

\begin{observation}\label{obs:SafetyGames}
    Computing the a.s.\ winning set for a disjunctive safety objective in an MDP 
    with player~1 vertices~$\vo$ and random vertices~$\vr$ is equivalent to 
    computing the same disjunctive safety objective in the 2-player graph-game
    with the same edges and the same player~1 vertices and 
    player~2 vertices~$V_2 = \vr$.
\end{observation}
\begin{proof}
 We show that a vertex $s$ is almost sure winning in the MDP if and only if it is winning for player~1 in the game graph.
 
 \noindent $\Leftarrow:$ Assume $s$ is not winning for player~1 in the graph-game. 
	       Then $s$ is winning for player~2 and thus player~2 has a strategy to visit all target sets from $s$.
	       As there are only finitely many target sets, all these target sets are visited after a finite number of steps,
	       lets say after $l$ steps.
	       Now consider the corresponding MDP; with some constant probability the random choices in the MDP will follow exactly the strategy of player~2 in the 
	       graph-game for the first $l$
	       steps and in that case player~1 cannot win almost surely from $s$.
	       Hence, $s$ is not in the a.s.\ winning set.
 
 \noindent $\Rightarrow:$ Assume player~1 has a winning strategy for the graph-game
 starting in $s$.
    By definition this strategy is also winning for the MDP (if it is winning for each possible choice of player~2 then it also
    winning for a random choice).
\end{proof}

\subsection{Conditional Lower Bounds for Safety Objectives}\label{subsec:safety_lowerbounds}
}
\inshort{
\subsection{Safety Objectives}\label{subsec:safety_lowerbounds}
}

We first present a lower bound for disjunctive safety based on STC that even holds on graphs.

\begin{theorem}\label{thm:safety_STChard}
  There is no combinatorial $O(n^{3-\epsilon})$ or $O((k\cdot n^2)^{1-\epsilon})$ algorithm \upbr{for any $\epsilon > 0$} for disjunctive safety \upbr{objectives or queries} in graphs under  Conjecture~\ref{conj:triangle} \upbr{i.e., unless STC and BMM fail}. 
  \infull{In particular, there is no such algorithm deciding whether the winning set is non-empty
  or deciding whether a specific vertex is in the winning set.}
\end{theorem}

The above is by the linear time reduction from triangle detection to disjunctive safety in graphs 
provided below.

\begin{reduction}\label{red:TriangletoGraphs}
 Given a graph $G=(V,E)$ \upbr{for triangle detection}, we build a graph $G'=(V',E')$ \upbr{for disjunctive safety} as follows.
 As vertices $V'$ we have four copies $V^1, V^2, V^3, V^4$ of $V$ and a vertex $s$.
 A vertex $v^i \in V^i$ has an edge to a vertex $u^{i+1} \in V^{i+1}$ iff $(v,u) \in E$. 	  
 Finally, $s$ has an edge to all vertices in $V^1$ and all vertices in $V^4$ have an edge to~$s$.
\end{reduction}
\begin{figure}
 \centering
 \begin{tikzpicture}[yscale=0.8, xscale=1.6,>=stealth]
 \small
		\draw  (0,-2.25)node[player](s){$s$};
  		\path 	(1,0)node[player](a1){$a^1$}
			++(0,-1.5)node[player](b1){$b^1$}
			++(0,-1.5)node[player](c1){$c^1$}
			++(0,-1.5)node[player](d1){$d^1$}
			;
  		\path 	(2,0)node[player](a2){$a^2$}
			++(0,-1.5)node[player](b2){$b^2$}
			++(0,-1.5)node[player](c2){$c^2$}
			++(0,-1.5)node[player](d2){$d^2$}
			;
  		\path 	(3,0)node[player](a3){$a^3$}
			++(0,-1.5)node[player](b3){$b^3$}
			++(0,-1.5)node[player](c3){$c^3$}
			++(0,-1.5)node[player](d3){$d^3$}
			;
		\path 	(4,0)node[player](a4){$a^4$}
			++(0,-1.5)node[player](b4){$b^4$}
			++(0,-1.5)node[player](c4){$c^4$}
			++(0,-1.5)node[player](d4){$d^4$}
			;
		\path [->, thick]
			(s) edge (a1)
			(s) edge (b1)
			(s) edge (c1)
			(s) edge (d1)
			(a1) edge (b2)
			(a2) edge (b3)
			(a3) edge (b4)
			(b1) edge (c2)
			(b2) edge (c3)
			(b3) edge (c4)
			(c1) edge (a2)
			(c2) edge (a3)
			(c3) edge (a4)
			(b1) edge (a2)
			(b2) edge (a3)
			(b3) edge (a4)
			
			(c1) edge (d2)
			(c2) edge (d3)
			(c3) edge (d4)
			
			(d1) edge (a2)
			(d2) edge (a3)
			(d3) edge (a4)
			;
		\draw[left,thick,->,rounded corners=5pt]  
		(a4) -- ++(0.5,0) -- (4.5,-5.2) -- (0,-5.2) -- (s);
		\draw[left,thick,->,rounded corners=5pt]  
		(b4) -- ++(0.5,0) -- (4.5,-5.2) -- (0,-5.2) -- (s);
		\draw[left,thick,->,rounded corners=5pt]  
		(c4) -- ++(0.5,0) -- (4.5,-5.2) -- (0,-5.2) -- (s);
		\draw[left,thick,->,rounded corners=5pt]  
		(d4) -- ++(0.5,0) -- (4.5,-5.2) -- (0,-5.2) -- (s);

 \end{tikzpicture}
 \infull{
 \caption{Illustration of Reduction~\ref{red:TriangletoGraphs}, with $G=(\{a,b,c,d\},\{(a,b), (b,a),(b,c), (c,a),(c,d),$ $(d,a)\})$.
 The target sets for disjunctive safety are $T_a = \set{b^1, c^1, d^1, b^4, c^4, d^4}$, $T_b = \set{a^1, c^1, d^1, a^4, c^4, d^4}$, $T_c = \set{a^1, b^1, d^1, a^4, b^4, d^4}$, and $T_d = \set{a^1, b^1, c^1, a^4, b^4, c^4}$.
 }
 }
 \inshort{
 \caption{Illustration of Reduction~\ref{red:TriangletoGraphs}, with $G=(\{a,b,c,d\},$ $\{(a,b),(b,a),(b,c),(c,a),(c,d),(d,a)\})$.
 The target sets for disjunctive safety are $T_a = \set{b^1, c^1, d^1, b^4, c^4, d^4}$, $T_b = \set{a^1, c^1, d^1, a^4, c^4, d^4}$, $T_c = \set{a^1, b^1, d^1, a^4, b^4, d^4}$, and $T_d = \set{a^1, b^1, c^1, a^4, b^4, c^4}$.}
 }
 \label{fig:TriangletoGraphs}
\end{figure}
Reduction~\ref{red:TriangletoGraphs} is illustrated in Figure~\ref{fig:TriangletoGraphs}.

\begin{lemma}
 Let $G'$ be the graph given by Reduction~\ref{red:TriangletoGraphs} for a graph~$G$
 and  
 let $\target_v= (V^1 \setminus \{v^1\}) \cup (V^4 \setminus \{v^4\})$. Then 
    \infull{the following statements are equivalent.
  \begin{enumerate}
  \item
  $G$ has a triangle.
  \item $s$ is in the winning set of  $(G',\bigvee_{v \in V} \objsty{Safety}{\target_v})$.
  \item The winning set of $(G', \bigvee_{v \in V} \objsty{Safety}{\target_v})$ is non-empty.
  \end{enumerate}
  }
  \inshort{$G$ has a triangle iff $s$ is in the winning set of  $(G',\bigvee_{v \in V} \objsty{Safety}{\target_v})$.}
\end{lemma}
\begin{proof}
 \infull{(1)$\Rightarrow$(2): Assume}\inshort{For the only if part assume} that
 $G$ has a triangle with vertices $a,b,c$ and 
 let $a^i$,$b^i$,$c^i$ be the copies of $a,b,c$ in $V^i$.
 Now a strategy for player~1 in $G'$ to satisfy $\objsty{Safety}{\target_a}$ is as follows:
 When in $s$, go to $a^1$; when in $a^1$, go to $b^2$; when in $b^2$, go to $c^3$;
 when in $c^3$, go to $a^4$; and when in $a^4$, go to $s$.
 As $a,b,c$ form a triangle, all the edges required by the above strategy exist.
 When player~1 starts in~$s$ and follows the above strategy,
 then he plays an infinite path
 that only uses vertices $s,a^1,b^2,c^3,a^4$ and thus satisfies $\objsty{Safety}{\target_a}$.

 \infull{(2)$\Rightarrow$(1): Assume}\inshort{For the if part assume} that there is a
 winning play starting in $s$ and satisfying $\objsty{Safety}{\target_a}$.
 Starting from $s$, this play has to first go to $a^1$, as all other successors of $s$ would violate
 the safety constraint. Then the play continues on some vertex $b^2 \in V^2$ and $c^3 \in V^3$
 and then, again by the safety constraint, has to enter $a^4$.
 Now by construction of $G'$ we know that there must be edges $(a,b), (b,c), (c,a)$ in the original graph $G$,
 i.e.\ there is a triangle in $G$.  \infull{
 
 (2)$\Leftrightarrow$(3): Notice that when removing $s$ from $G'$ we get an acyclic graph and thus each infinite
 path has to contain $s$ infinitely often. Thus, if the winning set is non-empty,
 there is a cycle winning for some vertex and then 
 this cycle is also winning for $s$. For the converse direction we have that if $s$ is in the winning set, then the winning set is non-empty.}
\end{proof}

The size and the construction time of the graph $G'$, constructed 
by Reduction~\ref{red:TriangletoGraphs}, is linear in the size of the 
original graph $G$ and we have $k = \Theta(n)$ target sets.
Thus if we would have a combinatorial $O(n^{3-\epsilon})$  or $O((k\cdot n^2)^{1-\epsilon})$ algorithm for disjunctive 
safety objectives or queries in graphs, we
would immediately get a combinatorial  $O(n^{3-\epsilon})$ algorithm for triangle detection, which contradicts STC (and thus BMM).

The above reduction uses a linear number of safety constraints which are all of linear size. 
Thus, a natural question is whether smaller safety sets would make the problem any easier. 
Next we argue that our result even holds for safety sets that are of logarithmic size.
To this end we modify Reduction~\ref{red:TriangletoGraphs} as follows. We remove all edges incident to $s$ and 
replace them by two complete binary trees. The first tree with $s$ as root and the vertices $V^1$ as leaves is directed towards the leaves, 
the second tree with root~$s$ and leaves~$V^4$ is directed towards~$s$.
Now for each pair $v^1,v^4$ one can select one vertex of each level of the trees (except for the root levels) 
for the set $\target_v$ such that the only safe path starting in $s$ has to use 
$v^1$ and each safe path to $s$ must pass $v^4$.
As the depth of the trees is logarithmic in the number of leaf vertices, we get
sets of logarithmic size. 
The construction with the binary trees 
is illustrated in Figure~\ref{fig:TriangletoGraphs_tree}.
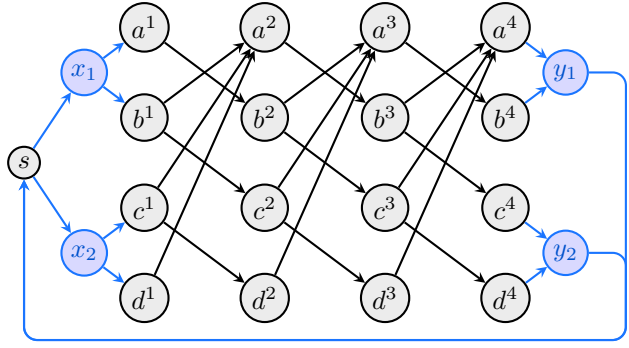
\begin{figure}
\centering
 \begin{tikzpicture}[yscale=0.8, xscale=1.6,>=stealth]
\tikzstyle{trcolor}=[color=blue!60!cyan!90!white,text=blue!60!cyan!80!black]
\tikzstyle{playerc}=[player,trcolor,fill=blue!15,]
\tikzstyle{edgec}=[thick,trcolor]
 \small
		\draw  (0,-2.25)node[player](s){$s$};
		\path (0.5,-0.75)node[playerc](x1){$x_1$}
			++(0,-3)node[playerc](x2){$x_2$}
			;
  		\path 	(1,0)node[player](a1){$a^1$}
			++(0,-1.5)node[player](b1){$b^1$}
			++(0,-1.5)node[player](c1){$c^1$}
			++(0,-1.5)node[player](d1){$d^1$}
			;
  		\path 	(2,0)node[player](a2){$a^2$}
			++(0,-1.5)node[player](b2){$b^2$}
			++(0,-1.5)node[player](c2){$c^2$}
			++(0,-1.5)node[player](d2){$d^2$}
			;
  		\path 	(3,0)node[player](a3){$a^3$}
			++(0,-1.5)node[player](b3){$b^3$}
			++(0,-1.5)node[player](c3){$c^3$}
			++(0,-1.5)node[player](d3){$d^3$}
			;
		\path 	(4,0)node[player](a4){$a^4$}
			++(0,-1.5)node[player](b4){$b^4$}
			++(0,-1.5)node[player](c4){$c^4$}
			++(0,-1.5)node[player](d4){$d^4$}
			;
		\path (4.5,-0.75)node[playerc](y1){$y_1$}
			++(0,-3)node[playerc](y2){$y_2$}
			;
		\path [->, thick]			
			(a1) edge (b2)
			(a2) edge (b3)
			(a3) edge (b4)
			(b1) edge (c2)
			(b2) edge (c3)
			(b3) edge (c4)
			(c1) edge (a2)
			(c2) edge (a3)
			(c3) edge (a4)
			(b1) edge (a2)
			(b2) edge (a3)
			(b3) edge (a4)
			
			(c1) edge (d2)
			(c2) edge (d3)
			(c3) edge (d4)
			
			(d1) edge (a2)
			(d2) edge (a3)
			(d3) edge (a4)
			;
		\path [->,edgec]
			(s) edge (x1)
			(s) edge (x2)
			(x1) edge (a1)
			(x1) edge (b1)
			(x2) edge (c1)
			(x2) edge (d1)
			(a4) edge (y1)
			(b4) edge (y1)
			(c4) edge (y2)
			(d4) edge (y2)
			;
		\draw[left,edgec,->,rounded corners=5pt]  
		(y1) -- ++(0.5,0) -- (5,-5.2) -- (0,-5.2) -- (s);
		\draw[left,edgec,->,rounded corners=5pt]  
		(y2) -- ++(0.5,0) -- (5,-5.2) -- (0,-5.2) -- (s);

 \end{tikzpicture}
  \caption{Illustration of how to reduce the number of entries in the target sets in Reduction~\ref{red:TriangletoGraphs} with two complete binary trees. Here
  $G=(\{a,b,c,d\},\{(a,b),(b,a),(b,c),(c,a),(c,d),(d,a)\})$ and the target sets
  for disjunctive safety are $T_a = \set{b^1, x_2, b^4, y_2}$, 
  $T_b = \set{a^1, x_2, a^4, y_2}$, $T_c = \set{d^1, x_1, d^4, y_1}$, and 
  $T_d = \set{c^1, x_1, c^4, y_1}$.}
 \label{fig:TriangletoGraphs_tree}
\end{figure}

Next we present an $\Omega(m^{2-o(1)})$ lower bound for disjunctive objective/query safety in sparse MDPs.

\begin{theorem}\label{thm:safety_OVChard}
  There is no $O(m^{2-\epsilon})$ or $O((k\cdot m)^{1-\epsilon})$ algorithm \upbr{for any $\epsilon > 0$} for disjunctive safety objectives/queries in MDPs under Conjecture~\ref{conj:ov} \upbr{i.e., unless
  OVC\infull{ and }\inshort{ \& }SETH fail}.
  \infull{In particular, there is no such algorithm for deciding whether the winning set is non-empty
  or deciding whether a specific vertex is in the winning set.}
\end{theorem}

To prove the above, we give a linear time reduction from OV to disjunctive safety objectives/queries.

\begin{reduction}\label{red:OVtoMDPsafety}
 Given two sets $S_1, S_2$ of $d$-dimensional vectors, we build the following MDP~$\mdp$.  
 \begin{itemize}
  \item   The vertices $V$ of the MDP~$\mdp$
  are given by a start vertex~$s$, vertices $S_1$ and $S_2$ representing the 
  sets of vectors, and 
  vertices $\mathcal{C}=\{c_i \mid 1 \leq i \leq d\}$ representing the 
  coordinates. The edges $E$ of~$\mdp$ are defined as follows: the start vertex~$s$
  has an edge to every vertex of $S_1$ and every vertex $v \in S_2$ has an edge to $s$;
  further for each $x \in S_1$
  there is an edge \emph{to} $c_i \in \mathcal{C}$ iff $x_i=1$ and for each $y \in S_2$
  there is an edge \emph{from} $c_i \in \mathcal{C}$ iff $y_i=1$.
  \item The set of vertices $V$ is partitioned into player~1 vertices $V_1=\{s\}\cup S_2$
	and random vertices $V_R= S_1 \cup \mathcal{C}$.
	Moreover, the probabilistic transition function for each vertex $v \in V_R$ chooses among $v$'s successors
	uniformly at random.	
 \end{itemize}
\end{reduction}
The reduction is illustrated on an example in Figure~\ref{fig:OVtoMDPsafety}. 

\begin{lemma}\label{lem:OVtoMDPsafety}
  Given two sets $S_1, S_2$ of $d$-dimensional vectors,
  the corresponding MDP $\mdp$ given by Reduction~\ref{red:OVtoMDPsafety} and  
  $\target_v=\{v\}$ for $v \in S_2$
  the following statements are equivalent
  \begin{enumerate}
    \item There exist orthogonal vectors $x \in S_1$, $y \in S_2$.
    \item $s \in \bigvee_{v \in S_2} \as{\mdp, \objsty{Safety}{\target_v}}$
    \item $s \in \as{\mdp, \bigvee_{v \in S_2} \objsty{Safety}{\target_v}}$
    \infull{\item The winning set $\bigvee_{v \in S_2} \as{\mdp, \objsty{Safety}{\target_v}}$ is non-empty.
    \item The winning set $\as{\mdp, \bigvee_{v \in S_2} \objsty{Safety}{\target_v}}$ is non-empty.}
  \end{enumerate}
\end{lemma}
\begin{proof}
 W.l.o.g.\ we assume that the $1$-vector, i.e., the vector with all coordinates being $1$, is  contained in $S_2$
 (adding the $1$-vector does not change the result of the OV instance).
 Then a play in the MDP $\mdp$ proceeds as follows. 
 Starting from $s$, player~1 chooses a vertex $x \in S_1$; then a vertex
 $c \in \mathcal{C}$
 and then a vertex $y \in S_2$ are picked randomly; then the play goes 
 back to $s$, starting another cycle of the play.
 
 (1)$\Rightarrow$(2): Assume there are orthogonal vectors $x \in S_1$, $y \in S_2$.
    Now player~1 can satisfy $\objsty{Safety}{\target_y}$ in the MDP $\mdp$ 
    by simply going to $x$ whenever the play is in $s$.
    The random player will then send it to some adjacent $c \in \mathcal{C}$
    and then to some adjacent vertex in $S_2$, 
    but as $x$ and $y$ are orthogonal, this $c$ is not connected to $y$. 
    Thus the play will never visit $y$.
 
 (2)$\Rightarrow$(3): Assume $s \in \bigvee_{v \in S_2} \as{\mdp, \objsty{Safety}{\target_v}}$. Then 
    there is a vertex $y \in S_2$ such that $s \in \as{\mdp, \objsty{Safety}{\target_y}}$. 
    Now we can enlarge the objective
    to $\bigvee_{v \in S_2} \objsty{Safety}{\target_v}$ and obtain
    $s \in \as{\mdp, \bigvee_{v \in S_2} \objsty{Safety}{\target_v}}$.
 
 (3)$\Rightarrow$(1): Assume $s \in \as{\mdp, \bigvee_{v \in S_2} \objsty{Safety}{\target_v}}$ and consider
    a corresponding strategy~$\sigma$. 
    W.l.o.g.\ we can assume that this strategy is memoryless~\cite{Thomas95}. 
    Thus whenever the play is in~$s$, it picks a fixed $x\in S_1$ as the next vertex.
    Assume towards contradiction that there is no orthogonal vector $y \in S_2$  for $x$.
    Then for each $y \in S_2$ we have that there is a $c \in \mathcal{C}$ connecting $x$ to $y$.
    In each cycle of the play one goes from $s$ to $x$ and then by random choice to some vertex in $S_2$.
    By the above, each of the vertices in $S_2$ has a non-zero probability to be 
    reached in this cycle, which can, for each fixed $n$,
    be lower bounded by a constant $p$.
    Thus after $n$ cycles in the play with probability at least $p^{|S_2|}$ 
    all vertices in $S_2$ have been visited and thus none of the safety objectives is 
    satisfied, a contradiction to the assumption that with probability~$1$ at least one
    safety objective is satisfied.
    Thus there must exist a vector $y \in S_2$ orthogonal to~$x$.\infull{
  
  (2)$\Leftrightarrow$(4) \& (3)$\Leftrightarrow$(5): 
     Notice that when removing $s$ from $\mdp$ we get an acyclic MDP and 
     thus each infinite
     path has to contain $s$ infinitely often. 
     Certainly if $s$ is in the a.s.\ winning set, this set is non-empty. 
     Thus let us assume there is a vertex $v$ different from $s$ with a winning strategy $\sigma$.
     All (winning) paths starting in $v$ cross $s$ after at most $3$ steps and thus 
     $\sigma$ must be also winning when starting in $s$.}
\end{proof}
  \begin{figure}
  \centering
  \begin{tikzpicture}[yscale=0.8,>=stealth]
\footnotesize
  		\path 	node[player](s){$s$}
			++(2,2) node[vrandom](x1){$(\!1,\!0,\!0\!)$}
			++(0,-2)node[vrandom] (x2){$(\!1,\!1,\!1\!)$}
			++(0,-2) node[vrandom](x3){$(\!0,\!1,\!1\!)$}
			(4,2) node[random](c1){$c_1$}
			++(0,-2) node[random](c2){$c_2$}
			++(0,-2) node[random](c3){$c_3$}
			(6,3) node[vplayer](y1){$(\!1,\!1,\!0\!)$}
			++(0,-2) node[vplayer](y2){$(\!1,\!1,\!1\!)$}
			++(0,-2)node[vplayer] (y3){$(\!0,\!1,\!0\!)$}
			++(0,-2) node[vplayer](y4){$(\!0,\!0,\!1\!)$}
			;
		\normalsize
		\path [->, thick]
			(s) edge (x1)
			(s) edge (x2)
			(s) edge (x3)
			(x1) edge (c1)
			(x2) edge (c1)
			(x2) edge (c2)
			(x2) edge (c3)
			(x3) edge (c2)
			(x3) edge (c3)
			[<-]
			(y1) edge (c1)
			(y1) edge (c2)
			(y2) edge (c1)
			(y2) edge (c2)
			(y2) edge (c3)
			(y3) edge (c2)
			(y4) edge (c3)
			;
			
		\draw[left,thick,->,rounded corners=5pt]  (y1) -- ++(1,0) -- (7,3.9) -- (0,3.9) -- (s);
		\draw[left,thick,->,rounded corners=5pt]  (y2) -- ++(1,0) -- (7,3.9) -- (0,3.9) -- (s);
		\draw[left,thick,->,rounded corners=5pt]  (y3) -- ++(1,0) -- (7,3.9) -- (0,3.9) -- (s);
		\draw[left,thick,->,rounded corners=5pt]  (y4) -- ++(1,0) -- (7,3.9) -- (0,3.9) -- (s);

  \end{tikzpicture} 
   \infull{
      \caption{Illustration of Reduction~\ref{red:OVtoMDPsafety}, for $S_1=\{(1,0,0), (1,1,1), (0,1,1)\}$ and 
    $S_2=\{(1,1,0), (1,1,1), (0,1,0), (0,0,1)\}$.}
   }
   \inshort{
      \caption{Illustration of Reduction~\ref{red:OVtoMDPsafety}, for $S_1=\{(1,0,0),$ $(1,1,1), (0,1,1)\}$ and 
    $S_2=\{(1,1,0), (1,1,1), (0,1,0),$ $(0,0,1)\}$.}
   }
  \label{fig:OVtoMDPsafety}
  \end{figure}

The number of vertices in the MDP~$\mdp$, constructed by Reduction~\ref{red:OVtoMDPsafety}, is $O(N)$, the number of edges~$m$ is $O(N \log N)$ 
(recall that $d \in O(\log N)$), we have $k \in \Theta(N)$ target sets, and the construction can be performed in 
$O(N \log N)$ time.
Thus, if we would have an $O(m^{2-\epsilon})$ or $O((k\cdot m)^{1-\epsilon})$
algorithm for disjunctive 
queries or disjunctive objectives of safety objectives for any $\epsilon > 0$, we
would immediately get an  $O(N^{2-\epsilon})$ algorithm for OV, which contradicts OVC (and thus SETH).

\section{Algorithms for MDPs with Streett objectives}\label{sec:streett}

In this section we extend algorithms for graphs with Streett objectives to MDPs. 
In particular we prove the following theorem.
\begin{theorem}
	For an MDP~$\mdp$ with Streett objectives defined by Streett pairs 
	$\SP= \{(L_i, U_i) \mid 1 \le i \le k\}$ with 
	$b = \sum_{i=1}^k (\lvert L_i \rvert + \lvert U_i \rvert)$
	the almost-sure winning set can be computed in $O(\min(n^2, m \sqrt{m \log n})
	+ b \log n)$ time.\footnote{It can also be computed in 
	$O((\textsc{MEC} + b) \cdot k)$ time, which is faster for some combinations
	of parameters with $k = O(\log n)$.}
\end{theorem}

We first describe the basic algorithm for MDPs with Streett objectives, which 
uses an algorithm for MEC-decomposition as a black box. We then develop a new 
algorithm that opens up this black box and after an initial computation of the 
MEC-decomposition only uses strongly connected components
and random attractor computations\infull{ (Section~\ref{sec:streettimpr})}. 
This algorithm reveals strong similarities to the known algorithms for graphs 
with Streett objectives. We then extend the two approaches that lead to the best
asymptotic running times on graphs, one for dense 
graphs\infull{ (Section~\ref{sec:streettdense})} 
and one for sparse\infull{ (Section~\ref{sec:streettdense})}
graphs, to MDPs. The algorithms for graphs are based on finding
``good'' strongly connected subgraphs and then determining which vertices can
reach these ``good components''. For MDPs we find \emph{good end-components}
and then compute almost-sure reachability with the union of all good 
end-components as target set
to determine the almost-sure winning set.
We first show that this approach is correct\infull{ 
(Section~\ref{sec:gec}, see also~\cite[Chap.~10.6.3]{baierbook})}
and then provide algorithms that identify all good end-components.

\subsection{Good End-Components}\label{sec:gec}

Good end-components are also useful for other objectives such as Rabin objectives.
The results of this subsection are valid for all objectives for which whether
an infinite path $\pat$ belongs to the objective depends only on the vertices
$\Inf(\pat)$ that occur infinitely often in $\pat$.
For such objectives we show that determining the winning set is equivalent to 
computing almost-sure reachability of the union of all good end-components.
We define a good end-component as an end-component for which the objective
is satisfied if exactly the vertices of the end-component are visited infinitely
often.
\begin{definition}[Good End-Component]\label{def:good_ec}
  Given an MDP $\mdp$ and an objective $\obj$,
  an end-component $\ec$ of $\mdp$ such that
  each path $\pat \in \Path$ with $\Inf(\omega)=\ec$ is in $\obj$
  is called a \emph{good $\obj$~end-component}.
\end{definition}
For a Streett objective the following is an equivalent definition.
\begin{definition}[Good Streett End-Component]
	Given an MDP $\mdp$ and a set $\SP= \{(L_i, U_i) \mid 1 \le i \le k\}$ of
	Streett pairs, 
	a \emph{good Streett end-component} is an end-component $\ec$ of $\mdp$ such that
	for each $1 \le i \le k$ either $L_i \cap \ec = \emptyset$ or
	$U_i \cap \ec \ne \emptyset$.
\end{definition}

The importance of end-components lies in the fact that player~1 can keep the play
in an end-component forever and can visit each vertex in the end-component
almost surely and also almost surely infinitely often
(Lemma~\ref{lem:ec_vist_infinitly}). This implies that in a good end-component
player~1 has an almost-sure winning strategy (Lemma~\ref{lem:wingood}) and thus
player~1 has an almost-sure winning strategy from every vertex that can 
almost-surely reach a good end-component (Lemma~\ref{lem:reachwin} and 
Corollary~\ref{cor:gecsound-gen}). This shows the soundness of the approach
of determining the almost-sure winning set for an objective determined by 
$\Inf(\pat)$ by computing almost-sure reachability of the union of all good
end-components.

\begin{lemma}\label{lem:ec_vist_infinitly}
  Given an MDP $\mdp$ and an end-component~$\ec$, player~1 has a strategy
  from each vertex of $\ec$ such that all vertices of~$\ec$ are 
  almost-surely reached infinitely often
  and only vertices of~$\ec$ are visited. 
\end{lemma}
\begin{proof}
  We define a strategy $\str$ as follows: 
  Choose some arbitrary numbering of the vertices in~$\ec$. The (not memoryless)
  strategy of player~1 is to first follow a shortest path within the end-component
  (with, say, lexicographic tie breaking) to the first vertex from the current 
  position of the play until this vertex is reached, then a shortest path 
  within the end-component to the 
  second vertex and so on, until he starts with the first vertex again. This is
  possible because an end-component is a strongly connected subgraph. 
  Since an end-component has no outgoing random edges, the 
  play does not leave the end-component when player~1 plays this strategy.
  Let $\ell = \lvert \ec \rvert$ and let $\alpha$ be the smallest
  positive transition probability in the MDP. Then the probability that the first 
  chosen shortest path is followed with the above strategy 
  is at least $\alpha^\ell$ and the 
  probability that a sequence of $\ell$ shortest paths within $\ec$ are followed 
  and thus all vertices of $\ec$ are visited is at least $\alpha^{\ell^2}$. Thus 
  the probability that not all 
  vertices in $\ec$ were visited after $q \cdot \ell^2$ steps is at most $(1 - 
  \alpha^{\ell^2})^q$, which goes to 0 when $q$ goes to infinity. Hence player~1
  has a strategy such that all vertices in $\ec$ are visited with probability~1.
  By the same argument all vertices in $\ec$ are visited infinitely often with 
  probability~1 because the probability that some vertex is not visited after 
  some finite prefix of length $t \cdot \ell^2$ can be bounded by $(1 - \alpha^{\ell^2})^{(q - t)}$.
\end{proof}

\begin{lemma}\label{lem:wingood}
  Player~1 has a strategy $\str$ from each vertex in a good $\obj$ end-component~$\ec$
  to satisfy $\obj$ almost-surely.
\end{lemma}
\begin{proof}
  By Lemma~\ref{lem:ec_vist_infinitly} player~1 has a strategy that almost-surely visits
  all nodes in $\ec$ infinitely often. By the definition of good $\obj$ end-component, 
  all paths visiting all nodes in $\ec$ infinitely often are in $\obj$.
  Hence, the strategy given by Lemma~\ref{lem:ec_vist_infinitly} is also almost-sure winning for~$\obj$.
\end{proof}

\begin{lemma}\label{lem:reachwin}
  Given an MDP $\mdp$, an objective $\obj$ that is determined by $\Inf(\pat)$,
  and a set $S$ of almost-sure winning nodes we have that 
  if $v \in \as{\mdp, \reacht{S}}$, then also $v \in \as{\mdp, \obj}$.
\end{lemma}
\begin{proof}
 Assume $v \in \as{\mdp, \reacht{S}}$ and consider the following strategy.
 Start with the strategy for reaching $S$ and as soon as one vertex $s$ of $S$ is reached 
 switch to the almost-sure winning strategy of $s$.
 As $S$ is (almost-surely) reached within a finite number of steps, the vertices
 visited by the strategy for reaching $S$ does not affect the objective $\obj$.
\end{proof}

\begin{corollary}[Soundness of Good End-Components]\label{cor:gecsound-gen}
  For a set of good end-components~$\mathcal{\ec}$ and an objective $\obj$ that is determined by $\Inf(\pat)$
  we have that  $\as{\mdp, \reacht{\bigcup_{\ec \in \mathcal{\ec}}\ec}}$ is contained in $\as{\mdp, \obj}$.
\end{corollary}

Another conclusion we can draw from the above lemmata is that if a MEC contains
a good end-component, then player~1 has an almost-sure winning strategy
for the whole MEC because he can reach the good end-component almost-surely from 
every vertex of the MEC. We exploit this observation in the improved algorithm 
for coBüchi objectives in Section~\ref{sec:cobuchialg}.
\begin{corollary}[of Lemmata~\ref{lem:ec_vist_infinitly} 
and~\ref{lem:reachwin}]\label{cor:winning_mec}
  Given an MDP $\mdp$ and an objective $\obj$ that is determined by $\Inf(\pat)$,
  if a MEC $\ec$ contains an almost-sure
  winning vertex \upbr{e.g.\ a good end-component~$\ec$},
  then all vertices in $\ec$ are almost-sure winning for player~1.
\end{corollary}

To show the completeness of the approach of computing good end-components,
we have to argue that every vertex from which player~1 can satisfy the objective 
almost-surely has also a strategy to reach a good end-component almost-surely.
For this we need two rather technical lemmata.
The intuition behind Lemma~\ref{lem:zeroprob_paths} is that if a random vertex 
occurs infinitely often on a path,
then almost-surely also each of its successors appears infinitely often on that 
path. Thus we can argue that vertex sets that are reached infinitely often 
with positive probability are closed under random edges and hence SCCs within such 
sets of vertices are end-components (Lemma~\ref{lem:random_closure}). 
To show completeness (Proposition~\ref{prop:geccompl-gen})
we then use a set of paths in the objective that are reached with positive 
probability to show that the vertices that these paths use infinitely often 
form good end-components. A similar proof is given for Büchi objectives in
\cite{CourcoubetisY95}. 

\begin{lemma}\label{lem:zeroprob_paths}
  Given an MDP $\mdp$, a strategy $\str$ of player~1,
  the set $\Path_\str$ of infinite paths starting at a vertex $v$ that are compatible with the strategy $\str$, and
  a vertex $a \in V_R$ with $\mathrm{Pr}_\sigma\left(\{\pat \in \Path_\str \mid a \in \Inf(\pat) \}\right)=p$,
  for each successor $b$ of $a$ we have  
  $\mathrm{Pr}_\sigma(\{\pat \in \Path_\str \mid a \in \Inf(\pat),$ $b \in \Inf(\pat) \})=p$ and
  $\mathrm{Pr}_\sigma\left(\{\pat \in \Path_\str \mid a \in \Inf(\pat), b \notin \Inf(\pat) \}\right)=0$.
\end{lemma}
\begin{proof}
  Whenever the strategy visits node $a$, with some constant probability $q$ the play continues in $b$.
  Thus the probability that $b$ was visited less than $\ell$ times after $a$ was visited $n$ times
  is upper bounded by $(1-q^\ell)^{n/\ell}$ which goes to $0$ with increasing $n$.
  Thus, we have $\mathrm{Pr}_\sigma\left(\{\pat \in \Path_\str \mid a \in \Inf(\pat), b \notin \Inf(\pat) \}\right)=0$
  and hence for the complement set $\mathrm{Pr}_\sigma(\{\pat \in \Path_\str \mid a \in \Inf(\pat), b \in \Inf(\pat) \})=p$.
\end{proof}

\begin{lemma}\label{lem:random_closure} 
  Given an MDP $\mdp$, a strategy $\str$ of player~1,
  the set $\Path_\str$ of infinite paths starting at a vertex $v$ that are compatible with the strategy $\str$, 
  a set $\Path' \subseteq \Path_\str$, and the set of vertices 
  $S=\{ a \mid \mathrm{Pr}_\sigma\left( \{ \pat \mid a \in \Inf(\pat), \pat \in \Path' \} \right) > 0\}$, then
  for each SCC $C$ of $S$ and each vertex $a \in C \cap V_R$ all successors of a are contained in $C$, i.e., $C$ is an end-component of $\mdp$.
\end{lemma}
\begin{proof}
 Consider an SCC $C$, a vertex $a \in C \cap V_R$, and a successor $b$.
 Then by definition $\mathrm{Pr}_\sigma\left( \{ \pat \mid a \in \Inf(\pat), \pat \in \Path'\} \right) = p$ 
 for a $p>0$ and by Lemma~\ref{lem:zeroprob_paths} 
 we get $\mathrm{Pr}_\sigma( \{ \pat \mid a \in \Inf(\pat),$ $b \notin \Inf(\pat), \pat \in \Path_\str \})=0$
 and thus $\mathrm{Pr}_\sigma ( \{ \pat \mid a \in \Inf(\pat), b \in \Inf(\pat),$ $\pat \in \Path'\}) = p$,
 i.e., $b \in S$.
 For each of the paths $\pat$ in the latter set we have a path from $b$ to $a$ consisting solely of nodes in $\Inf(\pat)$.
 As in $\mdp$ there are just finitely many paths from $b$ to $a$ at least one must have non-zero probability 
 and thus is also contained in $S$.  Hence, $b$ belongs to the SCC $C$.
\end{proof}

\begin{proposition}[Completeness of Good End-Components]\label{prop:geccompl-gen}
    Given an MDP $\mdp$ with an objective $\obj$ determined by $\Inf(\pat)$
    and let $\mathcal{\ec}$ be the set of all good $\obj$ end-components,
    then $\as{\mdp, \obj}$ is contained in  $\as{\mdp, \reacht{\cup_{\ec \in \mathcal{\ec}} \ec}}$.
\end{proposition}
\begin{proof}
  For a vertex $v \in \as{\mdp,\obj}$,
  fix a strategy $\str$ of player~1 such that the objective is satisfied almost-surely. 
  Let $\mdp_\str$ be the sub-MDP of $\mdp$ that consists of the vertices that 
  are visited infinitely often with non-zero probability when player~1 follows strategy~$\str$.
  Note that by Lemma~\ref{lem:random_closure} each SCC of $\mdp_\str$ is an 
  end-component of $\mdp$.
  Moreover, $\str$ is a strategy for almost-surely reaching $\mdp_\str$ 
  (each infinite path has to visit at least one vertex infinitely often).
  
  It remains to show that each vertex of $\mdp_\str$ can almost-surely reach 
  a good end component. 
  We will actually show that each vertex of $\mdp_\str$ is already contained in a good end component.
  To this end let $\Path_\str$ be the set of infinite paths starting at~$v$
  that are compatible with the strategy $\str$ and satisfy the objective.
  For an arbitrary node $u$ of $\mdp_\str$ we consider all paths 
  $\pat \in \Path_\str$ with $u \in \Inf(\pat)$ and group them by $\Inf(\pat)$.
  At least one of these groups has non-zero probability, 
  as there are only finitely many possible sets $\Inf(\pat)$ and $u \in \Inf(\pat)$ has non-zero probability.  
  Let us consider one of the groups of paths $\Path^S_\str$ with non-zero probability 
  and the corresponding set $S=\Inf(\pat)$ for $\pat\in \Path^S_\str$.
  By Lemma~\ref{lem:random_closure} the set $S$ is closed under random edges.
  Moreover, as in each path $\pat \in \Path_\str$ the vertices $\Inf(\pat)$ are strongly connected, the
  set $S$ is also strongly connected and thus an end-component.
  Finally, as the paths $\pat \in \Path^S_\str$ satisfy the objective
  and the objective $\obj$ is determined by $\Inf(\pat)=S$, the set $S$
  forms a good end component.
  Hence, we have shown that each vertex of $\mdp_\str$ is contained in a good $\obj$ 
  end-component, which completes the proof.
\end{proof}

\subsection{Algorithm Preliminaries}\label{sec:algprelim}
We introduce some additional notation for the algorithms for MDPs with Streett 
and Rabin objectives.
For a set $\RP= \{(L_i, U_i) \mid 1 \le i \le k\}$
of Rabin pairs or a set $\SP= \{(L_i, U_i) \mid 1 \le i \le k\}$, 
let $b = \sum_{i=1}^k (\lvert L_i \rvert + \lvert U_i \rvert)$.
A \emph{strongly connected component} \upbr{SCC} is a \emph{maximal} strongly 
connected subgraph. A single vertex is considered strongly connected. An SCC without
outgoing edges is a \emph{bottom SCC}, one without incoming edges a \emph{top SCC}.
The \emph{reverse graph} $\rev$ is constructed by reversing the direction of 
all edges of the graph $G$. In a graph $G = (V, E)$ the set of vertices $E(v)$ 
for some vertex $v$ denotes the set of vertices $w \in V$ for which $(v, w) \in E$.
The out-degree of $v \in V$ in $G$ is denoted with $\OutDeg_H(v)$, its in-degree
with $\InDeg_H(v)$. Let \textsc{MEC} denote the runtime to compute the maximal 
end-component decomposition of an MDP; we assume $\textsc{MEC} = \Omega(m)$.
Further we assume that each vertex in the input MDP 
has at least one outgoing edge, and thus we have $m = \Omega(n)$.

\begin{definition}[Random Attractor]\label{def:attr}
In an MDP $\mdp = ((V, E), (\vo, \vr), \trans)$ 
the \emph{random attractor} $\at(\mdp, W)$ of a set of vertices 
$W \subseteq V$ is defined as $\at(\mdp, W) = \bigcup_{j \ge 0} Z_j$ where $Z_0 = W$ and 
$Z_j$ for $j > 0$ is defined recursively as $Z_{j+1} = Z_j \cup \set{v \in \vr 
\mid E(v) \cap Z_j \ne \emptyset} \cup \set{v \in \vo \mid E(v) \subseteq Z_j}$.
The random attractor $\at(\mdp, W)$ can be computed in $O(\sum_{v \in 
\at(\mdp, W)} \InDeg(v))$ time~\cite{Beeri80,Immerman81}.
\end{definition}

All the algorithms for Streett objectives maintain vertex sets that are 
candidates for good end-components. For such a vertex set~$S$ we (a) 
refine the maintained sets according to the SCC decomposition of $\mdp[S]$
and (b) for a set of vertices~$W$ for which we know that it cannot be contained 
in a good end-component, we remove its random attractor from $S$. The following lemma 
shows the correctness of these operations.

\begin{lemma}\label{lem:eccontained}
	Given an MDP $\mdp = ((V, E), (\vo, \vr), \trans)$, let $\ec$ be an end-component with $X \subseteq S$ for 
	some $S \subseteq V$. 
	We have 
	
	\begin{itemize}
	 \item[\upbr{a}]$\ec \subseteq \scc$ for one SCC~$\scc$ of $\mdp[S]$ and
	 
	 \item[\upbr{b}] $\ec \subseteq S \setminus \at(\mdp', W) = \emptyset$ for each $W \subseteq V \setminus \ec$
			 and each sub-MDP~$\mdp'$ containing~$\ec$.
	\end{itemize}
\end{lemma}

\begin{proof}
	Property \upbr{a} holds since every end-component induces a strongly connected
	sub-MDP. We prove Property \upbr{b} by showing that $\at(\mdp', W)$
	does not contain a vertex of $\ec$ by induction over the recursive 
	definition of a random attractor. Let the sets $Z_j$ be as in
	Definition~\ref{def:attr} and let $E'(v)$ be the vertices to which $v$ has an 
	edge in~$\mdp'$. 
	We have $Z_0 = W$ and thus $Z_0 \cap \ec = \emptyset$.
	Assume we have $Z_j \cap \ec = \emptyset$ for some $j \ge 0$. No vertex of
	$\vr \cap \ec$ has an outgoing edge to $V \setminus \ec$ and thus the set 
	$\ec \cap \set{v \in \vr \mid E'(v) \cap Z_j \ne \emptyset}$ is empty.
	Further every vertex in $\vo \cap \ec$ has an outgoing edge to a vertex in $\ec$.
	Hence also $\ec \cap \set{v \in \vo \mid E'(v) \subseteq Z_j}$ is empty
	and we have that $Z_{j+1} \cap \ec = \emptyset$.
\end{proof}

Let $\ec$ be a good Streett end-component. Then $\ec \cap U_i = \emptyset$ implies 
$\ec \cap L_i = \emptyset$. Thus if $S \cap U_i = \emptyset$ for some vertex 
set $S$ and some index $i$, then we have $U_i \subseteq V \setminus \ec$ 
for each end-component~$\ec \subseteq S$. Hence we obtain the 
following corollary.

\begin{corollary}\label{cor:geccontained}
Given an MDP $\mdp$, let $\ec$ be a \emph{good} Streett end-component with 
$X \subseteq S$ for some $S \subseteq V$. 
For each $i$ with $S \cap U_i = \emptyset$ it holds that 
$\ec \subseteq S \setminus \at(\mdp[S],  L_i \cap S)$.
\end{corollary}

\subsection{Improving Upon the Basic Algorithm}\label{sec:streettimpr}

In Algorithm~\ref{alg:streettbasic}, the basic algorithm for MDPs with Streett 
objectives, we maintain a set of already identified 
(maximal) good end-components~$\good$, which is initially empty, and a set of 
candidate end-components~$\mathcal{\ec}$, which is initialized with the 
MECs of the input MDP~$\mdp$. In each iteration of the while-loop we remove 
an end-component $\ec$ from $\mathcal{\ec}$ and check whether it is a 
good end-component. For this check we find sets $U_i$ for $1 \le i \le k$ 
that do not intersect with $\ec$ and identify vertices in $\ec \cap L_i$ for 
such an $i$ as ``bad vertices''~$\badv$. If there are no bad vertices, then
$\ec$ is a good end-component and added to $\good$. Otherwise the bad 
vertices and their random attractor within $\ec$ are removed from $\ec$.
On the sub-MDP induced by the remaining vertices of $\ec$ we compute the 
MEC-decomposition, which identifies all remaining candidate end-components among
the vertices of $\ec$. The new candidates are then added to $\mathcal{\ec}$.
If the algorithm finds good end-components, it returns the almost-sure winning set
for the reachability of the union of them.

\begin{algorithm}
	\SetAlgoRefName{StreettMDPbasic}
	\caption{Basic Algorithm for MDPs with Streett Objectives}
	\label{alg:streettbasic}
	\SetKwInOut{Input}{Input}
	\SetKwInOut{Output}{Output}
	\BlankLine
	\Input{an MDP $\mdp = ((V, E), (\vo, \vr), \trans)$ and Streett pairs 
	$\SP= \{(L_i, U_i) \mid 1 \le i \le k\}$
	}
	\Output
	{
	$\as{\mdp, \streett{\SP}}$
	}
	\BlankLine
	$\good \gets \emptyset$\;
	$\mathcal{\ec} \gets \mecalg(\mdp)$\;
	\While{$\mathcal{\ec} \ne \emptyset$}{
		remove some $\ec \in \mathcal{\ec}$ from $\mathcal{\ec}$\;
		$\badv \gets \set{\inec \in \ec \mid \exists i \text{ s.t.\ } 
		\inec \in L_i \text{ and } \ec \cap U_i = \emptyset}$\;
		\If{$\badv \ne \emptyset$}{
			$\ec \gets \ec \setminus \at(\mdp[\ec], \badv)$\label{lbasic:remove}\;
			$\mathcal{\ec} \gets \mathcal{\ec} \cup \mecalg(\mdp[\ec])$\label{lbasic:mec}\;
		}\Else{
			$\good \gets \good \cup \set{\ec}$\;
		}
	}
	\Return{$\as{\mdp, \reacht{\bigcup_{\ec \in \good} \ec}}$}\;
\end{algorithm}

\begin{proposition}[Runtime of 
Algorithm~\ref{alg:streettbasic}]\label{prop:timestreetbasic}
	Algorithm~\ref{alg:streettbasic} can be implemented 
	to run in $O((\textsc{MEC} + b) \min(n, k))$ time.
\end{proposition}

\begin{proof}
	The initialization of $\mathcal{\ec}$ with all MECs of the input
	MDP $\mdp$ can clearly be done in $O(\textsc{MEC})$ time. Further by 
	Theorem~\ref{th:timedrmdp} the almost-sure reachability computation
	after the while-loop can be done in $O(\textsc{MEC})$ time. 
	
	Let $\ec_v$ denote the end-component of $\mathcal{\ec}$ currently containing 
	an arbitrary, fixed vertex $v \in V$ during Algorithm~\ref{alg:streettbasic}. 
	In each iteration
	of the while-loop in which $\ec_v$ is considered either (a) $\badv = \emptyset$
	and $\ec_v$ will not be considered further or (b) the number of vertices
	in $\ec_v$ is reduced by at least one and we have for some $1 \le i \le k$ that
	$\ec_v \cap L_i \ne \emptyset$ before the iteration of the while-loop and
	$\ec_v \cap L_i = \emptyset$ after the while-loop. Thus each vertex and 
	each edge of the MDP $\mdp$ is considered in at most $O(\min(n, k))$ iterations
	of the while-loop.
	
	Consider the $j$th iteration of the while-loop; let $\ec_j$ denote the set 
	removed from $\mathcal{\ec}$ in this iteration and let $\bits(\ec_j) = \sum_{i=1}^k 
	(\lvert L_i \cap \ec_j \rvert + \lvert U_i \cap \ec_j \rvert)$. 
	Assume that each vertex has a list of the sets $L_i$ and $U_i$ for 
	$1 \le i \le k$ it belongs to.
	(We can generate these lists from the lists of the Streett pairs in $O(b)$
	time at the beginning of the algorithm.)
	Then we can determine $\badv$ by going through all lists of the vertices 
	in $\ec_j$ in $O(\lvert \ec_j \rvert + b_j)$ time, which amounts to 
	$O((n + b) \min(n, k))$ total time over all iterations of the while-loop.
	The random attractor computed in Line~\ref{lbasic:remove} is removed and 
	not considered further, thus its computation takes $O(m)$ time over the whole 
	algorithm (see Definition~\ref{def:attr}). The computation
	of all MECs in $\mdp[\ec_j]$ takes total time $O(\textsc{MEC} \cdot \min(n, k))$
	over all iterations of the while loop. Thus the whole algorithm can be 
	implemented in $O((\textsc{MEC} + b) \min(n, k))$ total time.
\end{proof}

\begin{proposition}[Soundness of Algorithm~\ref{alg:streettbasic}]
	Let $W$ be the set returned by Algorithm~\ref{alg:streettbasic}.
	We have $W \subseteq \as{\mdp, \streett{\SP}}$.
\end{proposition}

\begin{proof}
By Corollary~\ref{cor:gecsound-gen} it is sufficient to show that every set 
$\ec \in \good$ is a good end-component. The algorithm explicitly
checks immediately before $\ec$ is added to $\good$ that we have for each 
$1 \le i \le k$ either $L_i \cap \ec = \emptyset$ or $U_i \cap \ec \ne
\emptyset$. Thus it only remains
to show that $\ec$ is an end-component when it is added to $\good$. Before 
a set is added to $\good$, the same set is contained in the set $\mathcal{\ec}$. 
We show that all sets in $\mathcal{\ec}$ are end-components at any point in 
the algorithm by induction over the iterations of the while-loop in the algorithm. 
Before the first iteration of the while-loop the sets $\ec \in \mathcal{\ec}$ are
the maximal end-components of $\mdp$. Now consider an iteration in which
a set $\ec$ is removed from $\mathcal{\ec}$ and new sets are added to 
$\mathcal{\ec}$. First, some vertices and their random attractor in the 
sub-MDP $\mdp[\ec]$ induced by $\ec$ are removed from $\ec$. Let $\ec'$ be the 
remaining set of vertices. By the definition of a random attractor there are no 
random edges from $\ec'$ to the removed random attractor. 
Further, by the induction hypothesis there are no random edges from $\ec$ to $V 
\setminus \ec$. Thus there are no random edges from $\ec'$ to $V \setminus \ec'$.
Then the algorithm adds the MECs of the sub-MDP $\mdp[\ec']$ to $\mathcal{\ec}$.
Let $\hat{\ec}$ be one such MEC. Since $\hat{\ec}$ is a MEC 
in $\mdp[\ec']$, it is a MEC in $\mdp$ if and only if it has no random edges 
from $\hat{\ec}$ to $V \setminus \ec'$. This holds by $\hat{\ec} \subseteq \ec'$ 
and the properties of $\ec'$ established above.
\end{proof}

\begin{proposition}[Completeness of Algorithm~\ref{alg:streettbasic}]\label{prop:basiccompl}
		Let $W$ be the set returned by Algorithm~\ref{alg:streettbasic}.
	We have $\as{\mdp, \streett{\SP}} \subseteq W$.
\end{proposition}

\begin{proof}
	By Proposition~\ref{prop:geccompl-gen} it is sufficient to show that at the end of 
	Algorithm~\ref{alg:streettbasic} the union of the sets in $\good$ contains
	all good end-components of the MDP $\mdp$. We show by induction that 
	every good end-component is a subset of either 
	$\good$ or $\mathcal{\ec}$ before and after each iteration of the while-loop
	in Algorithm~\ref{alg:streettbasic}; as $\mathcal{\ec}$ is empty at the 
	end of the algorithm, this implies the claim.
	
	Before the first iteration of the while-loop, the 
	set $\mathcal{\ec}$ is initialized with the MECs of $\mdp$, thus the induction 
	base holds. Let $\ec$ be the set of vertices removed from $\mathcal{\ec}$ in 
	an iteration of the while-loop and let ${\ec}^*$ be the union of the 
	good end-components contained in $\ec$. Either $\ec$ is added to $\good$ or 
	we have that for some indices~$i$ the set $\ec$ contains vertices of 
	$L_i$ but not of $U_i$; then for these indices the sets $L_i$ and their 
	random attractor are removed from $\ec$. 
	Let $\hat\ec$ be this the updated set, i.e., $\hat\ec= \ec \setminus \at(\mdp[\ec], \badv)$ .
	By Corollary~\ref{cor:geccontained} 
	we still have ${\ec}^* \subseteq \hat{\ec}$ after this step. Then all 
	MECs of $\mdp[\hat{\ec}]$ 
	are added to $\mathcal{\ec}$.
	Every good end-component contained in $\hat{\ec}$ is completely contained 
	in one MEC of $\mdp[\hat{\ec}]$, thus the claim continues to hold after the iteration
	of the while-loop.
\end{proof}

The essential observation towards faster algorithms for MDPs with Streett objectives
is the following.
Consider a set $\ec$ in an iteration of the basic algorithm after 
some vertices in $\at(\mdp[\ec], \badv)$ were removed.
We have that there are no random edges from $\ec$ to the remaining vertices
in the graph and further we have for each $1 \le i \le k$ either $L_i \cap 
\ec = \emptyset$ or $U_i \cap \ec \ne \emptyset$. Thus if $\mdp[\ec]$ is still 
strongly connected, then $\ec$ is a good
end-component and is added to $\good$ in one of the subsequent iterations
of the algorithm. If, however, the sub-MDP $\mdp[\ec]$ consists of multiple SCCs, 
then we have that the bottom SCCs of $\mdp[\ec]$ are end-components in $\mdp$ 
but the remaining SCCs of $\mdp[\ec]$ might have outgoing random edges within
$\mdp[\ec]$. Note, however, that we have for any good end-component $\hat{\ec}$
in $\mdp[\ec]$ and any SCC $\scc$ of $\mdp[\ec]$ that either $\hat{\ec} \subseteq \scc$ 
or $\hat{\ec} \cap \scc = \emptyset$, simply by the fact that every good end-component
is strongly connected (Lemma~\ref{lem:eccontained}~\upbr{a}). Let $\hat{\ec} \subseteq \scc$ and let 
$R$ be the random vertices of $\scc$ with edges to 
vertices not in $\scc$. Then the vertices in $R$ cannot intersect with $\hat{\ec}$
because an end-component has no outgoing random edges. Further, also the 
random attractor of $R$ cannot intersect with $\hat{\ec}$ (Lemma~\ref{lem:eccontained}~\upbr{b}). Thus we can 
remove $\at(\mdp[\ec],
R)$ from $\mdp[\ec]$ and all good end-components that were contained
in $\mdp[\ec]$ are still contained in the remaining sub-MDP. However, now the 
set of vertices in $\scc \setminus \at(\mdp[\ec], R)$ has no outgoing random edges.
Thus if it is still strongly connected, then it is an end-component.
With this observation we can avoid computing a MEC decomposition in the 
while-loop of the basic algorithm and instead only compute strongly connected 
components and random attractors, which both can be done in linear time.
Note that in the improved algorithm we do not have the property that every 
maintained set of vertices is an end-component (as in the basic algorithm) but
still none of the maintained sets has outgoing random edges. 

In this formulation the algorithm for MDPs with Streett objectives has a very 
similar structure to the algorithm for graphs with Streett objectives: 
We repeatedly remove ``bad vertices'' and recompute strongly connected components.
The main difference is that we additionally compute random attractors.
Based on this, we can indeed show that for Streett objectives 
the same techniques as for graphs also apply to MDPs and by this improve the 
runtime to the runtime for graphs plus the time to compute one MEC decomposition.
This can be seen as opening up the ``black-box''
use of a MEC-decomposition algorithm and combining the fastest algorithms for 
MEC-decomposition~\cite{ChatterjeeH11,ChatterjeeH14} and graphs with Streett 
objectives~\cite{HenzingerT96,ChatterjeeHL15}.
In contrast to graphs with Streett objectives, no $O((m+b)k)$
algorithm can be achieved for small values of $k$. Intuitively, this is because
it could be that only in a few iterations bad vertices are removed while
the majority of the iterations is actually used to recompute MECs.
We present the new algorithmic ideas for MDPs with Streett objectives in 
Algorithm~\ref{alg:streettimpr} (which is only faster for large enough $k$)
and then apply the known techniques for sparse and dense graphs in 
Algorithms~\ref{alg:streettsparse} and~\ref{alg:streettdense},
respectively, to beat the basic algorithm for all parameters except very small 
values of $k$; the basic algorithm is faster for e.g.\ $k = O(1)$ or $k = O(\sqrt{\log n})$ and $m = O(n^{4/3})$.

In our improved algorithms we use the data structure $\ds(\ec)$
from~\cite{HenzingerT96} to quickly identify and remove vertices in 
$\ec \cap L_i$ for which $\ec \cap U_i = \emptyset$ from a set of vertices $\ec$.
\begin{lemma}[\cite{HenzingerT96}]\label{lem:ds}
After a one-time preprocessing time of $O(k)$, there is a data structure 
$\ds(\ec)$ for a given set $\ec$ that can be initialized with the operation 
$\construct(\ec)$ in time $O(\bits(\ec) + \lvert \ec \rvert)$, where 
$\bits(\ec) = \sum_{i = 1}^k
\left( \lvert \ec \cap L_i \rvert + \lvert \ec \cap U_i \rvert \right)$. 
Further it supports the operation $\remove(\ec, \ds(\ec), B)$ that removes 
a set $B \subseteq V$ from $\ec$ and updates $\ds(\ec)$ accordingly in 
time $O(\bits(B) + \lvert B \rvert)$ and the operation $\bad(\ds(\ec))$
that returns a pointer to 
the set $\set{\inec \in \ec \mid \exists i \text{ s.t.\ } \inec \in L_i 
\text{ and } \ec \cap U_i = \emptyset}$ in constant time.
\end{lemma}

In Algorithm~\ref{alg:streettimpr} we maintain a list~$Q$ of data structures
of disjoint vertex sets that are candidates for good end-components. 
For every set $S$ with $\ds(S)$ in $Q$ we maintain that there are no random edges 
from $S$ to $V\setminus S$. The list~$Q$
is initialized with the data structures of all MECs of the input MDP~$\mdp$. 
In each iteration
of the outer while-loop the data structure of one vertex set~$S$ is pulled 
from~$Q$. In the inner while-loop the set of ``bad vertices'' 
$\set{\inec \in S \mid \exists i \text{ s.t.\ } \inec \in L_i 
\text{ and } S \cap U_i = \emptyset}$ is identified and its random attractor
is removed from $S$ and $\ds(S)$. Through removing the random attractor we 
maintain the property that there are no random edges from $S$ to $V \setminus S$
at this step. Thus we have that if $\mdp[S]$ is (still) strongly connected, then $\mdp[S]$ is a good end-component, which 
we identify in Line~\ref{limpr:good}.
If $\mdp[S]$ does not contain an edge, we do not have to consider it further.
If it contains an edge but is not strongly connected, the SCCs of $\mdp[S]$
are identified. For each SCC~$\scc$ we identify its random vertices that have edges 
to vertices of $S \setminus \scc$ and remove their random attractor from $\scc$.
After this step the data structure of the remaining vertices of $\scc$ is added to $Q$.
At this point we distinguish between the largest SCC and the other SCCs of $\mdp[S]$.
We construct a new data structure for all but the largest SCC and reuse the 
data structure of $S$ for the largest SCC. This improves the runtime because
we only spend time proportional to the smaller SCCs and a vertex can be in a 
smaller SCC at most $O(\log n)$ times. Note that at this point of the algorithm
the sub-MDP $\mdp[C]$ is not necessarily strongly connected since vertices were
removed after the SCC computation
but we maintain the property that there are no random edges from a vertex
set for which the data structure is in $Q$ to other vertices.
When the list~$Q$ becomes empty, the algorithm terminates. If good end-components
were identified, the almost-sure winning set for the reachability objective
of the union of the good end-components is output.

\begin{algorithm}
	\SetAlgoRefName{StreettMDPimpr}
	\caption{New Algorithm for MDPs with Streett Objectives}
	\label{alg:streettimpr}
	\SetKwInOut{Input}{Input}
	\SetKwInOut{Output}{Output}
	\BlankLine
	\Input{an MDP $\mdp = ((V, E), (\vo, \vr), \trans)$ and Streett pairs 
	$\SP= \{(L_i, U_i) \mid 1 \le i \le k\}$
	}
	\Output
	{
	$\as{\mdp, \streett{\SP}}$
	}
	\BlankLine
	$\good \gets \emptyset$; $Q \gets \emptyset$\;
	$\mathcal{\ec} \gets \mecalg(\mdp)$\;
	\lFor{$\ec \in \mathcal{\ec}$}{
		$Q \gets Q \cup \set{\construct(\ec)}$
	}
	\While{$Q \ne \emptyset$}{
		remove some $\ds(S)$ from $Q$\;
		\While{$\bad(\ds(S)) \ne \emptyset$\label{limpr:badfind}}{
			$A \gets \at(\mdp[S], \bad(\ds(S)))$\label{limpr:badattr}\;
			$(S, \ds(S)) \gets \remove(S, \ds(S), A)$\label{limpr:badrem}\;
		}
		\If{$\mdp[S]$ contains at least one edge}{
			\If{$\mdp[S]$ is strongly connected\label{limpr:testsc}}{
				$\good \gets \good \cup \set{S}$\label{limpr:good}\;
			}\Else{
				$\mathcal{\scc} \gets \allsccalg(\mdp[S])$\label{limpr:scc};
				$S' \gets S$\;
				\For{$\scc \in \mathcal{\scc}$}{
					$R \gets \set{v \in \vr \cap \scc \mid \exists w \in S' \setminus \scc 
					\text{ s.t.\ } (v, w) \in E}$\label{limpr:randout}\;
					$A \gets \at(\mdp[\scc], R)$\label{limpr:randoutattr}\;
					\If{$\scc$ is largest SCC in $\mathcal{\scc}$}{
						$(S, \ds(S)) \gets \remove(S, \ds(S), A)$\label{limpr:randoutrem}\;
					}\Else{
						$(S, \ds(S)) \gets \remove(S, \ds(S), \scc)$\label{limpr:smallrem}\;
						$\scc \gets \scc \setminus A$\;
						$Q \gets Q \cup \set{\construct(\scc)}$\label{limpr:smallconstr}\;
					}
				}
				$Q \gets Q \cup \set{\ds(S)}$\label{limpr:addback}\;
			}
		}
	}
	\Return{$\as{\mdp, \reacht{\bigcup_{\ec \in \good} \ec}}$}\;
\end{algorithm}

\begin{proposition}[Runtime Algorithm~\ref{alg:streettimpr}]
	Algorithm~\ref{alg:streettimpr} terminates in $O(mn + b \log n)$ time.
\end{proposition}

\begin{proof}
	Using the data structure of Lemma~\ref{lem:ds} (\cite{HenzingerT96}),
	the initialization phase of Algorithm~\ref{alg:streettimpr} takes
	$O(k + \textsc{MEC} + b + n)$ time, which is in $O(mn + b)$. Further by 
	Theorem~\ref{th:timedrmdp} the almost-sure reachability computation
	after the outer while-loop can be done in $O(\textsc{MEC})$ time. 
	
	Whenever bad vertices and their random attractor are identified in 
	lines~\ref{limpr:badfind}--\ref{limpr:badattr}, they are removed in 
	Line~\ref{limpr:badrem} and not considered further. Thus finding bad 
	vertices takes total time~$O(n)$,
	identifying the random attractor of bad vertices takes total time~$O(m)$ (see 
	Definition~\ref{def:attr}), and removing the bad vertices and their attractor 
	takes total time~$O(m + b)$ by Lemma~\ref{lem:ds}.
	
	After the initialization of $Q$ with the MECs of $\mdp$, all vertex sets 
	for which a data structure is stored in $Q$ induce a strongly connected sub-MDP.
	Consider the set $S$ when Line~\ref{limpr:scc} is reached and its smallest superset
	$S' \supseteq S$ that was identified as strongly connected in the algorithm
	(i.e.\ $S'$ is either a MEC of $\mdp$ or an SCC computed in Line~\ref{limpr:scc}
	in a previous iteration of the algorithm). We have that $S$ is a proper 
	subset of $S'$, i.e., either bad vertices were 
	removed from $S'$ in Line~\ref{limpr:badrem} or a non-empty set of random
	vertices was identified in Line~\ref{limpr:randout}. Hence any part of $\mdp$
	is considered in at most $n$ iterations of the outer while-loop. This implies
	that we can bound the total time spent in 
	lines~\ref{limpr:scc}--\ref{limpr:randoutattr} with $O(mn)$.
	
	By the same argument as for the removal of bad vertices and their attractor,
	the calls to $\remove$ in Line~\ref{limpr:randoutrem} take total time $O(n + b)$.
	It remains to bound the time for the calls to $\remove$ and $\construct$
	in lines~\ref{limpr:smallrem}--\ref{limpr:smallconstr}. Note that we avoid to
	make these calls for the largest of the SCCs of the sub-MDP induced by $S$,
	which are computed in Line~\ref{limpr:scc}.
	Thus whenever we call $\remove$ and $\construct$ for an SCC~$\scc$, we have 
	$\lvert \scc \rvert \le \lvert S \rvert /2$. Hence 
	we can charge the time for $\remove$ and $\construct$ to the vertices of $\scc$
	and to $\bits(\scc)$ such that every vertex $v$ and every $\bits(\set{v})$ is 
	charged $O(\log n)$ times. Thus we can bound the time for 
	lines~\ref{limpr:smallrem}--\ref{limpr:smallconstr} with $O((n + b) \log n)$.
	This proves the claimed runtime.
\end{proof}

\begin{proposition}[Soundness of Algorithm~\ref{alg:streettimpr}]\label{prop:streettimprsound}
	Let $W$ be the set returned by Algorithm~\ref{alg:streettimpr}.
	We have $W \subseteq \as{\mdp, \streett{\SP}}$.
\end{proposition}

\begin{proof}
By Corollary~\ref{cor:gecsound-gen} it is sufficient to show that every set 
$\ec \in \good$ is a good end-component. The algorithm explicitly
checks immediately before $\ec$ is added to $\good$ in Line~\ref{limpr:good}
that $\ec$ contains at least one edge and is strongly connected. Further 
we have by the termination condition of the inner while-loop that for each 
$1 \le i \le k$ either $L_i \cap \ec = \emptyset$ or $U_i \cap \ec \ne
\emptyset$. Thus it remains to show that there are no random edges from $\ec$
to $V \setminus \ec$. 

Let $\ec'$ be the set of vertices for which the data 
structure $\ds(\ec')$ was removed from $Q$ in the iteration of the outer while-loop
in which $\ec$ was added to $\good$. By the following invariant there are no random 
edges from $\ec'$ to $V \setminus \ec'$. 
\begin{invariant}\label{inv:streettnorand}
	For every set $S$ for which the data structure $\ds(S)$ is in $Q$ there 
	are no random edges from $S$ to $V \setminus S$.
\end{invariant}
Assume the invariant holds. If $\ec'$ is not equal to $\ec$, then
some vertices and their random attractor within $\mdp[\ec']$ were removed in the 
inner while-loop. By the definition of a random attractor there are no random 
edges from $\ec$ to $\ec' \setminus \ec$ and thus to $V \setminus \ec$. 

It 
remains to prove the invariant by induction over the iterations of the outer
while-loop.
Before the first iteration of the while-loop $Q$ is initialized with
the maximal end-components of $\mdp$ and thus the invariant holds. 
Assume the invariant holds before the beginning of an iteration of the outer while-loop
and let $S$ be the set of vertices for which the data structure 
is removed from $Q$ in this iteration. In the inner while-loop some vertices 
and their random attractor in $\mdp[S]$ might be removed from $S$. Let $S'$ be 
the remaining vertices. By the definition of a random attractor there are no 
random edges from $S'$ to $S \setminus S'$ and thus by the induction hypothesis
there are no random edges from $S'$ to $V \setminus S'$.

If $\mdp[S']$ is strongly connected, then no set is added to $Q$ in this iteration
of the while-loop. 
Otherwise the SCCs $\mathcal{\scc}$ of $\mdp[S']$ are considered as candidates to be 
added to $Q$. For each set $\scc \in \mathcal{\scc}$ the random vertices~$R$ in $\scc$
with edges to vertices in $S' \setminus \scc$ and their random attractor~$A$
in $\mdp[\scc]$ are removed from $\scc$. Let $\scc'$ be the remaining vertices.
We have that there are no random edges from $\scc'$ to $S' \setminus \scc$ by the
definition of~$R$ and that there are no random edges from $\scc'$ to $\scc \setminus 
\scc'$ by the definition of~$A$. Thus there are no random 
edges from $\scc'$ to $V \setminus \scc'$ for any set $\scc'$ for which the 
data structure is added to $Q$, which shows the invariant.
\end{proof}

\begin{proposition}[Completeness of Algorithm~\ref{alg:streettimpr}]\label{prop:streettimprcompl}
		Let $W$ be the set returned by Algorithm~\ref{alg:streettimpr}.
	We have $\as{\mdp, \streett{\SP}} \subseteq W$.
\end{proposition}

\begin{proof}
	By Proposition~\ref{prop:geccompl-gen} it is sufficient to show that at the end of 
	the algorithm the union of the sets in $\good$ contains
	all good end-components of the MDP $\mdp$. We show the following invariant 
	by induction over the iterations of the outer while-loop;
	as $Q$ is empty at the end of the algorithm, this implies the claim.
	\begin{invariant}\label{inv:geccontained}
		For each good end-component~$\ec$ of $\mdp$ and some set $Y \supseteq \ec$
		either $Y \in \good$ or $\ds(Y) \in Q$ holds before and after each iteration 
		of the outer while-loop.
	\end{invariant}

	Before the first iteration of the outer while-loop, the 
	set $Q$ is initialized with the MECs of $\mdp$, thus the induction base holds.
	Let $S$ be the set of vertices for which the data structure is 
	removed from $Q$ in an iteration of the outer while-loop and 
	let $\mathcal{\ec}_S$ be the set of good end-components 
	contained in $S$. We have $\ec \subseteq S'$ for every 
	$\ec \in \mathcal{\ec}_S$ after the inner while-loop by Corollary~\ref{cor:geccontained}.
	
	Since every end-component contains an edge, $\mdp[S']$ contains at least one
	edge if $\mathcal{\ec}_S$ is not empty. Then either $S'$ and thus all $\ec \in 
	\mathcal{\ec}_S$ are added to $\good$ or the SCCs $\mathcal{\scc}$ of $\mdp[S']$ 
	are computed. For each $\ec \in \mathcal{\ec}_S$ there exists $\scc \in 
	\mathcal{\scc}$ such that $\ec \subseteq \scc$ by Lemma~\ref{lem:eccontained}~\upbr{a}; 
	let $\ec$ and $\scc$ be such that $\ec \subseteq \scc$. Since $\ec$ has not 
	outgoing random edges, we have $R \cap \ec = \emptyset$ (Line~\ref{limpr:randout})
	and thus also $\ec \subseteq \scc \setminus \at(\mdp[\scc],R)$ by Lemma~\ref{lem:eccontained}~\upbr{b}. The data structure of $\scc \setminus A$
	is added to $Q$ in lines~\ref{limpr:smallconstr} or~\ref{limpr:addback},
	hence the claim holds after the outer while-loop.
\end{proof}

\subsection{Algorithm for Dense MDPs with Streett Objectives}\label{sec:streettdense}

\begin{algorithm}
	\SetAlgoRefName{StreettMDPdense}
	\caption{Algorithm for dense MDPs with Streett Objectives}
	\label{alg:streettdense}
	\SetKwInOut{Input}{Input}
	\SetKwInOut{Output}{Output}
	\SetKwData{found}{found}
	\SetKwData{true}{true}
	\SetKwData{false}{false}
	\BlankLine
	\Input{an MDP $\mdp = ((V, E), (\vo, \vr), \trans)$ and Streett pairs 
	$\SP= \{(L_i, U_i) \mid 1 \le i \le k\}$
	}
	\Output
	{
	$\as{\mdp, \streett{\SP}}$
	}
	\BlankLine
	$\good \gets \emptyset$; 	$Q \gets \emptyset$\;
	$\mathcal{\ec} \gets \mecalg(\mdp)$\;
	\lFor{$\ec \in \mathcal{\ec}$}{
		$Q \gets Q \cup \set{\construct(\ec)}$
	}
	\While{$Q \ne \emptyset$}{
		remove some $\ds(S)$ from $Q$\;
		\While{$\bad(\ds(S)) \ne \emptyset$}{
			$A \gets \at(\mdp[\ec], \bad(\ds(S)))$\;
			$\ds(S) \gets \remove(S, \ds(S), A)$\;
		}
		\If{$\mdp[S]$ contains at least one edge}{
			\For{$j \gets 1$ \KwTo $\lceil \log(\lvert S \rvert) \rceil$}{
				\ForEach{$H \in \set{G,\rev}$}{
					construct $H_j[S]$\;
					$\blue_j \gets \{v \in S \mid \OutDeg_H(v) > 2^j\}$\;
					$Z \gets S \setminus \reach(H_j[S], \blue_j)$\;
					\If{$Z \ne \emptyset$}{
						$\scc \gets \sccalg(H_j[Z])$\label{ldense:scc}\;
						\If{$\scc = S$}{
							$\good \gets \good \cup \set{C}$\label{ldense:good}\;
							continue with next iteration of while-loop\;
						}
						\If{$\lvert \scc \rvert \le \lvert S \rvert / 2$}{
							\If(\tcc*[f]{top SCC}){$H = \rev$}{
								$Q \gets Q \cup \remove(S, \ds(S), C)$\label{ldense:addback1}\;
								$R \gets \set{v \in \vr \cap \scc \mid \exists u \in S \setminus \scc 
								\text{ s.t.\ } (v, u) \in E}$\label{ldense:randout}\;
								$\scc \gets \scc \setminus \at(\mdp[C], R)$\; 
							}\Else(\tcc*[f]{bottom SCC}){
								$Q \gets Q \cup \remove(S, \ds(S), \at(\mdp[S], C))$\label{ldense:addback2}\;
							}
							$Q \gets Q \cup \construct(C)$\label{ldense:sccconstr}\;
							continue with next iteration of while-loop\;
						}
					}
				}
			}
		}
	}
	\Return{$\as{\mdp, \reacht{\bigcup_{\ec \in \good} \ec}}$}\;
\end{algorithm}

Algorithm~\ref{alg:streettdense} combines Algorithm~\ref{alg:streettimpr}
with the ideas of the MEC-algorithm for dense MDPs of~\cite{ChatterjeeH14} and 
the algorithm for graphs with Streett objectives of~\cite{ChatterjeeHL15}.
The difference to Algorithm~\ref{alg:streettimpr} lies in the search for 
strongly connected components. To detect a good end-component, it is essential 
to detect when a sub-MDP $\mdp[S]$ remains strongly connected after some 
vertices and their random attractor were removed from the vertex 
set $S$ for which the data structure $\ds(S)$ is
maintained in $Q$. For this it is sufficient to identify one 
strongly connected component~$\scc$ of the sub-MDP $\mdp[S]$:
The sub-MDP is strongly connected if and only if the SCC spans the whole 
sub-MDP, i.e., $\scc = S$.
As for Algorithm~\ref{alg:streettimpr}, the correctness of the algorithm
is based on maintaining the Invariants~\ref{inv:streettnorand}
and~\ref{inv:geccontained}. For maintaining these invariants it makes no difference
whether we compute all SCCs of $\mdp[S]$ or just one. Whenever $\mdp[S]$ is not 
strongly connected, there exists a top or bottom SCC that contains at most 
half of the vertices of $S$. In 
Algorithm~\ref{alg:streettdense} we search for such a ``small'' top or bottom
SCC of $\mdp[S]$. 
The search for a top SCC is done by searching for a bottom SCC in the reverse
graph. To search for a bottom SCC,
a sparsification technique called \emph{Hierarchical Graph Decomposition}
is used. This technique was introduced by~\cite{HenzingerKW99} for undirected 
graphs and extended to directed graphs and game graphs by~\cite{ChatterjeeH14}.
In the level-$j$ graph $H_j$ of a graph $H$ only the first $2^j$ outgoing edges 
of each vertex are considered, thus $H_j$ has $O(n \cdot 2^j)$ edges. The main
observation (Lemma~\ref{lem:decomp}) is that we can identify each bottom SCC 
with at most $2^j$ vertices by searching for bottom SCCs of $H_j$ that 
only contain vertices for which all their outgoing edges in $H$ are also in $H_j$.
The search is started at level $j = 1$ and then $j$ is doubled until such a bottom
SCC is found in $H_j$. Note that $H_j = H$ for $j \ge \log n$. When a bottom 
SCC is identified at level $j^*$ but not at $j^*-1$, then this bottom SCC has 
$\Omega(2^{j^*})$ vertices by the above observation. Further, the number of 
edges in the graphs from level $1$ to $j^*$ form a geometric series. Thus 
the work spent in all the levels up to $j^*$ can be bounded in terms of the number of
edges in $H_{j^*}$, that is, the bottom SCC of size $\Omega(2^{j^*})$ is 
identified in $O(n \cdot 2^{j^*})$ time. By searching ``in parallel'' for 
top and bottom SCCs and charging the needed time to the identified SCC, 
the total runtime can be bounded by $O(n^2)$. To identify only bottom SCCs of $H_j$
for which all the outgoing edges are present in $H_j$ we determine the 
set of ``blue'' vertices $\blue_j$ that have an out-degree higher than $2^j$
and remove vertices that can reach blue vertices before computing SCCs.
In the following we provide formal definitions and proofs for Algorithm~\ref{alg:streettdense}.

\begin{definition}[Hierarchical Graph Decomposition]\label{def:decomp}
Let $H = (V,E)$ be a simple directed graph. We consider for $j \in \mathbb{N}$ 
the subgraphs $H_j = (V, E_j)$ of $H$ where~$E_j$ contains for 
each vertex of $V$ its first $2^j$ outgoing edges in $E$ \upbr{for some arbitrary but
fixed ordering of the outgoing edges of each vertex}. Note that when
$j \ge \log (\max_{v \in V}{\OutDeg_H(v)})$, then $H_j = H$.
Let the set~$\blue_j$ denote all vertices with out-degree more than $2^j$ in~$H$.
\end{definition}

\begin{lemma}[See e.g.~\cite{HenzingerKL15}]\label{lem:decomp}
We use Definition~\ref{def:decomp}.
\begin{enumerate}
	\item A set $\scc \subseteq V \setminus \blue_j$ is a bottom SCC in $H_j$
	if and only if it is a bottom SCC in~$H$.
	\item If a set $\scc \subseteq V$ with $\lvert \scc \rvert \le 2^j$
	is a bottom SCC in $H$, then $\scc \subseteq V \setminus \blue_j$.
\end{enumerate}
\end{lemma}
\begin{proof}
\begin{enumerate}
	\item
	By $\scc \subseteq V \setminus \blue_j$ the outgoing edges of the vertices in $\scc$
	are the same in $H_j$ and in $H$. Thus we have $H_j[\scc] = H[\scc]$
	and $\scc$ has no outgoing edges in $H_j$ if and only if it has no outgoing 
	edges in $H$.
	\item
	In $H$ all outgoing edges of each vertex of $\scc$ have to go to other vertices 
	of~$\scc$. Thus each vertex of $\scc$ has an 
	out-degree of at most $\lvert \scc \rvert \le 2^j$ in $H$.\qedhere
\end{enumerate}
\end{proof}

\begin{proposition}[Runtime of Algorithm~\ref{alg:streettdense}]
		Algorithm~\ref{alg:streettdense} terminates in $O(n^2 + b \log n)$ time.
\end{proposition}

\begin{proof}
	Using the data structure of Lemma~\ref{lem:ds} (\cite{HenzingerT96}),
	the initialization phase of Algorithm~\ref{alg:streettdense} takes 
	$O(\textsc{MEC} + b + n)$ time, which is in $O(n^2 + b)$~\cite{ChatterjeeH14}.
	Further by Theorem~\ref{th:timedrmdp} the almost-sure reachability computation
	after the outer while-loop can be done in $O(\textsc{MEC})$ time. 
	Removing bad vertices takes total time~$O(n + b)$ by Lemma~\ref{lem:ds}. 
	Whenever a random attractor is computed, its edges are not considered further;
	thus all attractor computations take $O(m)$ total time by Definition~\ref{def:attr}.
	Whenever $\remove$ or $\construct$ are called (after the
	initialization of $Q$), the vertices that are removed resp.\ added are either
	\upbr{1} vertices for which the size of the SCC containing them was at least halved
	or \upbr{2} vertices that are not considered further. For each vertex 
	case~\upbr{1} can happen at most $O(\log n)$ times and case~\upbr{2} at 
	most once, thus all calls to $\remove$ or $\construct$ take
	total time $O((n+b)\log n)$ by Lemma~\ref{lem:ds}.
	
	To efficiently construct the graphs $H_j$ and compute $\blue_j$
	for $1 \le j < \lceil \log(n) \rceil$ and $H \in \set{G,\rev}$, we maintain 
	for all vertices a list of their incoming and outgoing edges, which we 
	update whenever we encounter obsolete entries while constructing $H_j$. 
	Each entry can be removed at most once, thus this can be done in $O(m)$ total time.
	
	Let $S$ be the set of vertices considered in an iteration of the outer while-loop
	and let $\lvert S \rvert = n'$.
	The $j$th iteration of the for-loop takes $O(n' \cdot 2^j)$ time
	because $H_j$ contains $O(n' \cdot 2^j)$ edges and constructing $H_j$ and $\blue_j$
	and computing reachability, SCCs, and $R$ can all be done in time linear in the 
	number of edges. The search in $G$ and $\rev$ only increases the runtime 
	by a factor of two. Further \emph{all} iterations up to the $j$th iteration
	can be executed in time $O(n' \cdot 2^j)$ as their runtimes 
	form a geometric series.
	Note that whenever a graph is not strongly connected, it contains a top SCC and 
	a bottom SCC and one of them has at most half of the vertices. Thus in some
	iteration $j^*$ a top or bottom SCC with either $\scc = S$ or 
	$\lvert \scc \rvert \le n' /2$ is found by Lemma~\ref{lem:decomp}.
	Since $\scc$ was not found 
	in iteration $j^*-1$, we have $\lvert \scc \rvert = \Omega(2^{j^*})$ by
	Lemma~\ref{lem:decomp}.
	
	In the case $\scc = S$ the vertices in $S$ are not considered further by 
	the algorithm. Thus we can bound the time for this iteration with 
	$O(n' \cdot 2^{\log(n')}) = O(n'^2)$
	and hence the total time for this case with $O(n^2)$.
	
  It remains to bound the time for the case $\lvert \scc \rvert \le n' /2$.
	Let $\lvert \scc \rvert = n_1$ and let $c$ be some constant such that 
	the time spent for the search of $\scc$ is bounded by $c \cdot n_1 \cdot n'$.
	We denote this time for the set $S$ over the whole algorithm with $f(n')$ and 
	show $f(n') = 2c n'^2$ by induction as follows:
	\begin{align*}
	f(n') &\le f(n_1) + f(n' - n_1) + c n' n_1 \,,\\
	&\le 2c n_1^2 + 2c (n'-n_1)^2 + c n' n_1 \,,\\
	&= 2c n_1^2 + 2c n'^2 - 4c n' n_1 + 2c n_1^2 + c n' n_1 \,,\\
	&= 2c n'^2 + 4c n_1^2 - 3c n' n_1 \,,\\
	&\le 2c n'^2 \,,
	\end{align*}
	where the last inequality follows from $n_1 \le n'/2$.
\end{proof}

\begin{proposition}[Soundness of Algorithm~\ref{alg:streettdense}]\label{prop:streettdensesound}
	Let $W$ be the set returned by Algorithm~\ref{alg:streettdense}.
	We have $W \subseteq \as{\mdp, \streett{\SP}}$.
\end{proposition}

\begin{proof}
	We follow the proof of Proposition~\ref{prop:streettimprsound}. 
	Let $\scc$ be a set of vertices added to $\good$ in Line~\ref{ldense:good}. 
	Since $\mdp[\scc]$ is strongly connected by 
	Lemma~\ref{lem:decomp}, we have that
	immediately before $\scc$ is added to $\good$ it was checked that
	$\mdp[\scc]$ contains at least one edge, is strongly connected, and $\bad(\ds(C))$
	is empty. Thus it is sufficient to show that Invariant~\ref{inv:streettnorand} holds
	in Algorithm~\ref{alg:streettdense}. 
	
	Before the first iteration of the while-loop $Q$ is initialized with
the maximal end-components of $\mdp$ and thus the invariant holds. 
Assume the invariant holds before the beginning of an iteration of the outer while-loop
and let $S$ be the set of vertices for which the data structure 
is removed from $Q$ in this iteration. In the inner while-loop some vertices 
and their random attractor in $\mdp[S]$ might be removed from $S$. Let $S'$ be 
the remaining vertices. By the definition of a random attractor there are no 
random edges from $S'$ to $S \setminus S'$ and thus by the induction hypothesis
there are no random edges from $S'$ to $V \setminus S'$.

Then either $\mdp[S']$ is strongly connected and no set is added to $Q$ in 
this iteration of the while-loop or either a top or a bottom SCC $\scc$ of $\mdp[S']$
is identified by Lemma~\ref{lem:decomp}.

If $\scc$ is a top SCC, then there are no edges from $S' \setminus \scc$ 
to $\scc$ and thus $S' \setminus \scc$ has no outgoing random edges.
Hence the invariant is maintained when $\ds(S' \setminus \scc)$ is added to $Q$.
Then the random vertices of $\scc$ with edges to vertices in $S' \setminus \scc$ and 
their random attractor are removed from $\scc$. Thus the remaining vertices
of $\scc$ have no random edges to $V \setminus \scc$ and the invariant is maintained
when the data structure of this vertex set is added to $Q$.

If $\scc$ is a bottom SCC, then there are no edges from $\scc$ to $S' \setminus \scc$;
thus the invariant is maintained when $\ds(\scc)$ is added to $Q$. The random
attractor of $\scc$ is removed from $S' \setminus \scc$ before the data structure of 
the remaining vertices is added to $Q$, hence the invariant is maintained in all 
cases.
\end{proof}

\begin{proposition}[Completeness of Algorithm~\ref{alg:streettdense}]\label{prop:streettdensecompl}
		Let $W$ be the set returned by Algorithm~\ref{alg:streettdense}.
	We have $\as{\mdp, \streett{\SP}} \subseteq W$.
\end{proposition}

\begin{proof}
	Following the proof of Proposition~\ref{prop:streettimprcompl}, it is sufficient
	to show by induction over the iterations of the outer-while loop that Invariant~\ref{inv:geccontained} holds in Algorithm~\ref{alg:streettdense}. 
	
	Before the first iteration of the outer while-loop, the 
	set $Q$ is initialized with the MECs of $\mdp$, thus the induction base holds.
	Let $S$ be the set of vertices for which the data structure is 
	removed from $Q$ in an iteration of the outer while-loop and 
	let $\mathcal{\ec}_S$ be the set of good end-components 
	contained in $S$. 
	Let $S'$ be the subset of $S$ that is not removed in the inner while-loop.
	We have $\ec \subseteq S'$ for every $\ec \in \mathcal{\ec}_S$ by
	Corollary~\ref{cor:geccontained}.
	
	Since every end-component contains an edge, $\mdp[S']$ contains at least one
	edge if $\mathcal{\ec}_S$ is not empty. 
	Then either $S'$ and thus all $\ec \in \mathcal{\ec}_S$ are added to $\good$
	(Line~\ref{ldense:good}) or an SCC $\scc \subsetneq S'$ of $\mdp[S']$ is 
	identified in Line~\ref{ldense:scc} by Lemma~\ref{lem:decomp}. By 
	Lemma~\ref{lem:eccontained}~\upbr{a} each $\ec \in \mathcal{\ec}_S$ is either
	a subset of $\scc$ or of $S' \setminus \scc$.
	For $\ec \subseteq \scc$ we have $R \cap \ec = \emptyset$ (Line~\ref{ldense:randout})
	since $\ec$ has no outgoing random edges and thus $\ec \subseteq \scc \setminus 
	\at(\mdp[\scc],R)$ by Lemma~\ref{lem:eccontained}~\upbr{b}.
	For $\ec \subseteq S' \setminus \scc$ we have $\ec \cap \scc = \emptyset$
	and thus $\ec \subseteq S' \setminus \at(\mdp[S'],\scc)$
	by Lemma~\ref{lem:eccontained}~\upbr{b}.	
	The data structures of $\scc \setminus \at(\mdp[\scc],R)$
	and of $S' \setminus \at(\mdp[S'],\scc)$
	are added to $Q$ in lines~\ref{ldense:sccconstr} and either~\ref{ldense:addback1}
	or~\ref{ldense:addback2}, hence the invariant holds after the outer while-loop.
\end{proof}

\subsection{Algorithm for Sparse MDPs with Streett Objectives}\label{sec:streettsparse}

\begin{algorithm}
	\SetAlgoRefName{StreettMDPsparse}
	\caption{Algorithm for sparse MDPs with Streett Objectives}
	\label{alg:streettsparse}
	\SetKwInOut{Input}{Input}
	\SetKwInOut{Output}{Output}
	\BlankLine
	\Input{an MDP $\mdp = ((V, E), (\vo, \vr), \trans)$ and Streett pairs 
	$\SP= \{(L_i, U_i) \mid 1 \le i \le k\}$
	}
	\Output
	{
	$\as{\mdp, \streett{\SP}}$
	}
	\BlankLine
	$\good \gets \emptyset$; $Q \gets \emptyset$; $\mathcal{\ec} \gets \mecalg(\mdp)$\;
	\lFor{$\ec \in \mathcal{\ec}$}{
		$Q \gets Q \cup \set{\construct(\ec)}$
	}
	\While{$Q \ne \emptyset$}{
		remove some $\ds(S)$ from $Q$\;
		\While{$\bad(\ds(S)) \ne \emptyset$}{
			$A \gets \at(\mdp[\ec], \bad(\ds(S)))$\;
			$(S, \ds(S)) \gets \remove(S, \ds(S), A)$\;
			add label $h$ ($t$) to vertices that just lost 
			an incoming (outgoing) edge\label{lsparse:label1}\;
		}
		$H \gets \set{v \in S \mid h \in label(v)}$; $T \gets \set{v \in S \mid t \in label(v)}$\;
		\If{$\mdp[S]$ contains at least one edge}{
			\lIf{$\lvert H \rvert + \lvert T \rvert = 0$}{
				$\good \gets \good \cup \set{S}$\label{lsparse:good1}
			}\ElseIf{$\lvert H \rvert + \lvert T \rvert \ge \sqrt{m / \log n}$}{
			\tcc*[l]{like Algorithm~\ref{alg:streettimpr} plus maintaining labels}
				remove all labels from $S$\;
				$\mathcal{\scc} \gets \allsccalg(\mdp[S])$; $S' \gets S$\;
				\For{$\scc \in \mathcal{\scc}$}{
					$A \gets \at(\mdp[\scc], \set{v \in \vr \cap \scc \mid \exists w \in 
					S' \setminus \scc \text{ s.t.\ } (v, w) \in E})$\;
					add label $h$ ($t$) to vertices with incoming (outgoing) edge from (to) $A$\label{lsparse:label2}\;
					\lIf{$\scc$ is largest SCC in $\mathcal{\scc}$}{
						$(S, \ds(S)) \gets \remove(S, \ds(S), A)$
					}\Else{
						$(S, \ds(S)) \gets \remove(S, \ds(S), \scc)$; $\scc \gets \scc \setminus A$\;
						$Q \gets Q \cup \set{\construct(\scc)}$\;
					}
				}
				$Q \gets Q \cup \ds(S)$\;
			}\Else{
				Search in lock-step from each $v \in T$ in $G[S]$ and from each 
				$v \in H$ in $\rev[S]$, terminate when first search has found a bottom SCC~$\scc$\label{lsparse:scc}\;
				\tcc*[l]{like Alg.~\ref{alg:streettdense} plus maintaining labels}
				\lIf{$\scc = S$}{
					$\good \gets \good \cup \set{S}$\label{lsparse:good2}
				}\Else{
					remove all labels from $\scc$\;
					\If(\tcc*[f]{top SCC}){$\scc$ is bottom SCC in $\rev[S]$}{
						$Q \gets Q \cup \remove(S, \ds(S), \scc)$\label{lsparse:addback1}\;
						$\scc \gets \scc \setminus \at(\mdp[\scc], \set{v \in \vr \cap \scc \mid 
						\exists u \in S \setminus \scc \text{ s.t.\ } (v, u) \in E})$\label{lsparse:randout}\;
					}\Else(\tcc*[f]{bottom SCC}){
					$Q \gets Q \cup \remove(S, \ds(S), \at(\mdp[S],\scc))$\label{lsparse:addback2}\;
					}
					add label $h$ ($t$) to vertices that just lost an incoming (outgoing) edge\label{lsparse:label3}\;
					$Q \gets Q \cup \construct(\scc)$\label{lsparse:sccconstr}\;
				}
			}
		}
	}
	\Return{$\as{\mdp, \reacht{\bigcup_{\ec \in \good} \ec}}$}\;
\end{algorithm}

Algorithm~\ref{alg:streettsparse} combines Algorithm~\ref{alg:streettimpr}
with the ideas of the MEC-algorithm for sparse MDPs of~\cite{ChatterjeeH14} and 
the algorithm for graphs with Streett objectives of~\cite{HenzingerT96}.
As for dense graphs, the difference to
Algorithm~\ref{alg:streettimpr} lies in the search for strongly connected 
components in the sub-MDP $\mdp[S]$ induced by a vertex set~$S$ for which 
the data structure was maintained in $Q$ and then some vertices (and their 
random attractor) might have been removed from it. The algorithm is based 
on the following observation: Whenever a strongly connected component $\scc$
is not strongly connected after some vertices $A$ were removed from it,
then \upbr{a} there is a top and a bottom SCC in $\mdp[\scc \setminus A]$ and 
\upbr{b} some vertex of the top SCC had an incoming edge from a vertex of~$A$
and some vertex of the bottom SCC had an outgoing edge to a vertex of~$A$.
We label vertices that lost an incoming edge since the last SCC computation 
with $h$ (for head) and vertices that lost an outgoing edge with $t$ (for tail).
If more than $\sqrt{m / \log n}$ vertices are labeled, we remove all labels and 
compute SCCs as in Algorithm~\ref{alg:streettimpr}; this can happen at most 
$\sqrt{m \log n}$ times. Otherwise we search for the smallest top or bottom SCC 
of $\mdp[S]$ by searching in \emph{lock-step} from all labeled vertices.
Lock-step means that one step in each of the searches is executed before 
the next step of a search is started and all searches are stopped as soon as 
one search finishes. The search for top SCCs is done by 
searching for bottom SCCs in the reverse graph. Tarjan's depth-first search 
based SCC algorithm detects a bottom SCC in time proportional to the number 
of edges in the bottom SCC when the search is started from a vertex inside
the bottom SCC. As there are at most $\sqrt{m / \log n}$ parallel searches,
the time for all the lock-step searches is $O(\sqrt{m / \log n})$ times the 
number of edges in the smallest top or bottom SCC of $\mdp[S]$. Since 
each edge can be in the smallest SCC at most $O(\log n)$ times, this leads to 
a total runtime of $O(m \sqrt{m \log n})$.
Whenever an SCC is identified, the labels of its vertices are removed. The 
Invariants~\ref{inv:streettnorand} and~\ref{inv:geccontained} are maintained
as in Algorithm~\ref{alg:streettimpr}.

\begin{lemma}[Label Invariant]\label{lem:streettsparselabel}
	In Algorithm~\ref{alg:streettsparse} the following invariant is maintained
	for every set $S$ for which the data structure $\ds(S)$ is in $Q$: 
	Either \upbr{1} no vertex of $S$ is labeled and $\mdp[S]$ is strongly connected or 
	\upbr{2} in each top SCC of $\mdp[S]$ at least one vertex is labeled with~$h$
	and in each bottom SCC of $\mdp[S]$ at least one vertex is labeled with~$t$.
\end{lemma}

\begin{proof}
	The proof is by induction over the iterations of the outer while-loop.
	After the initialization of $Q$ with the MECs of $\mdp$ no vertex is labeled
	and every set $S$ with $\ds(S) \in Q$ is strongly connected.
	Let now $S$ denote the set for which $\ds(S)$ is removed from $Q$ 
	at the beginning of an iteration of the outer while-loop and assume the 
	invariant holds for $S$.
	\begin{observation*}
	We have for non-empty vertex sets $W$ and $Z = W \setminus Y$
	with $Y \subsetneq W$ that if $C$ is a top \upbr{bottom} SCC in $\mdp[Z]$
	but had incoming \upbr{outgoing} edges in $P[W]$, then these incoming 
	\upbr{outgoing} edges were from \upbr{to} vertices in $Y$. Thus when the 
	invariant holds for $W$ and
	we label each vertex of $Z$ with an incoming edge from $Y$ with $h$ and 
	each vertex of $Z$ with an outgoing edge to $Y$ with $t$, then the invariant
	holds for $Z$.
	\end{observation*}
	By this observation the invariant remains to hold for $S$ after the inner 
	while-loop.
	In the case $\lvert H \rvert + \lvert T \rvert \ge \sqrt{m / \log n}$
	all labels are removed from $S$ and then each SCC $\scc$ of $\mdp[S]$ is 
	considered separately. Note that for each $\scc$ the invariant holds
	and thus the invariant remains to hold for the set $\scc$ added to $Q$ after 
	the vertices in $A$ were removed and the corresponding labels were added in 
	Line~\ref{lsparse:label2}.
	In the case $\lvert H \rvert + \lvert T \rvert < \sqrt{m / \log n}$
	a bottom or top SCC~$\scc$ of $\mdp[S]$ is identified and all labels
	of $\scc$ are removed. The invariant holds for $\scc$ and thus the invariant remains 
	to hold for the set $\scc$ added to $Q$ after 
	vertices were removed from $\scc$ in Line~\ref{lsparse:randout}
	and the corresponding labels were added in Line~\ref{lsparse:label3}.
	By the above observation with $W = S$ and $Y = \at(\mdp[S], \scc)$
	the invariant also holds for the set $S \setminus \at(\mdp[S], \scc)$
	for which the data structure is added to $Q$ after the corresponding
	labels are added in Line~\ref{lsparse:label3}.
\end{proof}

\begin{proposition}[Runtime of Algorithm~\ref{alg:streettsparse}]
	Algorithm~\ref{alg:streettsparse} takes $O(m \sqrt{m \log n} + b \log n)$ time.
\end{proposition}

\begin{proof}
	Using the data structure of Lemma~\ref{lem:ds} (\cite{HenzingerT96}),
	the initialization phase of Algorithm~\ref{alg:streettsparse} takes
	$O(\textsc{MEC} + b + n)$ time, which is in $O(m\sqrt{m} + b)$~\cite{ChatterjeeH14}.
	Further by Theorem~\ref{th:timedrmdp} the almost-sure reachability computation
	after the outer while-loop can be done in $O(\textsc{MEC})$ time. 
	Removing bad vertices takes total time~$O(n + b)$ by Lemma~\ref{lem:ds}. 
	Since a label is added only when an edge is not considered further by the 
	algorithm, the total time for adding and removing labels is $O(m)$.
	Whenever a random attractor is computed, its edges are not considered further;
	thus all attractor computations take $O(m)$ total time by Definition~\ref{def:attr}.
	Note that whenever a graph is not strongly connected, it contains a top SCC and 
	a bottom SCC and one of them has at most half of the vertices. Thus whenever 
	a top or bottom SCC $\scc$ with $\scc \subsetneq S$ is identified in 
	Line~\ref{lsparse:scc}, then $\lvert \scc \rvert \le \lvert S \rvert / 2$.
	This implies by Lemma~\ref{lem:streettsparselabel}
	that whenever $\remove$ or $\construct$ are called (after the
	initialization of $Q$), the vertices that are removed resp.\ added are either
	\upbr{1} vertices for which the size of the SCC containing them was at least halved
	or \upbr{2} vertices that are not considered further. Case \upbr{1} can happen
	at most $O(\log n)$ times, thus all calls to $\remove$ or $\construct$ take
	total time $O((n+b)\log n)$ by Lemma~\ref{lem:ds}.
	
	It remains to bound the time for identifying SCCs and determining the
	random boundary vertices $R =  \set{v \in \vr \cap \scc \mid 
			\exists u \in S \setminus \scc \text{ s.t.\ } (v, u) \in E})$ in
	Case~1, $\lvert H \rvert + \lvert T \rvert \ge \sqrt{m / \log n}$,
	and Case~2, $\lvert H \rvert + \lvert T \rvert < \sqrt{m / \log n}$.
	Since labels are added only when edges are not considered further and all 
	labels of the considered vertices are deleted when Case~1 occurs, Case~1 can 
	happen at most $\sqrt{m \log n}$ times. Thus the total time for Case~1 can be 
	bounded by $O(m \sqrt{m \log n})$. In Case~2 we charge the time for the 
	$O(\sqrt{m / \log n})$ lock-step searches to the edges in the identified 
	SCC~$\scc$. With Tarjan's SCC algorithm~\cite{T72} a bottom SCC is identified 
	in time proportional to the number of edges in the bottom SCC when the search
	is started at a vertex in the bottom SCC, which is in 
	Algorithm~\ref{alg:streettsparse} guaranteed by Lemma~\ref{lem:streettsparselabel}
	for both top and bottom SCCs. Since always the smallest top or bottom SCC in 
	$\mdp[S]$ is identified, each edge is charged at most $O(\log n)$ times.
	Thus the total time for identifying SCCs in Case~2 is $O(m \sqrt{m \log n})$.
	Determining the random boundary vertices $R$ in Case~2 can be charged to 
	the edges in $\scc$
	and to the edges from $\scc$ to $S \setminus \scc$, which are then not considered 
	further by the algorithm. Thus the total runtime of the algorithm is 
	$O(m \sqrt{m \log n})$.
\end{proof}

\begin{proposition}[Correctness of Algorithm~\ref{alg:streettsparse}]
	Let $W$ be the set returned by Algorithm~\ref{alg:streettsparse}.
	We have $W = \as{\mdp, \streett{\SP}}$.
\end{proposition}

\begin{proof}
	Lemma~\ref{lem:streettsparselabel} implies that whenever a vertex set is added 
	to $\good$ in Line~\ref{lsparse:good1}, it induces a strongly connected sub-MDP.
	Thus we have that immediately before a set of vertices~$\scc$ is added to 
	$\good$ in Line~\ref{lsparse:good1} or Line~\ref{lsparse:good2}, it is checked that 
	$\mdp[\scc]$ contains at least one edge, is strongly connected, and $\bad(\ds(C))$
	is empty. For the soundness and completeness of Algorithm~\ref{alg:streettsparse} 
	it remains to show the Invariants~\ref{inv:streettnorand} and~\ref{inv:geccontained}.
	We have for each iteration of the outer while-loop: 
	The inner while-loop is the same as in Algorithms~\ref{alg:streettimpr} 
	and~\ref{alg:streettdense}. In the case $\lvert H \rvert + \lvert T \rvert = 0$,
	the currently considered set of vertices is added to $\good$ and
	no set is added to $Q$. If $\lvert H \rvert + \lvert T \rvert \ge \sqrt{m / \log n}$,
	the same operations as in Algorithm~\ref{alg:streettimpr} are performed.
	If $\lvert H \rvert + \lvert T \rvert < \sqrt{m / \log n}$, 
	like in Algorithm~\ref{alg:streettdense},
	either a top or a bottom SCC is identified and then the same operations as in 
	Algorithm~\ref{alg:streettdense} are applied to the identified SCC and the 
	remaining vertices. As the operations in 
	Algorithms~\ref{alg:streettimpr} and~\ref{alg:streettdense} preserve 
	the invariants, this is also true for Algorithm~\ref{alg:streettsparse}.	
\end{proof}

\section{MDPs with Rabin and Disjunctive Büchi and coBüchi Objectives}\label{sec:rabin}

In the first part of this section we prove the following conditional lower bounds for Rabin, 
and disjunctive Büchi and coBüchi objectives.

\begin{theorem}
  Assuming STC, there is no combinatorial $O(n^{3-\epsilon})$ or $O((kn^2)^{1-\epsilon})$ algorithm for 
  each of the following problems:
  \begin{enumerate}
   \item computing the a.s.~winning set in an MDP with a disjunctive Büchi query;
   \item computing the winning set in a graph with a disjunctive coBüchi objective and
	 thus also computing the a.s.~winning set in an MDP for disjunctive coBüchi objective or 
						 a disjunctive coBüchi query;
   \item computing the a.s.~winning set in an MDP with a Rabin objective.
  \end{enumerate}
\end{theorem}

\begin{theorem}
  Assuming SETH or OVC, there is no $O(m^{2-\epsilon})$ or $O((k\cdot m)^{1-\epsilon})$ algorithm for
  each of the following problems:
  \begin{enumerate}
   \item computing the a.s.~winning set in an MDP with a disjunctive Büchi query;
   \item computing the a.s.~winning set in an MDP with a disjunctive coBüchi objective or 
						 a disjunctive coBüchi query;
   \item computing the a.s.~winning set in an MDP with a disjunctive Singleton coBüchi objective or
   						 a disjunctive Singleton coBüchi query;
   \item computing the a.s.~winning set in an MDP with a Rabin objective.
  \end{enumerate}  
\end{theorem}

On the algorithmic side we prove the following theorem in the second part of this section.
Note that a Rabin objective corresponds to a disjunctive objective over 
1-pair Rabin objectives.
\begin{theorem}
	Given an MDP $\mdp = ((V, E), (\vo, \vr), \trans)$
 	and a Rabin objective wit Rabin pairs $\RP= \{(L_i, U_i) \mid 1 \le i \le k\}$,
 	let $b = \sum_{i=1}^k (\lvert L_i \rvert + \lvert U_i \rvert)$.
  Let $\textsc{MEC}$ denote the time to compute a MEC-decomposition.
  \begin{enumerate}
  	\item The almost-sure winning set $\as{\mdp, \objsty{Rabin}{\RP}}$ can be computed in 
  	$O(k \cdot \textsc{MEC})$ time.
  	\item If $U_i = \emptyset$ for all $ 1 \le i \le k$ \upbr{i.e.\ the Rabin 
  	pairs are Büchi objectives},
  	then the almost-sure winning set for the disjunctive objective over the Rabin pairs 
  	can computed in $O(\textsc{MEC} + b)$ time and the disjunctive query in 
  	$O(k \cdot m + \textsc{MEC})$ time.
  	\item If $L_i = V$ for all $ 1 \le i \le k$ \upbr{i.e.\ the Rabin 
  	pairs are coBüchi objectives}, then the 
  	almost-sure winning set for the disjunctive objective and the disjunctive
  	query over the Rabin pairs 
  	can computed in $O(k \cdot m + \textsc{MEC})$ time.
  \end{enumerate}
\end{theorem}

\subsection{Conditional Lower Bounds for Rabin, Büchi and coBüchi}

The conditional lower bounds for Rabin, and disjunctive  Büchi and coBüchi objectives are based 
on our results for reachability (see Section~\ref{subsec:reach_lowerbounds}) and safety objects 
(see Section~\ref{subsec:safety_lowerbounds}) and 
the Observations \ref{obs:nonempty_safety_coBuchi}, \ref{obs:Reachability_Buchi} \& \ref{obs:DisjBuchiRabin}
that interlink these objectives.

\begin{proposition}
  Assuming STC, there is no combinatorial $O(n^{3-\epsilon})$ or $O((k\cdot n^2)^{1-\epsilon})$ algorithm for 
  \begin{enumerate}
   \item computing the winning set in an MDP with a disjunctive Büchi query,
   \item computing the winning set in a graph with a disjunctive coBüchi objective, and
   \item computing the winning set in an MDP with a Rabin objective.
  \end{enumerate}
  Moreover, there is no such algorithm deciding whether the winning set is non-empty
  or deciding whether a specific vertex is in the winning set.
\end{proposition}
\begin{proof}
 1) By Observation~\ref{obs:Reachability_Buchi} in MDPs reachability can be reduced in
 linear time to Büchi. Thus the result follows from
    the corresponding hardness result for reachability (cf.\ Theorem~\ref{thm:reach_STChard}).
    
 2) By Observation \ref{obs:nonempty_safety_coBuchi} the winning set of disjunctive safety is non-empty iff
    the winning set of disjunctive coBüchi with the same target sets is non-empty. Thus 
    the result follows from the corresponding hardness result for safety (cf.\ Theorem~\ref{thm:safety_STChard}).\\
    For the problem of deciding whether a specific vertex is in the winning set, recall
    that the graph $G'$ constructed in Reduction~\ref{red:TriangletoGraphs} is such that vertex $s$
    appears in each infinite path and thus if there is a winning strategy starting in some vertex,
    then there is also one starting in $s$. 
    That is, deciding on $G'$ whether $s$ is winning is equivalent to deciding whether the winning
    set is non-empty. Hence, the lower bound for the former follows.    
    
 3) The result follows from (2) and Observation~\ref{obs:DisjBuchiRabin}, by which disjunctive coBüchi objectives are special instances of Rabin objectives.
\end{proof}

\begin{proposition}
  Assuming SETH or OVC, there is no $O(m^{2-\epsilon})$ or $O((k\cdot m)^{1-\epsilon})$ algorithm for
  \begin{enumerate}
   \item computing the winning set in an MDP with a disjunctive Büchi query,
   \item computing the winning set in an MDP with a disjunctive coBüchi objective or 
						 a disjunctive coBüchi query,
   \item computing the winning set in an MDP with a disjunctive Singleton coBüchi objective or
   						 a disjunctive Singleton coBüchi query, and
   \item computing the winning set in an MDP with a Rabin objective.
  \end{enumerate}  
  Moreover, there is no such algorithm for deciding whether the winning set is non-empty
  or deciding whether a specific vertex is in the winning set. 
\end{proposition}
\begin{proof}
 1) By Observation~\ref{obs:Reachability_Buchi} in MDPs reachability can be reduced in
 linear time to Büchi. Thus the result follows from
    the corresponding hardness result for reachability (cf.\ Theorem~\ref{thm:reach_OVChard}).
    
 2) By Observation \ref{obs:nonempty_safety_coBuchi} the winning set of disjunctive safety is non-empty iff
    the winning set of disjunctive coBüchi with the same target sets is non-empty. Thus 
    the result follows from the corresponding hardness result for safety (cf.\  Theorem~\ref{thm:safety_OVChard}).\\
    For the problem of deciding whether a specific vertex is in the winning set, recall
    that the MDP~$\mdp$ constructed in Reduction~\ref{red:OVtoMDPsafety} is such that vertex~$s$
    appears in each infinite path and thus if there is a winning strategy starting in some vertex,
    then there is also one starting in~$s$. 
    That is, deciding on $\mdp$ whether $s$ is winning is equivalent to deciding 
    whether the winning
    set is non-empty. Hence, the lower bound for the former follows.  
 
 3) This holds by (2) and the fact that all sets $\target_i$ in Lemma~\ref{lem:OVtoMDPsafety} are singletons.
 
 3) The result follows from (2) and Observation~\ref{obs:DisjBuchiRabin}, by which disjunctive coBüchi objectives are special instances of Rabin objectives.
\end{proof}

\subsection{Algorithm for MDPs with Rabin Objectives}
In this section we describe an algorithm for MDPs with Rabin objectives that 
considers each MEC of the input MDP separately. This formulation has the 
advantage that we can obtain a faster runtime than previously known 
for the special case of disjunctive coB{\"u}chi objectives, which we describe in
Section~\ref{sec:cobuchialg}. The special case of B{\"u}chi objectives is 
described in Section~\ref{sec:buchialg}.

For Rabin objectives a good end-component could, equivalently to 
Definition~\ref{def:good_ec}, be defined as follows.
\begin{definition}[Good Rabin End-Component]\label{def:good_rabin_ec}
  Given an MDP $\mdp$ and a set $\RP= \{(L_i, U_i) \mid 1 \le i \le k\}$ of 
  Rabin pairs, 
  a \emph{good Rabin end-component} is an end-component $\ec$ of $\mdp$ such that
  $L_i \cap \ec \ne \emptyset$ and  $U_i \cap \ec = \emptyset$ for some 
  $1 \le i \le k$.
\end{definition}
As for Streett objectives, we determine the almost-sure winning set for Rabin 
objectives by computing almost-sure reachability of the union of all good Rabin 
end-components. The correctness of this approach
follows from Corollary~\ref{cor:gecsound-gen} and Proposition~\ref{prop:geccompl-gen}.
We use the notation defined in Section~\ref{sec:algprelim}.
Our strategy to find all good Rabin end-components is as follows. First the
MEC-decomposition of the input MDP $\mdp$ is determined. For each MEC~$\ec$
and separately for each $1 \le i \le k$ we first remove the set $U_i$ and its
random attractor and then compute the MEC-decomposition in the sub-MDP induced
by the remaining vertices. Every newly computed MEC that contains a vertex of $L_i$
is a good Rabin end-component. If the MEC $\ec$ of $\mdp$ contains one such 
good end-component, then by Corollary~\ref{cor:winning_mec} all vertices of $\ec$
are in the almost-sure winning set for the Rabin objective. Thus we can immediately
add $\ec$ to the set of winning MECs in Line~\ref{lRabinMDPbasic:endif1}.\footnote{
We could alternatively add only the vertices in the good end-component because
the winning MEC would be detected as winning in the final almost-sure reachability
computation; the presented formulation shows the similarities
to the coB{\"u}chi algorithm in Section~\ref{sec:cobuchialg}. Additionally, this 
allows reusing the initial MEC-decomposition for the almost-sure reachability
computation.}

\begin{algorithm}
	\SetAlgoRefName{RabinMDP}
	\caption{Algorithm for MDPs with Rabin Objectives}
	\label{alg:rabinbasic}
	\SetKwInOut{Input}{Input}
	\SetKwInOut{Output}{Output}
	\BlankLine
	\Input{MDP $\mdp = ((V, E), (\vo, \vr), \trans)$ and 
	Rabin pairs $\RP= \{(L_i, U_i) \mid 1 \le i \le k\}$
	}
	\Output
	{
	$\as{\mdp, \objsty{Rabin}{\RP}}$
	}
	\BlankLine
	$\mathcal{\ec} \gets \mecalg(\mdp)$;
	$\winning \gets \emptyset$\;
	
	\ForEach{$\ec \in \mathcal{\ec}$\label{lRabinMDPbasic:foreach1}}{
		\For{$1 \leq i \leq k$ \label{lRabinMDPbasic:forloop}}{
		    \If{$L_i \cap  \ec \ne \emptyset$\label{lRabinMDPbasic:if1}}{
			  $\mathcal{Y} \gets \mecalg(\ec \setminus \at(\mdp[{\ec}], U_i))$ \label{lRabinMDPbasic:subMECs}\;
			  \ForEach{$Y \in \mathcal{Y}$}{
			      \If{$L_i \cap  Y \ne \emptyset$\label{lRabinMDPbasic:if2}}{
				  $\winning \gets \winning \cup \{\ec\}$\label{lRabinMDPbasic:endif1}\;
				  continue with next $\ec \in \mathcal{\ec}$\;
			      }
			  }
		    }	    
		}
	}
	\Return{$\as{\mdp, \reacht{\bigcup_{\ec \in \winning} \ec}}$ \label{lRabinMDPbasic:return}\;}
\end{algorithm}
\begin{proposition}[Runtime of Algorithm~\ref{alg:rabinbasic}]
    Algorithm~\ref{alg:rabinbasic} can be implemented in $O(k \cdot \textsc{MEC})$ time.
\end{proposition}
\begin{proof}
    The initialization of $\mathcal{\ec}$ with all MECs of the input
    MDP $\mdp$ can clearly be done in $O(\textsc{MEC})$ time. Further by 
    Theorem~\ref{th:timedrmdp} the final almost-sure reachability computation
    can be done in $O(\textsc{MEC})$ time\footnote{
    Actually the almost-sure reachability computation can be done in $O(m)$
    reusing the already computed MEC decomposition.
    }. 
    Assume that each vertex has a list of the sets $L_i$ and $U_i$ for $1 \le i \le k$ it belongs to.
    (We can generate these lists from the lists of the Rabin pairs in $O(b) 
    = O(nk)$ time at the beginning of the algorithm.)
    Consider an iteration of the outer for-each loop, 
    let $\ec$ denote the considered MEC, and fix one iteration $i$ of the $k$ 
    iterations of the for loop. Line~4 requires $O(|\ec|)$ time.
    Let $m_\ec$ be the number of edges in $\mdp[\ec]$ and let 
    $\textsc{MEC}_\ec$ denote the time needed to compute a MEC-decomposition
    on $\mdp[\ec]$. The Line~5 requires $O(m_\ec + \textsc{MEC}_\ec) =
    O(\textsc{MEC}_\ec)$ time. The inner for-each loop takes $O(|\ec|)$ time 
    as in each iteration we need $O(|Y|)$ in Line~7 
    and constant time in Line~8. Thus in total we have 
    $O(b+\textsc{MEC} + 
	\sum_{\ec \in \mathcal{\ec}}
	    k \cdot
	    (
	      |\ec| + 
	      \textsc{MEC}_\ec 
	    )
    )
    =
    O(k \cdot \textsc{MEC})
    $.
\end{proof}

\begin{proposition}[Correctness of Algorithm~\ref{alg:rabinbasic}]
    Algorithm~\ref{alg:rabinbasic} computes $\as{\mdp, \objsty{Rabin}{\RP}}$.
\end{proposition}
\begin{proof}
  By the Corollaries \ref{cor:winning_mec} \& \ref{cor:gecsound-gen} and Proposition~\ref{prop:geccompl-gen} we know that it 
  suffices to correctly classify each MEC as either \emph{winning} or \emph{not winning}; we say a MEC is \emph{winning} iff it contains a good Rabin EC, 
  that is, it contains 
  an EC $\ec$ such that $L_i \cap \ec \ne \emptyset$ and $U_i \cap \ec = \emptyset$ for some $1 \leq i \leq k$.
  The loops in Lines~\ref{lRabinMDPbasic:foreach1} \& \ref{lRabinMDPbasic:forloop} iterate over all MECs $\ec$
  and all Rabin Pairs $(L_i,U_i)$.
  What remains to show is that  
  Lines~\ref{lRabinMDPbasic:if1}--\ref{lRabinMDPbasic:endif1} correctly classify
  whether a MEC contains a good EC satisfying the Rabin pair $(L_i,U_i)$.
 
  \begin{itemize}
    \item Assume $\ec$ contains a good EC $\ec'$ that satisfies $(L_i,U_i)$, i.e.,\
	  $L_i \cap \ec' \ne \emptyset$ and $U_i \cap \ec' = \emptyset$.
	  Then the if condition in Line~\ref{lRabinMDPbasic:if1} is true and the algorithm subtracts the random attractor
	  of $U_i$.
	  As $\ec'$ is strongly connected, has no outgoing random edges, and $U_i \cap \ec' = \emptyset$, 
	  it does not intersect with $\at(\mdp[{\ec}], U_i)$ (see also Lemma~\ref{lem:eccontained}).
	  Thus there is a MEC $Y \in \mathcal{Y}$ that contains $\ec'$ and thus $L_i \cap Y \ne \emptyset$.
	  Hence, the algorithm correctly classifies the set $\ec$ as winning MEC.
    \item Assume the algorithm classifies a MEC~$\ec$ as winning. 
	  Then for some $i$ in Line~\ref{lRabinMDPbasic:if2} there is an end-component 
	  $Y \in \mathcal{Y}$ of $\mdp[\ec \setminus \at(\mdp[{\ec}], U_i)]$
	  with $L_i \cap  Y \ne \emptyset$ and $U_i \cap  Y = \emptyset$,
	  i.e., $Y$ is a good end-component in $\mdp[\ec \setminus \at(\mdp[{\ec}], U_i)]$.
	  Moreover, there cannot be a random edge from $u \in Y$ to $\at(\mdp[{\ec}], U_i)$ as such an $u$
	  would be included in the random attractor $\at(\mdp[{\ec}], U_i)$.
	  Thus $Y$ is also a good end-component of the full MDP~$\mdp$,
	  i.e., it was classified correctly.	  
  \end{itemize}
  By the above we have that whenever the outer for-each loop terminates, the set $\winning$ consists of all
  winning MECs and then by Corollary~\ref{cor:gecsound-gen} and Proposition~\ref{prop:geccompl-gen} we can compute 
  $\as{\mdp, \objsty{Rabin}{\RP}}$ by computing almost-sure reachability
  of the union of all winning MECs.
\end{proof}

\subsection{Algorithms for MDPs with Büchi Objectives}
\label{sec:buchialg}

As Büchi objectives can be encoded as Rabin pairs, 
Algorithm~\ref{alg:rabinbasic} can also be used 
to compute the a.s.~winning set for disjunctive Büchi objectives. 
However, Büchi objectives allow for some immediate simplifications that result in Algorithm~\ref{alg:DisjObjBuchiMDP}.
This simplifications are based on the observation that for Büchi all sets $U_i$ are empty and 
therefore also the random attractors computed in Line~\ref{lRabinMDPbasic:subMECs} of Algorithm~\ref{alg:rabinbasic} are empty.
Hence, there is also no need to recompute the MECs and 
deciding whether a MEC is winning reduces to testing whether it intersects with one of the target sets.

\begin{algorithm}
	\SetAlgoRefName{DisjObjBüchiMDP}
	\caption{Algorithm for MDPs with Disjunctive Büchi Objectives}
	\label{alg:DisjObjBuchiMDP}
	\SetKwInOut{Input}{Input}
	\SetKwInOut{Output}{Output}
	\BlankLine
	\Input{
	  MDP $\mdp = ((V, E), (\vo, \vr), \trans)$ and Büchi objectives $\target_i$ for  $1 \le i \le k$
	}
	\Output
	{
	  $\as{\mdp, \bigvee_{1 \le i \le k} \objsty{Büchi}{\target_i}}$
	}
	\BlankLine
	$\mathcal{\ec} \gets \mecalg(\mdp)$;
	$\winning \gets \emptyset$\;
	
	\ForEach{$\ec \in \mathcal{\ec}$}{
	    \If{$\bigcup_{1 \leq i \leq k} \target_i \cap  \ec \ne \emptyset$\label{lDisjObjBuchiMDP:if}}
	    {
		 $\winning \gets \winning \cup \{\ec\}$ \label{lDisjObjBuchiMDP:addec}\;
	    }	    
	}
	\Return{$\as{\mdp, \reacht{\bigcup_{\ec \in \winning} \ec}}$\;}
\end{algorithm}

\begin{proposition}[Runtime of Algorithm~\ref{alg:DisjObjBuchiMDP}]
	Algorithm~\ref{alg:DisjObjBuchiMDP} can be implemented in $O(\textsc{MEC} + b)$ time.
\end{proposition}
\begin{proof}
    The initialization of $\mathcal{\ec}$ with all MECs of the input
    MDP $\mdp$ can clearly be done in $O(\textsc{MEC})$ time. Further by 
    Theorem~\ref{th:timedrmdp} the final almost-sure reachability computation
    can be done in $O(\textsc{MEC})$ time. 
    Assume that each vertex has a flag indicating whether it is in one of the sets $\target_i$ or in none of them 
    (We can generate these flags from lists of the sets $\target_i$ in $O(b)$
    time at the beginning of the algorithm.).
    Consider an iteration of the for-each loop, 
    let $\ec$ denote the considered MEC and fix some iteration $i$ of the for loop. 
    One Iteration costs $O(|\ec|)$ as in each iteration we need $O(|X|)$ in Line~\ref{lDisjObjBuchiMDP:if} 
    and constant time in Line~\ref{lDisjObjBuchiMDP:addec}.
    Thus in total the algorithm takes 
    $O(\textsc{MEC} + n + b) = O(\textsc{MEC} + b )$ time.
\end{proof}

When it comes to disjunctive Büchi queries with $k$ sets $\target_i$,
one basically solves $k$ 
Büchi problems and then computes disjunctive almost-sure reachability 
queries of the winning sets of the Büchi problems.
However, as the MEC-decomposition is independent of the sets $\target_i$,
is suffices to compute the MEC-decomposition once. 
This results in an $O(k\cdot m+\textsc{MEC} + b) = O(k\cdot m+\textsc{MEC})$ 
time algorithm (see Algorithm~\ref{alg:DisjQueryBuchiMDP}).

\begin{algorithm}
	\SetAlgoRefName{DisjQueryBüchiMDP}
	\caption{Algorithm for Disjunctive Büchi Queries on MDPs}
	\label{alg:DisjQueryBuchiMDP}
	\SetKwInOut{Input}{Input}
	\SetKwInOut{Output}{Output}
	\BlankLine
	\Input{
	  MDP $\mdp = ((V, E), (\vo, \vr), \trans)$ and Büchi objectives  $\target_i$ for  $1 \le i \le k$
	}
	\Output
	{
	  $\bigvee_{1 \le i \le k} \as{\mdp, \objsty{Büchi}{\target_i}}$
	}
	\BlankLine
	$\mathcal{\ec} \gets \mecalg(\mdp)$\;
	\For{$1 \leq i \leq k$}{
	  $\winning_i \gets \emptyset$\;
	}
	
	\ForEach{$\ec \in \mathcal{\ec}$}{
		\For{$1 \leq i \leq k$}{
		    \If{$\target_i \cap  \ec \ne \emptyset$}{
			  $\winning_i \gets \winning_i \cup \{\ec\}$\;
		    }	    
		}
	}
	\Return{$\bigvee_{1 \leq i \leq k}\as{\mdp, \reacht{\bigcup_{\ec \in \winning_i} \ec}}$\;}
\end{algorithm}

\subsection{Algorithms for MDPs with coBüchi Objectives}
\label{sec:cobuchialg}

Again, as coBüchi objectives can be encoded as Rabin pairs, one can use
Algorithm~\ref{alg:rabinbasic} to compute the a.s.~winning set for disjunctive 
coBüchi objectives. 
However, coBüchi objectives allow for some simplifications that result in 
the simpler and more efficient Algorithm~\ref{alg:DisjObjCoBuchiMDP}.
This simplifications are based on the observation that for coBüchi all sets $L_i$ coincide with the set of all vertices and 
therefore the if conditions in Lines~\ref{lRabinMDPbasic:if1} \& \ref{lRabinMDPbasic:if2} of Algorithm~\ref{alg:rabinbasic} are always true.
That is, whenever there is a vertex in a MEC~$\ec$ of $\mdp$ that is not contained 
in $\at(\mdp[\ec],\target_i)$, then there is a MEC in $\mdp[\ec \setminus 
\at(\mdp[\ec],\target_i)]$, which is a good end-component of $\mdp$.
Testing whether a MEC contains a good EC for a coBüchi objective 
$\objsty{coB{\"u}chi}{\target_i}$ thus reduces to testing 
whether the random attractor of $\target_i$ covers the whole MEC. 

\begin{observation}\label{obs:disjStreett}
The same ideas can be used for the disjunction of one-pair Streett objectives
\upbr{Table~\ref{tab:condis}}. For each MEC~$\ec$ and each $i$ we check whether
$\ec \cap L_i \ne \emptyset$ and $\ec \cap U_i = \emptyset$. If this is the case,
then we determine whether the random attractor of $L_i$ covers the whole MEC.
If not, then the MEC contains a good end-component for the one-pair Streett 
objective.
\end{observation}

\begin{algorithm}
	\SetAlgoRefName{DisjObjCoBüchiMDP}
	\caption{Algorithm for MDPs with Disjunctive coBüchi Objectives}
	\label{alg:DisjObjCoBuchiMDP}
	\SetKwInOut{Input}{Input}
	\SetKwInOut{Output}{Output}
	\BlankLine
	\Input{
	  MDP $\mdp = ((V, E), (\vo, \vr), \trans)$ and coBüchi objectives  $\target_i$ for  $1 \le i \le k$
	}
	\Output
	{
	  $\as{\mdp, \bigvee_{1 \le i \le k}  \objsty{coBüchi}{\target_i}}$
	}
	\BlankLine
	$\mathcal{\ec} \gets \mecalg(\mdp)$;
	$\winning \gets \emptyset$\;
	
	\ForEach{$\ec \in \mathcal{\ec}$}{
		\For{$1 \leq i \leq k$}{
		    \If{$\ec \not\subseteq \at(\mdp[{\ec}], \target_i)$ \label{lDisjObjCoBuchiMDP:attractor}}{
				  $\winning \gets \winning \cup \{\ec\}$ \label{lDisjObjCoBuchiMDP:addec}\;
				  continue with next $\ec \in \mathcal{\ec}$\;
		    }
		}
	}
	\Return{$\as{\mdp, \reacht{\bigcup_{\ec \in \winning} \ec}}$\;}
\end{algorithm}

\begin{proposition}[Runtime] 
	Algorithm~\ref{alg:DisjObjCoBuchiMDP} can be implemented in $O(k\cdot m + \textsc{MEC})$ time.
\end{proposition}
\begin{proof}
    The initialization of $\mathcal{\ec}$ with all MECs of the input
    MDP $\mdp$ can clearly be done in $O(\textsc{MEC})$ time. Further by 
    Theorem~\ref{th:timedrmdp} the final almost-sure reachability computation
    can be done in $O(\textsc{MEC})$ time. 
    Consider an iteration of the for-each loop, 
    let $\ec$ denote the considered MEC, and fix some iteration $i$ of the for loop. 
    Let $m_\ec$ be the number of edges in $\mdp[\ec]$.
    In the $i$th iteration we need $O(|m_\ec|)$ time
    to compute the random attractor in Line~\ref{lDisjObjCoBuchiMDP:attractor}
    and constant time in Line~\ref{lDisjObjCoBuchiMDP:addec}.
    Thus the total time is $O(k\cdot m +\textsc{MEC})$.
\end{proof}

When it comes to disjunctive coBüchi queries with $k$ sets $\target_i$, 
we have to remember which of the sets $\target_i$ are satisfied by a MEC and 
then compute disjunctive almost-sure reachability queries,
one query per set $\target_i$.
This increases the running time for the almost-sure reachability computation
to $O(k\cdot m)$ (given the MEC-decomposition),
which, however, is subsumed by the total running time of $O(k\cdot m +\textsc{MEC})$.
The resulting algorithm is stated as Algorithm~\ref{alg:DisjQueryCoBuchiMDP}.

\begin{algorithm}
	\SetAlgoRefName{DisjQueryCoBüchiMDP}
	\caption{Algorithm for Disjunctive coBüchi Queries on MDPs}
	\label{alg:DisjQueryCoBuchiMDP}
	\SetKwInOut{Input}{Input}
	\SetKwInOut{Output}{Output}
	\BlankLine
	\Input{
	  MDP $\mdp = ((V, E), (\vo, \vr), \trans)$ and coBüchi objectives $\target_i$ for  $1 \le i \le k$
	}
	\Output
	{
	  $\bigvee_{1 \le i \le k} \as{\mdp, \objsty{coBüchi}{\target_i}}$
	}
	\BlankLine
	$\mathcal{\ec} \gets \mecalg(\mdp)$\;
	\For{$1 \leq i \leq k$}{
	  $\winning_i \gets \emptyset$\;
	}
	
	\ForEach{$\ec \in \mathcal{\ec}$}{
		\For{$1 \leq i \leq k$}{
		    \If{$\ec \not\subseteq \at(\mdp[{\ec}], \target_i)$}{
				  $\winning_i \gets \winning_i \cup \{\ec\}$\;
		    }
		}
	}
	\Return{$\bigvee_{1 \leq i \leq k}\as{\mdp, \reacht{\bigcup_{\ec \in \winning_i} \ec}}$\;}
\end{algorithm}

\section{Algorithm for Graphs with Singleton coBüchi Objectives}\label{sec:singleton}
In this section we show how to compute  in linear time the winning set
for graphs with a special type of coBüchi objectives, namely when all sets $\target_i$
for $1 \le i \le k$ have cardinality one. 
\begin{theorem}
	Given a graph $G = (V, E)$ and coBüchi objectives $\target_i$ with 
	$\lvert \target_i \rvert = 1$ for $1 \le i \le k$, the winning set for the 
	disjunction over the coBüchi objectives can be computed in $O(m)$ time.
\end{theorem}
To compute the winning set it is sufficient to detect whether a strongly connected
graph contains a cycle that does \emph{not contain all} the vertices 
in the set $\target = \bigcup_{1 \le i \le k}\target_i$.
To see this, first note that each non-trivial SCC of the graph (i.e., each 
SCC that contains at least one edge) that does not contain all vertices of $\target$
is winning. If there is no SCC~$S$ with $\target \subseteq S$, then
we can determine the winning set in linear time by computing
the vertices that can reach any non-trivial SCC. Thus it remains to consider
an SCC $S$ with $\target \subseteq S$. For the relevant case of $\lvert \target 
\rvert > 1$ we have that $S$ is a non-trivial SCC.
Since $S$ is strongly connected, 
the vertices of $S$ can reach each other and hence it is sufficient to compute whether
$S$ contains a cycle that does not contain all the vertices of $\target$ (i.e.\ 
solving the \emph{non-emptiness} problem).
If such a cycle exists, then also $S$ is winning, otherwise $S$ is not winning.
In any case, the winning set can then be determined by computing the vertices
that can reach some winning SCC.

\begin{algorithm}
	\SetAlgoRefName{SingletonCBGraph}
	\caption{Disjunctive Singleton coBüchi on Graphs}
	\label{alg:SingletonCoBuchiGraph}
	\SetKwInOut{Input}{Input}
	\SetKwInOut{Output}{Output}
	\BlankLine
	\Input{
	  strongly connected graph $G = (V, E)$ and coBüchi objectives $\target_i$ 
	  with $\lvert \target_i \rvert = 1$ for $1 \le i \le k$ and $k > 1$, 
	  let $\target = \bigcup_i \target_i$
	}
	\Output
	{
	 ``yes'' if there is a cycle $C$ with $T \not\subseteq C$; ``no'' otherwise
	}
	\BlankLine
	$\mathcal{S} \gets \allsccalg(G[V \setminus \target_1])$\;
	\If{$\mathcal{S}$ contains non-trivial SCC}{
		\Return{yes}\;
	}\Else{
		let $s$ be the vertex in $\target_1$\;
		replace $s$ with $s_{\text{in}}$ and $s_{\text{out}}$:
		$s_{\text{in}}$ gets in-edges and $s_{\text{out}}$ gets out-edges of $s$\;
		$Q_0 \gets \set{s_{\text{out}}}$; mark $s_{\text{out}}$\;
		\For{$j \gets 0$ \KwTo $k-1$}{
			 $Q_{j+1} \gets \emptyset$\;
			\While{$Q_j \ne \emptyset$}{
				remove $v$ from $Q_j$\;
				\If{$v = s_{\text{in}}$}{\Return{yes}\;}
				\ForEach{$(v, w) \in E$ with $w$ not marked}{
					mark $w$\;
					\lIf{$w \in \target$}{add $w$ to $Q_{j+1}$}
					\lElse{add $w$ to $Q_j$}
				}
			}
		}
		\Return{no}
	}
\end{algorithm}

We now describe the algorithm to determine whether a strongly connected
graph $G = (V, E)$ contains a simple cycle~$C$ such that we have 
$\target_i \cap C = \emptyset$ for some $1 \le i \le k$, given $\lvert \target_i 
\rvert = 1$ for all~$i$.
First we check whether $G[V \setminus \target_1]$ contains a non-trivial SCC. 
If this is true, then
$G$ contains a cycle that does not contain $\target_1$ and we are done. Otherwise
every cycle of~$G$ contains $\target_1$. We assign the edges of~$G$ edge lengths
as follows: All edges $(v, w) \in E$ for which $w \in \target$ have length~1,
all other edges have length~0. 
Let $s$ denote the vertex in $\target_1$. 
Let $\delta$ be the length of the shortest path (w.r.t.\ the edge lengths 
defined above) from $s$ to $s$ that uses at least one edge, i.e., the minimum 
length of a cycle containing $s$. We have that $\delta < k$ if and only if 
this cycle with the length $\delta$ does not contain all vertices of $\target$.
Thus if $\delta < k$, then $G$ is winning for the coBüchi objective, otherwise not.
Note that this algorithm would also work for a Rabin objective where we have
for each $1 \le i \le k$ that \upbr{a} $L_i = \set{s}$ for some $s \in V$ and 
\upbr{b} $\lvert U_i \rvert = 1$.

Since all edge lengths are zero or one, we can compute $\delta$ in linear time. 
In Algorithm~\ref{alg:SingletonCoBuchiGraph} we additionally use that all 
incoming edges of a vertex have the same length. After checking whether
$G[V \setminus \target_1]$ contains a non-trivial SCC, the algorithm works as
follows. 
We modify the graph by replacing the vertex~$s$ by two vertices, $s_{\text{in}}$ 
and $s_{\text{out}}$, and replacing $s$ in all edges $(v, s) \in E$ with 
$s_{\text{in}}$ and in all edges $(s, v) \in E$ with $s_{\text{out}}$.
Then $\delta$ is equal to the shortest path from $s_{\text{out}}$ to 
$s_{\text{in}}$. For the algorithm we consider both $s_{\text{in}}$ and 
$s_{\text{out}}$ to be contained in $\target$.
In the $j$th iteration of the for-loop we consider two ``queues'',
$Q_j$ and $Q_{j+1}$ (can be implemented as sets). Each vertex is added to a 
queue at most once during the 
algorithm, which is ensured by marking vertices when they are added to a queue
and only add before unmarked vertices. 
The following lemma shows that, until the vertex $s_{\text{in}}$ is removed from 
$Q_j$ and the algorithm terminates,
precisely the vertices with distance~$j$ from $s_{\text{out}}$ are 
added to $Q_j$ for each $j$. Thus $s_{\text{in}}$ is added to $Q_j$ for some $j < k$ if and 
only if $\delta < k$, which shows the correctness of the algorithm.
The runtime of the algorithm is $O(m)$ because each vertex is added to and 
removed from a queue at most once and thus the outgoing edges of a vertex are
only considered once, namely when it is removed from a queue.
\begin{lemma}
 Before each iteration $j$ of the for-loop in 
 Algorithm~\ref{alg:SingletonCoBuchiGraph}, $Q_j$ contains
 the vertices of $\target$ with distance $j$ from $s_{\text{out}}$.
 During iteration $j$, the vertices of $V \setminus \target$
 with distance $j$ from $s_{\text{out}}$ are added to $Q_j$. 
 No other vertices are added to~$Q_j$.
\end{lemma}
\begin{proof}
The proof is by induction over the iterations of the for-loop.
Before the first iteration ($j=0$), $Q_0$ is initialized with $s_{\text{out}}$
and all queues $Q_j$ for $j > 0$ are empty, 
thus the induction base holds. Assume the claim holds before the $j$th 
iteration.
At the end of the while-loop, $Q_j$ is empty; every vertex $v$ that was added 
to $Q_j$ before or in the $j$th iteration of the for-loop is removed from $Q_j$
in some iteration of the while-loop. Then all the unmarked vertices $w$ with 
$(v, w) \in E$ are marked and added to $Q_j$ if the edge $(v, w)$ has length 
zero or added to $Q_{j+1}$ if the edge $(v, w)$ has length one.
A vertex $u \in V \setminus \target$ with distance
at least~$j$ from $s_{\text{out}}$ has distance exactly $j$ if and only if 
it can be reached from some vertex $v \in \target$ that has distance $j$
by a sequence of zero length edges. The while-loop precisely adds these vertices
to $Q_j$. Further, a vertex $u \in V \cap \target$ has distance $j+1$ if and 
only if it has an edge from some vertex $v \in V$ that has distance $j$.
The while-loop adds exactly these vertices to~$Q_{j+1}$.
\end{proof}

\section{Conclusion}\label{sec:conclusion}
In this work we present improved algorithms and the first conditional
super-linear lower bounds for several fundamental model-checking 
problems in graphs and MDPs w.r.t.\ to $\omega$-regular objectives. 
Our results establish the first model separation results for graphs
and MDPs w.r.t.\ to classical $\omega$-regular objectives, and first 
objective separation results both in graphs and MDPs for dual objectives,
and conjunction and disjunction of same objectives.
An interesting direction of future work is to consider similar results
for other models, such as, games on graphs.

\section*{Acknowledgments.}
K.~C.\ and M.~H.\ are supported by the Austrian Science Fund (FWF): P23499-N23.
K.~C.\ is supported by S11407-N23 (RiSE/SHiNE), an ERC Start Grant (279307: Graph Games), and a Microsoft Faculty Fellows Award.
For W.~D., M.~H., and V.~L.\ the research leading to these results has received funding from the 
European Research Council under the European Union's Seventh Framework Programme 
(FP/2007-2013) / ERC Grant Agreement no. 340506.
\bibliographystyle{plain}

\end{document}